\documentclass[journal]{IEEEtran}
\usepackage{epsfig}
\usepackage{cite}
\usepackage{graphicx}
\graphicspath{{Figures/}}
\usepackage{url}
\usepackage{stfloats}
\usepackage{amsmath}
\usepackage{amsfonts}   % to get mathbb for the reel symbol
\usepackage{amsopn}     % to get \DeclareMathOperator
\usepackage{psfrag}
\usepackage[tight]{subfigure}
\usepackage{mdwtab}     %ERW: added to get pretty tables
\usepackage{mdwlist} 
\usepackage{enumerate}
\usepackage{color}
\usepackage{amsmath}
\usepackage{amsthm}
\usepackage{epsfig}
\usepackage{cite}
\usepackage{graphicx}
\graphicspath{{Figures/}}
\usepackage{subfigure}
\usepackage{url}
\usepackage{stfloats}
\usepackage{amsmath}
\usepackage{amsfonts} % to get mathbb for the reel symbol
\usepackage{amsopn}    % to get \DeclareMathOperator
\usepackage{psfrag}
\usepackage{xcolor}
\usepackage{subfigure}
\usepackage{mdwtab} %ERW: added to get pretty tables
\usepackage{epstopdf}
\usepackage[normalem]{ulem}

\newtheorem{theorem}{\bf Theorem}
\newtheorem{rem}{\bf Remark}
\newtheorem{corollary}{\bf Corollary}
\newtheorem{lemma}{\bf Lemma}
\newtheorem{definition}{\bf Definition}
\newtheorem{prop}{\bf Proposition}
\ifCLASSINFOpdf
\else
\fi
\begin{document}
\title{Input-to-State Stability of Periodic Orbits of Systems with Impulse Effects via Poincar\'e Analysis}
%\title{Poincar\'e Analysis for Hybrid Periodic Orbits of Systems with Impulse Effects under External Inputs}

\author{Sushant~Veer, Rakesh, and~Ioannis~Poulakakis% <-this % stops a space
        
%\author{Sushant~Veer,~\IEEEmembership{Student Member,~IEEE,}
%        Rakesh, and~Ioannis~Poulakakis,~\IEEEmembership{Member,~IEEE}% <-this % stops a space
        
\thanks{S. Veer and I. Poulakakis are with the Department
of Mechanical Engineering and Rakesh is with the Department of Mathematical Sciences, University of Delaware, Newark,
DE, 19716 USA e-mail: {\tt\small \{veer, rakesh, poulakas\}@udel.edu.}}
\thanks{This work is supported in part by NSF CAREER Award IIS-1350721 and by NRI-1327614.}% <-this % stops a space
}

\maketitle

\begin{abstract}
In this paper we investigate the relation between robustness of periodic orbits exhibited by systems with impulse effects and robustness of their corresponding Poincar\'e maps. In particular, we prove that input-to-state stability (ISS) of a periodic orbit under external excitation in both continuous and discrete time is equivalent to ISS of the corresponding 0-input fixed point of the associated \emph{forced} Poincar\'e map. This result extends the classical Poincar\'e analysis for asymptotic stability of periodic solutions to establish orbital input-to-state stability of such solutions under external excitation. In our proof, we define the forced Poincar\'e map, and use it to construct ISS estimates for the periodic orbit in terms of ISS estimates of this map under mild assumptions on the input signals. As a consequence of the availability of these estimates, the equivalence between exponential stability (ES) of the fixed point of the 0-input (unforced) Poincar\'e map and ES of the corresponding orbit is recovered. The results can be applied naturally to study the robustness of periodic orbits of continuous-time systems as well. Although our motivation for extending classical Poincar\'e analysis to address ISS stems from the need to design robust controllers for limit-cycle walking and running robots, the results are applicable to a much broader class of systems that exhibit periodic solutions.  
\end{abstract}


% \fi
\IEEEpeerreviewmaketitle

%==============================================================
%==============================================================
\section{Introduction}
%==============================================================
%==============================================================

\IEEEPARstart{S}{ystems} with impulse effects (SIEs) are characterized by a set of ordinary differential equations (ODEs) and a discrete map that reinitializes the ODEs when the corresponding solution reaches a switching surface, possibly resulting in  discontinuous evolution. These systems arise in a broad range of fields; a non-exhaustive list of examples includes impact mechanics~\cite{bauinov1989systems}, modeling of population dynamics~\cite{ballinger1997permanence}, communication~\cite{yang1997impulsive}, and legged robotics~\cite{westervelt2007feedback}; a  collection of methods for analyzing SIEs can be found in~\cite{haddad2006impulsive}.  

In this paper, we study the stability properties of limit cycles exhibited by SIEs under external excitation. Our interest in this specific class of systems arises from dynamically-stable legged robots, where periodic walking gaits are modeled as limit cycles of SIEs. This approach has been successful in generating asymptotically stable periodic gaits for bipedal robots through a variety of methods, including hybrid zero dynamics~\cite{westervelt2003hybrid, ames2014rapidly}, geometric control~\cite{Gregg2010Geometric}, virtual holonomic constraints~\cite{freidovich2009passive}, to name a few. Recent extensions of these methods resulted in generating continuums of limit-cycle gaits for bipedal walkers~\cite{razavi2016symmetry, veer2017continuum}, and switching among them~\cite{saglam2013switching, motahar2016composing, veer2017supervisory}, to enlarge the behavioral repertoire of these robots in order to accomplish tasks that require adaptability to typical human-centric environments~\cite{veer2017driftless}, and human (or robot) collaborators~\cite{motahar2017steering}. Practical use of these robots demands robustness to external disturbances, which has led many researchers---including the authors of the present paper---to analyze~\cite{veer2015adaptation, veer2016localISS, kolathaya2016parameter} and design~\cite{Nguyen2015RSS, veer2017supervisory, hamed2016exponentially} controllers that enhance the robustness of limit-cycle walking gaits. With this being our motivation, we develop in this paper a framework  for rigorously analyzing the robustness of limit cycles, by relating orbital input-to-state stability (ISS) for hybrid limit cycles of SIEs with the corresponding Poincar\'e map.

The notion of ISS has been widely used to study robustness of equilibrium points in continuous ~\cite{sontag1996new}, discrete~\cite{jiang2001input}, and hybrid~\cite{cai2009characterizations} systems. Intuitively, the solutions  emanating in a neighborhood of an ISS equilibrium point remain bounded when the external inputs are bounded. In addition, when the inputs vanish, these solutions converge back to the equilibrium. Beyond equilibrium points, ISS can be naturally applied to study robustness of zero-invariant sets by considering the point-to-set distance~\cite{sontag1996new,sontag2001iss}. Establishing ISS in this context poses a considerable challenge, which, in the case of SIEs, is exacerbated by the hybrid nature of the system. However, for periodic orbits---such as those of interest in this paper---we show here that the problem can be reduced to studying ISS of an unforced (0-input) fixed point of a discrete dynamical system, thus avoiding direct analysis of hybrid solutions. This discrete system  arises through the Poincar\'e map construction suitably extended to incorporate external inputs, thereby  resulting in the definition of a \emph{forced Poincar\'e map.}
%This map is induced by the usual definition of the Poincar\'e map---see~\cite{GUHO96} for example---by suitably incorporating the effects of external inputs in continuous and discrete time.
%To the best of the authors' knowledge, establishing the connection between the ISS of a 0-input fixed point of the forced Poincar\'e map and the ISS of the underlying  periodic orbit has not been addressed in the relevant literature.} 

Numerous results exist that analyze forced Poincar\'e maps of systems evolving under the influence of external inputs; in the context of SIEs, examples include~\cite{kolathaya2016parameter, veer2016localISS, martin2017stable}, in which  the input signals are not necessarily periodic. However, the exact relation of conclusions deduced on the basis of the Poincar\'e map to properties of the underlying periodic orbit has \emph{not} been explicitly discussed in the relevant literature. Indeed, rigorous results that relate the stability properties of the Poincar\'e map with those of the corresponding periodic orbit are restricted to systems \emph{without} inputs; e.g., \cite[Theorem~6.4]{khalil1992noninear} addresses local asymptotic stability (LAS) of periodic orbits in continuous systems, while \cite[Theorem~1]{grizzle2001asymptotically}, \cite[Theorem~13.1]{haddad2006impulsive} address LAS and \cite[Theorem~1]{morris2005restricted} local exponential stability (LES) of such orbits in SIEs. The relation between the behavior of periodic orbits of SIEs \emph{under external inputs} and the corresponding \emph{forced} Poincar\'e map is at the core of this paper.
%acting in both the continuous and discrete parts of a SIE
%These papers indicate the significance of  understanding the behavior of periodic solutions in the presence of external inputs.
%cite{foale1992dynamical, glass1984global} for periodically excited systems and
%\emph{without} considering their effect on the limit cycle
%Numerous results exist that analyze the Poincar\'e map under external inputs \emph{without} considering their effect on the limit cycle; see \cite{foale1992dynamical, glass1984global} for periodic inputs and \cite{kolathaya2016parameter, veer2016localISS, martin2017stable} for inputs that are not necessarily periodic. On the other hand, results that relate the stability properties of a fixed point of a Poincar\'e map and the corresponding limit cycle are restricted to unforced systems; see \cite[Theorem~6.4]{khalil1992noninear} for continuous systems and \cite[Theorem~1]{grizzle2001asymptotically}, \cite[Theorem~13.1]{haddad2006impulsive}, \cite[Theorem~1]{morris2005restricted} for systems with impulse effects. This gap in the relevant literature calls for the development of new tools that enable limit-cycle robustness analysis through the study of the corresponding forced Poincar\'e map. Offering such tools is at the core of this paper. 

Specifically, the main contribution of this work (Theorem~\ref{thm:LISS-equivalence}) is that ISS of a limit cycle  exhibited by a SIE is equivalent to ISS of a 0-input fixed point of the corresponding forced Poincar\'e map. This result significantly simplifies analysis, as it replaces the problem of establishing ISS of a hybrid limit cycle with the simpler problem of checking asymptotic stability of a 0-input fixed point of a discrete dynamical system defined by the corresponding Poincar\'e map (Theorem~\ref{thm:LAS-LISS}). To ensure the level of generality required by practical applications, we consider inputs affecting both the continuous and the discrete dynamics of the system. The continuous-time inputs belong in the (Banach) space of continuous bounded functions under the supremum norm. The resulting forced Poincar\'e map is a nonlinear \emph{functional} defined over an infinite-dimensional function space, thus significantly extending prior work that considers finite dimensional disturbances; see \cite{veer2016localISS, kolathaya2016parameter, fradkov1998introduction} for example. Finally, the proof of the main result provides an explicit connection between ISS estimates of the forced Poincar\'e map and those of the hybrid orbit. 

The results presented in this paper generalize previous contributions such as \cite[Theorem~1]{morris2005restricted}, which is widely used to establish exponential stability (ES) of a hybrid limit cycle when the  fixed point of the corresponding Poincar\'e map is ES. Indeed, \cite[Theorem~1]{morris2005restricted} can be obtained as a consequence of Theorem~\ref{thm:LISS-equivalence} of Section~\ref{sec:main-results} below. Furthermore, Proposition~\ref{prop:equiv-norm} of Section~\ref{sec:main-results} and Lemma~\ref{lem:solution-compare-S+} of Section~\ref{sec:thm-proof} complete crucial arguments that were omitted in the proof of~\cite[Theorem~1]{morris2005restricted}. Moreover, our results can offer useful tools for  the design of robust controllers for limit cycles of SIEs. For example, the methods in~\cite{hamed2017decentralized} that are based on Poincar\'e map analysis can be supported using Theorem~\ref{thm:LISS-equivalence}. As a final note, the results of this paper can naturally be applied to study ISS of limit cycles of continuous-time nonlinear systems under external excitation. Hence, their relevance extends to other bioinspired robots---including aerial robots with flapping wings~\cite{ramezani2017describing} and robot snakes~\cite{liljeback2011controllability}---which, like legged robots, realize locomotion through periodic forceful interactions with their environment.

%space of \textcolor{red}{continuous functions}, which, equipped with the supremum norm, obtains the structure of a
%The paper is organized as follows. Section~\ref{sec:background} introduces the class of systems under study and develops the forced Poincar\'e map; Section~\ref{sec:main-results} presents the main results of the paper; Section~\ref{sec:prop-proof} and \ref{sec:thm-proof} present the proofs for Proposition~\ref{prop:equiv-norm} and Theorem~\ref{thm:LISS-equivalence}, respectively; and Section~\ref{sec:conclusions} provides conclusions. 

%==============================================================
%==============================================================
\section{Background}
\label{sec:background}
%==============================================================
%==============================================================

This section introduces the class of systems with impulse effects pertinent to this paper, and develops a forced Poincar\'e map suitable for studying periodic orbits of such systems under the influence of continuous and discrete exogenous inputs; such inputs could be command or disturbance signals. We begin with a few notes on the notation used in the paper.

\subsection{Notation}
\label{subsec:notation}

Let $\mathbb{R}$ and $\mathbb{Z}$ denote the sets of real and integer numbers, and $\mathbb{R}_+$ and $\mathbb{Z}_+$ the corresponding subsets that include the non-negative reals and integers, respectively. For any $x\in\mathbb{R}^n$, the Euclidean norm is represented as $\|x\|$. An open ball of radius $\delta>0$ centered at $x$ is denoted by $B_\delta(x)$. The point-to-set distance of $x$ from $\mathcal{A}\subseteq \mathbb{R}^n$ is defined as $\mathrm{dist}(x,\mathcal{A}):= \inf_{y\in \mathcal{A}} \|x-y\|$. We use $\mathcal{P}(\mathcal{A})$ to represent the power set of $\mathcal{A}$, and $\mathcal{A}^{\rm c}$ to denote the complement of $\mathcal{A}$ with respect to $\mathbb{R}^n$.

For any interval $E\subseteq\mathbb{R}$ let $u:E \to \mathbb{R}^p$ be a function that represents the continuous-time inputs. The norm of $u$ is defined as $\|u\|_\infty:= \sup_{t\in E}\|u(t)\|$. The set of continuous-time inputs we work with belongs to $\mathcal{U}:=\{u:E \to \mathbb{R}^p~|~u~\mathrm{is}~\mathrm{continuous},~\|u\|_\infty\in\mathbb{R}_+ \}$. Discrete-time inputs $\bar{v} : \mathbb{Z}_+ \to \mathbb{R}^q$ correspond to sequences $\bar{v}=\{v_k\}_{k=0}^\infty$ with $v_k\in\mathbb{R}^q$ for $k\in\mathbb{Z}_+$. The norm of $\bar{v}$ is defined as $\|\bar{v}\|_\infty:=\sup_{k\in\mathbb{Z}_+} \|v_k\|$. The discrete inputs belong to $\mathcal{V}:=\{ \bar{v} : \mathbb{Z}_+ \to \mathbb{R}^q ~|~ \|\bar{v}\|_\infty \in \mathbb{R}_+ \} $. With an abuse of notation we use $\|\cdot\|_\infty$ to denote the norm for both $\mathcal{U}$ and $\mathcal{V}$. No ambiguity arises because the meaning of $\|\cdot\|_\infty$ depends on whether the argument is continuous or discrete. 

A function $\alpha:\mathbb{R}_+ \to \mathbb{R}_+$ belongs to class $\mathcal{K}$ if it is continuous, strictly increasing, and $\alpha(0)=0$. A function $\beta:\mathbb{R}_+ \times \mathbb{R}_+ \to \mathbb{R}_+$ belongs to class $\mathcal{KL}$ if it is continuous, $\beta(\cdot,t)$ belongs to $\mathcal{K}$ for any fixed $t\geq 0$, $\beta(s,\cdot)$ is strictly decreasing, and $\lim_{t\to \infty} \beta(s,t)=0$, for any fixed $s\geq 0$.

%==============================================================
\subsection{Forced Systems With Impulse Effects}
\label{subsec:system}
%==============================================================

We are interested in studying the stability of periodic orbits exhibited by systems with impulse effects under externally applied inputs. These systems are characterized by alternating continuous and discrete phases. The evolution of the state $x\in\mathbb{R}^n$ during the continuous phase is governed by an ODE
%eq
\begin{equation}\label{eq:cont-dyn}
\dot{x}(t) = f(x(t),u(t)) \enspace,
\end{equation}
%eq
where the input $u:\mathbb{R_+}\to \mathbb{R}^p$ is an element of $\mathcal{U}$ defined in Section~\ref{subsec:notation} and $u(t) \in \mathbb{R}^p$ is its value. The vector field $f$ in the right-hand side of \eqref{eq:cont-dyn} satisfies the following assumption:
\begin{enumerate}[{{A}.1)}] % Broken enumerations uses mdwlist package
\item \label{ass:f-c2} $f:\mathbb{R}^n\times\mathbb{R}^p \to \mathbb{R}^n$ is twice continuously differentiable\footnote{Whenever we state that a function is $k$-times continuously differentiable, it applies to all its arguments.}.
\suspend{enumerate}

Local existence and uniqueness of solutions of \eqref{eq:cont-dyn} for a fixed $u$ follows from \cite[Theorem~3.1]{khalil2002nonlinear} based on assumption A.\ref{ass:f-c2} and the continuity of $u$ as a function of $t$. We denote the flow of \eqref{eq:cont-dyn} starting from the initial state $x(0)$ and evolving under the influence of the input $u$ by $\varphi(t,x(0),u)$. 

The continuous phase terminates when the flow of \eqref{eq:cont-dyn} reaches a set $\mathcal{S}\subset \mathbb{R}^n$ defined as 
%eq
\begin{equation}\label{eq:S}
\mathcal{S}:= \{x\in\mathbb{R}^n~|~H(x)=0\} \enspace,
\end{equation}
%eq
where it is assumed that
\resume{enumerate}[{[{{A}.1)}]}]
\item \label{ass:S}$\mathcal{S} \!\neq\! \emptyset$,~ $H :\mathbb{R}^n \to \mathbb{R}$ is twice continuously differentiable, and for all $\hat{x}\in \mathcal{S}$, $\frac{\partial H}{\partial x}\big|_{\hat{x}} \neq 0$, i.e., $\mathcal{S}$ is a co-dimension $1$ embedded submanifold in $\mathbb{R}^n$; see \cite[p. 431]{GUHO96}.
% \textcolor{blue}{see \cite[Definition~B.1]{westervelt2007feedback}.}
\suspend{enumerate}
For future use we define the sets $\mathcal{S}^+:=\{x\in\mathbb{R}^n~|~H(x)>0\}$ and $\mathcal{S}^-:=\{x\in\mathbb{R}^n~|~H(x)<0\}$.
 
The intersection of the flow of \eqref{eq:cont-dyn} with $\mathcal{S}$ initiates the discrete phase, which is governed by the mapping
%eq
\begin{equation}\label{eq:disc-dyn}
x^+ = \Delta(x^-,v)~~~\mathrm{for}~x^-\in\mathcal{S} \enspace,
\end{equation}
%eq
where $x^-$, $x^+$ are the states right before and after impacting $\mathcal{S}$, respectively, and $v\in\mathbb{R}^q$ is a member of the discrete input $\bar{v}$ that belongs in $\mathcal{V}$ defined in Section~\ref{subsec:notation}. It is assumed that
\resume{enumerate}[{[{{A}.1)}]}]
\item $\Delta:\mathbb{R}^{n}\times \mathbb{R}^q \to \mathbb{R}^n$ is continuously differentiable. \label{ass:delta-C-1} \label{ass:Delta}
%\item 
\suspend{enumerate}

Putting together the continuous and discrete phases \eqref{eq:cont-dyn} and \eqref{eq:disc-dyn}, the \textit{forced system with impulse effects} takes the form
%eq
\begin{eqnarray}
\Sigma:
\begin{cases}
\begin{aligned}
\dot{x}(t) &= f(x(t),u(t)) & \hspace{-1mm}\mathrm{if}~x(t)\notin \mathcal{S}	\\
x^+(t)     &= \Delta(x^-(t),v) &\hspace{-1mm}\mathrm{if}~x^-(t)\in \mathcal{S}
\end{aligned}
\end{cases} \enspace, \label{eq:sys-imp-eff}
\end{eqnarray}
%eq
where $x^-(t) := \lim_{\tau \nearrow t} x(\tau)$ and $x^+(t) := \lim_{\tau \searrow t} x(\tau)$.
% are the left and right limits of a trajectory $x(t)$

At any time instant for which it exists, the solution of \eqref{eq:sys-imp-eff} evolves according to either \eqref{eq:cont-dyn} or \eqref{eq:disc-dyn}. This allows us to represent the hybrid flow of \eqref{eq:sys-imp-eff} as the solution of \eqref{eq:cont-dyn} which, on approaching $\mathcal{S}$, is interrupted by the discrete map \eqref{eq:disc-dyn}. Let $\psi(t,x(0),u,\bar{v})$ denote the flow of \eqref{eq:sys-imp-eff} for some initial state $x(0)$, continuous input $u$, and discrete input $\bar{v}$. Adapting the definition in~\cite[Section III-A]{grizzle2001asymptotically}, $\psi(t,x(0),u,\bar{v})$ as a function of time $t \in [0, t_\mathrm{f})$, $t_\mathrm{f} \in \mathbb{R}_+ \cup \{\infty\}$, satisfies the following:
(i) it is right continuous\footnote{To avoid the state having to take two values at impact, a choice is to be made as to whether the state just before or just after impact---i.e., $x^-$ or $x^+$, respectively---is included in the solution. The former corresponds to left continuity and the latter to right continuity of $\psi$ as a function of time. We assume here that $\psi$ is right continuous with respect to $t$; note however that the results that follow hold regardless of this choice~\cite{haddad2006impulsive}.} on $[0, t_\mathrm{f})$;
(ii) left limits exist at each point in $(0, t_\mathrm{f})$; and
(iii) there exists a discrete subset $\mathcal{T} \subset [0, t_\mathrm{f})$ such that (a) for every $t \notin \mathcal{T}$, $\psi(t,x(0),u,\bar{v})$ satisfies \eqref{eq:cont-dyn} for the input $u$ considered, and (b) for $t \in \mathcal{T}$, $x^-(t) = \lim_{\tau \nearrow t} \psi(\tau,x(0),u,\bar{v}) \in \mathcal{S}$ and $x^+(t) = \lim_{\tau \searrow t} \psi(\tau,x(0),u,\bar{v}) = \Delta(x^-(t), v)$ where $v$ is a member of the sequence $\bar{v}$. 
Note that right continuity implies that at time $t \in \mathcal{T}$ the solution attains the value $x^+(t)$ and not $x^-(t)$; that is, $x(t)=\psi(t,x(0),u,\bar{v})=x^+(t)$. Moreover, Proposition~\ref{prop:global-existence} in Section~\ref{sec:main-results} below ensures that in the neighborhood of distinct locally input-to-state stable periodic orbits---such as those of interest in this work---the solutions of \eqref{eq:sys-imp-eff} exist for arbitrary $t_\mathrm{f}>0$, they do not possess consecutive discrete jumps (beating) and do not exhibit Zeno behavior; see~\cite{haddad2006impulsive, goebel2012hybrid} for definitions.
%\footnote{That is, solutions in which successive resets do not engage the continuous-time part of \eqref{eq:sys-imp-eff}; such phenomenon is called beating in the terminology of~\cite{haddad2006impulsive}.}
%\footnote{That is, solutions in which---loosely speaking---an infinite number of discrete jumps occur in finite time; see \cite{goebel2012hybrid}.} behavior.  

Let $x^*\in\mathcal{S}$ and $T^*\in(0,\infty)$ such that the following hold:
\resume{enumerate}[{[{{A}.1)}]}]
\item $\mathrm{dist}(\Delta(x^*,0),\mathcal{S})>0$, and $\Delta(x^*,0)\in\mathcal{S}^+$. The choice $\Delta(x^*,0)\in \mathcal{S}^+$ does not result in loss of generality; if $\Delta(x^*,0)\in \mathcal{S}^-$ re-define $\mathcal{S}$ with $\hat{H}(x):=-H(x)$.\label{ass:no-zeno}
\item $\varphi(t,\Delta(x^*,0),0)$ exists for all $t\in[0,T^*]$ and $\varphi(T^*,\Delta(x^*,0),0)=x^*$. \label{ass:orbit-exist}
\suspend{enumerate}
Using assumptions A.\ref{ass:no-zeno}-A.\ref{ass:orbit-exist} define
%eq
\begin{equation}\label{eq:orbit-def}
\mathcal{O}:= \{ \varphi(t,\Delta(x^*,0),0)~|~t\in[0,T^*)\} \enspace,
\end{equation}
%eq
and suppose further that $\mathcal{O}$ satisfies the following assumptions:
\resume{enumerate}[{[{{A}.1)}]}]
\item Let $\overline{\mathcal{O}}$ be the closure of $\mathcal{O}$, then $\mathcal{S}\cap\overline{\mathcal{O}} = \{ x^* \}$.\label{ass:x*}
\item $\mathcal{O}$ is transversal to $\mathcal{S}$ at $x^*$, i.e., $L_f H(x^*,0):=\frac{\partial H}{\partial x}\big|_{x^*} f(x^*,0)< 0$. \label{ass:orbit-transversal}
\suspend{enumerate}
\noindent It follows from our assumptions that $\mathcal{O}$ is a bounded unforced $(u\equiv 0,\bar{v}\equiv 0)$ hybrid periodic orbit of $\Sigma$ that exhibits only one impact with $\mathcal{S}$ at $x^*$ and has period $T^*$. Moreover, $\mathcal{O}$ is not a closed curve; see Fig.~\ref{fig:orbit} for a geometric illustration.

%fig
\begin{figure}[t]
\vspace{+0.05in}
\begin{centering}
\includegraphics[width=0.7\columnwidth]{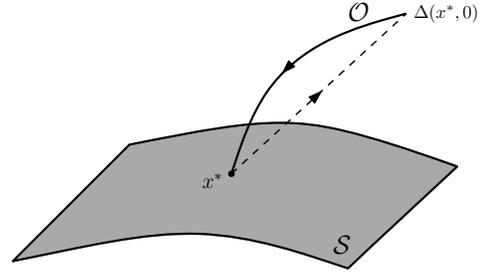} 
\par\end{centering}
\vspace{-0.12in}
\caption{Geometric illustration of $\mathcal{O}$. The switching surface $\mathcal{S}$ is in grey.}
\vspace{-0.15in}
\label{fig:orbit} 
\end{figure}
%fig

%==============================================================
\subsection{Forced Poincar\'e Map}
\label{subsec:poinc}
%==============================================================

The Poincar\'e map is a common tool used for analyzing systems with periodic orbits. Given a Poincar\'e section---which is an embedded submanifold transversal to the orbit---the Poincar\'e map returns consecutive intersections of the system's flow with the Poincar\'e section. Here, we study the map which returns the intersection of the solution of \eqref{eq:sys-imp-eff} with $\mathcal{S}$ under the influence of the external inputs $u$ and $\bar{v}$. Consequently, it is natural to call this map the \emph{forced} Poincar\'e map.

As was mentioned in Section~\ref{subsec:system}, when the input $u$ affecting \eqref{eq:cont-dyn} is a \emph{fixed} signal from $\mathcal{U}$, existence and uniqueness of the solution emanating from\footnote{Without loss of generality, we use $t=0$ as the initial time to avoid confusion with $t_0$, which, in our notation, is the instant of the first impact.} $x(0) \in \mathbb{R}^n$ can be established over an interval $J \subseteq \mathbb{R}_+$ with $0 \in J$, by \cite[Theorem~3.1]{khalil2002nonlinear} applied on the time-varying vector field $\hat{f}(t,x):=f(x,u(t))$. To develop the forced Poincar\'e map, however, we need to compare solutions with different initial conditions \emph{and} different inputs. To do this, it is important to be able to consider the forced solution $\varphi(t,x(0),u)$ of \eqref{eq:cont-dyn} as a mapping from $J \times \mathbb{R}^n \times \mathcal{U}$ to $\mathbb{R}^n$, interpreting $u$ as an infinite-dimensional ``parameter'' residing in the Banach space $(\mathcal{U},\|\cdot\|_\infty)$. We can then analyze variations of $\varphi(t,x(0),u)$ with respect to its arguments, including $u$. The following lemma shows that, over its maximal interval of existence, the solution $\varphi(t,x(0),u)$ of \eqref{eq:cont-dyn} is continuously differentiable in its arguments, with differentiability understood in the Fr\'echet sense~\cite[p. 333]{lang1993real}.
%lem
\begin{lemma}\label{lem:frechet}
Let $f:\mathbb{R}^n\times\mathbb{R}^p\to\mathbb{R}^n$ in \eqref{eq:cont-dyn} be continuously differentiable, and let $u\in\mathcal{U}$ with $\mathcal{U}$ as in Section~\ref{subsec:notation}. Then, the solution $\varphi: J \times\mathbb{R}^n\times\mathcal{U}\to \mathbb{R}^n$ is continuously differentiable in its arguments in the Fr\'echet sense.
%Let $f:\mathbb{R}^n\times\mathbb{R}^p\to\mathbb{R}^n$ in \eqref{eq:cont-dyn} be continuously differentiable. Let $u\in\mathcal{U}$ with $\mathcal{U}$ as in Section~\ref{subsec:notation} and define $F:\mathbb{R}_+\times\mathbb{R}^n\times\mathcal{U} \to \mathbb{R}^n$ as $F(t,x,u):=f(x,G(t,u))$, where $G:\mathbb{R}_+\times\mathcal{U} \to \mathbb{R}^p$ is $G(t,u)=u(t)$. Then, the solution $\varphi: J \times\mathbb{R}^n\times\mathcal{U}\to \mathbb{R}^n$ of $\dot{x} = F(t,x,u)$ is continuously differentiable in its arguments in the Fr\'echet sense.
\end{lemma}
%lem
The proof makes use of Banach calculus \cite{lang1993real} and is presented in Appendix~\ref{app:existence}. We only note here that, for notational convenience, we use the same symbol $u$ to denote both the finite-dimensional values of the input function at given instants and the infinite-dimensional input signal as a function in the Banach space $(\mathcal{U},\|\cdot\|_\infty)$; the distinction will always be clear through the domain of definition of the corresponding map.

Let $T_{\rm I}:\mathcal{S} \times \mathcal{U} \times \mathbb{R}^q \to \mathbb{R}_+\cup \{\infty \}$ be the time-to-impact map defined as 
%eq
\begin{align}\label{eq:time-to-imp}
& T_{\rm I}(x,u,v):= 
\begin{cases}
\inf \{t\geq 0~|~\varphi (t,\Delta(x,v),u) \in\mathcal{S} \}, \\
~~~~~~~~~~~\mathrm{if}~\exists t: \varphi(t,\Delta(x,v),u)\in\mathcal{S} \\
\infty,   ~~~~~~~\mathrm{otherwise} 
\end{cases} 
\end{align}
%eq
Lemma~\ref{lem:TI-cont} below establishes that the time-to-impact function $T_\mathrm{I}$ is well defined and continuously differentiable in $x$, $u$ and $v$. Note that the dependence of $T_{\rm I}$ on $u$ is to be understood with $u$ interpreted as a function in $(\mathcal{U},\|\cdot\|_\infty)$. The proof is based on the implicit mapping theorem \cite[Chapter~XIV,~Theorem~2.1]{lang1993real} and on Lemma~\ref{lem:frechet} above, and is presented in Appendix~\ref{app:existence}.
%\footnote{A map $g:\mathcal{X}\to\mathcal{Y}$ is well-defined on its domain $\mathcal{X}$ if $g(\mathrm{x})$ is defined for each $\mathrm{x}\in\mathcal{X}$ and $g(\mathrm{x})$ is a single element in the range $\mathcal{Y}$}
%, thus enabling the use of the implicit mapping theorem \cite[Chapter~XIV,~Theorem~2.1]{lang1993real}, which is crucial in proving the lemma.
%lem
\begin{lemma}\label{lem:TI-cont}
Consider \eqref{eq:time-to-imp}. Suppose that \eqref{eq:sys-imp-eff} satisfies assumptions A.\ref{ass:f-c2}-A.\ref{ass:orbit-transversal}. Then, there exists a $\delta>0$ such that $T_{\rm I}$ is continuously differentiable for any $x\in B_\delta(x^*)\cap\mathcal{S}$, $u\in\mathcal{U}$ with $\|u\|_\infty<\delta$, and $v\in B_\delta(0)$.
\end{lemma}
%lem

We are now ready to define the forced Poincar\'e map $P\!:\!\mathcal{S} \times \mathcal{U} \times \mathbb{R}^q \to \mathcal{S}$ as 
%eq
\begin{equation} \nonumber%\label{eq:poinc-map}
P(x,u,v):=\varphi(T_{\rm I}(x,u,v),\Delta(x,v),u) \enspace. 
\end{equation}
%eq
From Lemmas~\ref{lem:frechet} and~\ref{lem:TI-cont} it follows that $P$ is well defined and continuously differentiable for any $x\in B_\delta(x^*)\cap\mathcal{S}$, $u\in\mathcal{U}$ with $\|u\|_\infty<\delta$, and $v\in B_\delta(0)$. 

Let $\psi(t,x(0),u,\bar{v})$ be a solution of \eqref{eq:sys-imp-eff} and $x(t)=\psi(t,x(0),u,\bar{v})$ be the value of the state at time $t$. For $k\in\mathbb{Z}_+$, let $t_k$ be the instant at which the $(k+1)$-th ``intersection'' of $x(t)$ with $\mathcal{S}$ occurs. Define $u_k(t):=u(t)$ for $t_k \leq t <t_{k+1}$, and let $v_k$ be the $k$-th element of the sequence $\bar{v}$. Then, the forced Poincar\'e map gives rise to the forced discrete system
%$x(t)=\psi(t,x(0),u,\bar{v})$
%eq
\begin{equation}\label{eq:poinc-dyn}
x_{k+1} = P(x_k,u_k,v_k) \enspace,
\end{equation}
%eq
where $x_k:=\lim_{t \nearrow t_k} x(t)$. The discrete system \eqref{eq:poinc-dyn} captures the evolution of the system from \emph{just before} an impact with $\mathcal{S}$ to \emph{just before} the next impact, assuming that the next impact occurs. It should be emphasized that the state $x(t)$ does \emph{not} attain\footnote{Right continuity of $\psi$ in $t$ implies that, in general, there is no $t$ for which $x(t) \in \mathcal{S}$; see related comments in Section~\ref{subsec:system} and in \cite[Section 4.1.2 ]{westervelt2007feedback}.} the value $x_k$ at $t_k$ because $\psi$ has been assumed right continuous in $t$; in fact,  $x(t_k) = \Delta(x_k, v_k) = \lim_{t \searrow t_k} x(t) \neq \lim_{t \nearrow t_k} x(t) = x_k$. Let $x^*$ be as in assumption~A.\ref{ass:x*}, then $x^*$ is the 0-input fixed point of \eqref{eq:poinc-dyn}, i.e., $x^* = P(x^*,0,0)$.
%, i.e., each $x_k\in\mathcal{S}$ corresponds to the $(k+1)$-th intersection of $x(t)$ and $\mathcal{S}$

For future use, we also define $\hat{T}_{\rm I}:\mathcal{S}^+\times\mathcal{U}\to\mathbb{R}_+\cup\{\infty\}$ as the time-to-impact function for solutions of \eqref{eq:cont-dyn} starting from states in $\mathcal{S}^+$ as
%eq
\begin{align}\label{eq:time-to-imp-S+}
& \hat{T}_{\rm I}(x,u):=  \\ 
&
\begin{cases}
\inf \{t\geq 0~|~\varphi(t,x,u) \in\mathcal{S} \},& \mathrm{if}~\exists t:\varphi(t,x,u)\in\mathcal{S}\\
\infty,  &  \mathrm{otherwise} 
\end{cases}. \nonumber
\end{align}
%eq
It is noted that for any point $w \in \mathcal{O}$ there exists a $\delta>0$ such that $\hat{T}_{\rm I}$ is continuously differentiable for any $x \in B_\delta(w)$ and any $u \in \mathcal{U}$ with $\| u \|_\infty < \delta$. The proof similar to that of Lemma~\ref{lem:TI-cont}, and it is not presented for brevity.

Finally, Remark~\ref{rem:hyb-cont-sol-rel} clarifies the relation between $\psi$, $\varphi$ and the sequence $\{x_k\}_{k=0}^{\infty}$ and Remark~\ref{rem:u_k} indicates that the results are valid even when the input $u$ is a piecewise continuous signal, as long as the points of discontinuity are at $t_k$.
%; this will be important when proving the main result. 
%rem
\begin{rem}\label{rem:hyb-cont-sol-rel}
%\textcolor{red}{The solution $\varphi()$ of \eqref{} and $\psi()$ of \eqref{} are different objects, sometimes we write .This is to be understood in the sense that } 
Let $\psi(t,x(0),u,\bar{v})$ be a solution of \eqref{eq:sys-imp-eff} that exists for all $t\geq 0$ and $x(t)=\psi(t,x(0),u,\bar{v})$ be the value of the state at time $t$. Then, if $t_0 := \hat{T}_{\rm I}(x(0),u)$ is the time instant of the first crossing of $\mathcal{S}$ we have
%eq
\begin{align}\nonumber
x(t) = \varphi(t,x(0),u_{[0,t_0)}),~\mathrm{for}~0\leq t <t_0 \enspace.
\end{align} 
%eq
%Furthermore, if $x_k = \lim_{t \nearrow t_k}x(t)$ $t_k=t_{k-1}+T_\mathrm{I}()$
Moreover, if the sequence $\{x_k\}_{k=0}^{\infty}$ is the solution of \eqref{eq:poinc-dyn} for the initial state $x_0:=\lim_{t \nearrow t_0}x(t)$ and the sequence of input functions $\{u_k\}_{k=0}^\infty$ are as defined above, then for all $k \in \mathbb{Z}_+$
%eq
\begin{align}\nonumber
x(t) = \varphi(t,\Delta(x_k,v_k),u_k),~\mathrm{for}~t_k\leq t <t_{k+1} \enspace,
\end{align}
%eq
where $v_k$ is the $k$-th element of $\bar{v}$ and
%eq
\begin{equation}\nonumber 
 t_{k+1}=t_k+T_\mathrm{I}(x_k, u_k, v_k)=t_0 + \sum^k_{j=0}T_\mathrm{I}(x_j,u_j,v_j)\enspace.
\end{equation}
%eq
\end{rem}
%rem
%rem
\begin{rem}\label{rem:u_k}
The results of the paper hold when $u$ is discontinuous, as long as each $u_k$ is a continuous function in $\mathcal{U}$ that agrees with $u$ over the interval $[t_k, t_{k+1})$. 
%Note that $\|u_k\|_\infty \leq \|u\|_\infty$ for all $k\in\mathbb{Z}_+$.
% and $\|v_k\| \leq \|\bar{v}\|_\infty$
\end{rem}
%rem

%==============================================================
\subsection{Pertinent Stability Definitions}
\label{subsec:defs}
%==============================================================

Notions of orbital stability that will be studied in this paper are introduced here. We begin with local input-to-state stability (LISS) of the periodic orbit.
%def
\begin{definition}\label{def:orbit-LISS} % \begin{subdefinition} {def:orbit-LISS-tmax}
The periodic orbit $\mathcal{O}$ of \eqref{eq:sys-imp-eff} is orbitally LISS if there exists a $\delta>0$, $\alpha_1,\alpha_2\in\mathcal{K}$, and $\beta\in\mathcal{KL}$ such that $x(t)=\psi(t,x(0),u,\bar{v})$ satisfies for all $t\in [0,t_\mathrm{f})$, $t_\mathrm{f} \in \mathbb{R}_+ \cup \{\infty\}$,
%eq
\begin{align}
\mathrm{dist}(x(t),\mathcal{O}) \leq~&\beta(\mathrm{dist}(x(0),\mathcal{O}),t)+\alpha_1(\|u\|_\infty) \nonumber\\ 
&+ \alpha_2(\|\bar{v}\|_\infty) \label{eq:orbit-LISS-ineq}\enspace,
\end{align}
%eq
%%eq
%\begin{align}
%\mathrm{dist}(x(t),\mathcal{O}) \leq &\beta(\mathrm{dist}(x(0),\mathcal{O}),t)+\alpha_1(\|u\|_\infty) + \alpha_2(\|\bar{v}\|_\infty) \nonumber ,
%\end{align}
%%eq
for  any $x(0)\in\mathcal{S}^+$ with $\mathrm{dist}(x(0),\mathcal{O})<\delta$, $u\in\mathcal{U}$ with $\|u\|_\infty<\delta$, and $\bar{v}\in\mathcal{V}$ with $\|\bar{v}\|_\infty<\delta$.
\end{definition}
%def
\noindent Proposition~\ref{prop:global-existence} in Section~\ref{sec:main-results} below asserts  that shrinking $\delta>0$ in Definition~\ref{def:orbit-LISS} guarantees that all ensuing hybrid solutions exist for all time; that is, $t_\mathrm{f}>0$ in Definition~\ref{def:orbit-LISS} can be chosen arbitrarily large. Since we focus on local properties of distinct periodic orbits $\mathcal{O}$ of \eqref{eq:sys-imp-eff}, we work with solutions in a small enough neighborhood of $\mathcal{O}$ that satisfy \eqref{eq:orbit-LISS-ineq} for all $t \geq 0$.
%\noindent \textcolor{blue}{Proposition~\ref{prop:global-existence} in Section~\ref{sec:main-results} below asserts  that shrinking $\delta$ sufficiently in Definition~\ref{def:orbit-LISS} ensures that all ensuing hybrid solutions exist for all $t \geq 0$; i.e., $t_\mathrm{f}$ in Definition~\ref{def:orbit-LISS} can be chosen arbitrarily. Furthermore, these solutions do not possess consecutive discrete jumps\footnote{Solutions in which successive resets do not engage the continuous-time part of \eqref{eq:sys-imp-eff}; such phenomenon is called beating in the terminology of~\cite{haddad2006impulsive}.}  and do not exhibit Zeno\footnote{That is, solutions in which---loosely speaking---an infinite number of discrete jumps occur in finite time; see \cite{goebel2012hybrid}.} behavior.}
%In view of this fact, the following definition of LISS will be used.
%%def
%\begin{subdefinition}\label{def:orbit-LISS}
%The periodic orbit $\mathcal{O}$ of \eqref{eq:sys-imp-eff} is orbitally LISS if there exists a $\delta>0$, $\alpha_1,\alpha_2\in\mathcal{K}$, and $\beta\in\mathcal{KL}$ such that $x(t)=\psi(t,x(0),u,\bar{v})$ satisfies for all $t\geq 0$,
%%eq
%\begin{align}
%\mathrm{dist}(x(t),\mathcal{O}) \leq~&\beta(\mathrm{dist}(x(0),\mathcal{O}),t)+\alpha_1(\|u\|_\infty) \nonumber\\ 
%&+ \alpha_2(\|\bar{v}\|_\infty) \label{eq:orbit-LISS-ineq}\enspace,
%\end{align}
%%eq
%for  any $x(0)\in\mathcal{S}^+$ with $\mathrm{dist}(x(0),\mathcal{O})<\delta$, $u\in\mathcal{U}$ with $\|u\|_\infty<\delta$, and $\bar{v}\in\mathcal{V}$ with $\|\bar{v}\|_\infty<\delta$.
%\end{subdefinition}
%%def

Besides LISS, we will briefly consider local exponential stability (LES) of $\mathcal{O}$. As above, Proposition~\ref{prop:global-existence} of Section~\ref{sec:main-results} below ensures that in a small enough neighborhood of $\mathcal{O}$ solutions of \eqref{eq:sys-imp-eff} exist over arbitrarily long intervals. Hence, the definition below assumes existence of solutions for all $t \geq 0$. 
%Noting that LES is a specialized case of LISS, we can use Proposition~\ref{prop:global-existence} of Section~\ref{sec:main-results} below and choose a sufficiently small $\delta$ that ensures the existence of LES solutions for all time; hence, in the  definition above, we assume that $t_\mathrm{f}\to\infty$.

%def
\begin{definition}\label{def:orbit-ES}
The periodic orbit $\mathcal{O}$ of \eqref{eq:sys-imp-eff} is LES if there exists a $\delta>0$, $N>0$, and $\omega>0$ such that $x(t)=\psi(t,x(0),0,0)$ satisfies for all $t\geq 0$,
%eq
\begin{align}
\mathrm{dist}(x(t),\mathcal{O}) \leq N \mathrm{e}^{-\omega t} \mathrm{dist}(x(0),\mathcal{O}) \nonumber \enspace,
\end{align}
%eq
for any $x(0)\in\mathcal{S}^+$ with $\mathrm{dist}(x(0),\mathcal{O})<\delta$.
\end{definition}
%def
 
In addition to orbital stability, we also present notions of stability for the discrete system \eqref{eq:poinc-dyn}. 
%def
\begin{definition}\label{def:poinc-LISS}
The system \eqref{eq:poinc-dyn} is LISS if there exists a $\delta>0$, $\alpha_1,\alpha_2 \in \mathcal{K}$, and $\beta \in \mathcal{KL}$, such that for all $k\in\mathbb{Z}_+$, 
%eq
\begin{align}\label{eq:poinc-LISS}
\|x_k-x^*\| \leq\beta(\|x_0-x^*\|,k)+\alpha_1(\|u\|_\infty) +\alpha_2(\|\bar{v}\|_\infty)\enspace,
\end{align}
%eq
is satisfied for any $x_0\in \mathcal{S}$ with $\|x_0-x^*\| < \delta$, $u\in\mathcal{U}$ with $\|u\|_\infty <\delta$, and $\bar{v}\in\mathcal{V}$ with $\|\bar{v}\|_\infty <\delta$. 
\end{definition}
%def
\noindent Finally, the 0-input fixed point $x^*$ of \eqref{eq:poinc-dyn} satisfies $x^*=P(x^*,0,0)$, and it is locally asymptotically stable (LAS) or LES if it satisfies the following definition.
%def
\begin{definition}\label{def:poinc-ES}
The 0-input fixed point $x^*$ of \eqref{eq:poinc-dyn} is LAS if there exists a $\delta>0$ and $\beta \in \mathcal{KL}$ such that for all $k\in\mathbb{Z}^+$,
%eq
\begin{align}
\|x_k-x^*\| \leq \beta(\|x_0-x^*\|,k) \nonumber \enspace,
\end{align}
%eq
is satisfied for any $x_0\in \mathcal{S}$ with $\|x_0-x^*\| < \delta$. Furthermore, if there exists a $N>0$, and $0<\rho<1$ such that 
\[
\beta(\|x_0-x^*\|,k) \leq N \rho^k \|x_0-x^*\| \enspace,
\]
then $x^*$ is a LES 0-input fixed point of \eqref{eq:poinc-dyn}. 
\end{definition}
%def
%\noindent Finally we provide the definition of LES which prescribes an exponential convergence rate to a LAS fixed point.
%}
%%def
%\begin{definition}\label{def:poinc-ES}
%The fixed point $x^*$ of \eqref{eq:poinc-dyn} is LES if there exists a $\delta>0$, $N>0$, and $0<\rho<1$ such that for all $k\in\mathbb{Z}^+$,
%%eq
%\begin{align}
%\|x_k-x^*\| \leq N \rho^k \|x_0-x^*\| \nonumber \enspace,
%\end{align}
%%eq
%is satisfied for any $x_0\in \mathcal{S}$ with $\|x_0-x^*\| < \delta$.
%\end{definition}
%%def

%===============================================================
%===============================================================
\section{Main Results}
\label{sec:main-results}
%===============================================================
%===============================================================

In this section we present the main results of this paper. First, we introduce an important proposition on the geometric relation between $\mathcal{O}$ and $\mathcal{S}$. This proposition allows us to express bounds on the orbital distance of any $x\in\mathcal{S}$ from $\mathcal{O}$ equivalently based on the Euclidean distance of $x$ from $x^*$, and vice-versa. The importance of this proposition becomes clear by observing the distance metrics used in Definition~\ref{def:orbit-LISS} and Definition~\ref{def:poinc-LISS}. Hence, it serves as an important bridge between the orbital notions of stability and the Poincar\'e map's stability.
%prop
\begin{prop}\label{prop:equiv-norm}
Let $\mathcal{S}$ be defined as in \eqref{eq:S} and satisfy assumption A.\ref{ass:S}. Let $\mathcal{O}$ be defined as in \eqref{eq:orbit-def} and satisfy assumptions A.\ref{ass:no-zeno}-A.\ref{ass:orbit-transversal}. Then, there exists a $0<\lambda< 1$ such that
%eq
\begin{equation}\label{eq:prop-bound}
\lambda\|x-x^*\| \leq \mathrm{dist}(x,\mathcal{O}) \leq \|x-x^*\| \enspace,
\end{equation}
%eq
for all $x\in\mathcal{S}$.
\end{prop}
%prop
\noindent The proof of Proposition~\ref{prop:equiv-norm} is detailed in Section~\ref{sec:prop-proof} below. 

Proposition~\ref{prop:equiv-norm} can be used to show that solutions in a small enough neighborhood of a LISS periodic orbit $\mathcal{O}$ and for sufficiently small continuous and discrete input signals do not exhibit beating and Zeno behavior, and exist indefinitely. This statement is made precise by the following proposition, a proof of which can be found in Appendix~\ref{app:existence}. 
%prop
\begin{prop}\label{prop:global-existence}
Consider the system \eqref{eq:sys-imp-eff} which satisfies assumptions A.\ref{ass:f-c2}-A.\ref{ass:orbit-transversal}. Suppose that the solutions of \eqref{eq:sys-imp-eff}, denoted by $x(t)=\psi(t,x(0),u,\bar{v})$ and defined in Section~\ref{subsec:system}, satisfy Definition~\ref{def:orbit-LISS}. Then, there exists a $\delta>0$ such that for all $x(0)\in\mathcal{S}^+$ with $\mathrm{dist}(x(0),\mathcal{O})<\delta$, $u\in\mathcal{U}$ with $\|u\|_\infty<\delta$, and $\bar{v}\in\mathcal{V}$ with $\|\bar{v}\|_\infty<\delta$ the following holds:
\begin{itemize}
\item[(i)] $x(t)$ has no consecutive discrete jumps,
\item[(ii)] $x(t)$ does not exhibit Zeno behavior, and
\item[(iii)]$x(t)$ exists for all $t\geq 0$.
\end{itemize}
\end{prop}
%prop

Now we are ready to present the main result of the paper.  
%thm
\begin{theorem}\label{thm:LISS-equivalence}
Consider the system \eqref{eq:sys-imp-eff} which satisfies assumptions A.\ref{ass:f-c2}-A.\ref{ass:delta-C-1} and possesses a periodic orbit $\mathcal{O}$ that is defined as in \eqref{eq:orbit-def} and satisfies assumptions A.\ref{ass:no-zeno}-A.\ref{ass:orbit-transversal}. Then, the following are equivalent.
\begin{itemize}
\item[(i)] $\mathcal{O}$ is an LISS orbit of \eqref{eq:sys-imp-eff};
\item[(ii)] $x^*$ is an LISS fixed point of \eqref{eq:poinc-dyn}.
\end{itemize}
\end{theorem}
%thm

It is straightforward to note that in the absence of inputs ($u\equiv 0$, $\bar{v}\equiv 0$), Theorem~\ref{thm:LISS-equivalence} reduces to the Poincar\'e result for  asymptotic stability of periodic orbits of systems with impulse effects, providing an alternative  proof to \cite[Theorem~1]{grizzle2001asymptotically}. Note though that the proof detailed in the following sections explicitly constructs the class-$\mathcal{KL}$ functions involved in the definitions, thereby providing useful insight on the rates of convergence. The following result can be stated as an immediate corollary of Theorem~\ref{thm:LISS-equivalence}.
%cor
\begin{corollary}\label{cor:ES-equivalence}
Under the assumptions of Theorem~\ref{thm:LISS-equivalence}, the following are equivalent
\begin{itemize}
\item[(i)] $\mathcal{O}$ is an LES 0-input orbit of \eqref{eq:sys-imp-eff};
\item[(ii)] $x^*$ is an LES 0-input fixed point of \eqref{eq:poinc-dyn}.
\end{itemize}
\end{corollary}
%cor

The following remarks are in order.

%rem
\begin{rem}\label{rem:grizzle-proof}
The equivalence between ES of a periodic orbit and ES of the corresponding fixed point of the associated Poincar\'e map has been discussed in~\cite[Theorem~1]{morris2005restricted}, which has been subsequently used in a number of relevant publications, e.g.,\cite{ames2014rapidly, poulakakis2009spring, sreenath2013embedding, ames2007geometric}, and many more. However, \cite[Theorem~1]{morris2005restricted} is proved only for initial states in the Poincar\'e section, as noted above \cite[Equation (6)]{morris2005restricted}, rather than for initial states in a neighborhood of the \emph{entire} orbit, as Definition~\ref{def:orbit-ES} requires. Furthermore, Proposition~\ref{prop:equiv-norm}, which is crucial for commuting between Definition~\ref{def:orbit-ES} and Definition~\ref{def:poinc-ES} is omitted in the proof of~\cite[Theorem~1]{morris2005restricted}, resulting in the estimate in~\cite[Equation~(6)]{morris2005restricted}  being incomplete; the final estimate should have been expressed in terms of $\mathrm{dist}(x, \mathcal{O})$, which requires the use of Proposition~\ref{prop:equiv-norm}.
\end{rem}
%rem
%rem
\begin{rem}
It should be emphasized that the results of this paper can be used to study limit-cycle solutions of continuous-time forced systems like \eqref{eq:cont-dyn} by replacing the discrete update map $\Delta$ with the identity map for the $x$ component and the zero map for the $v$ component.
\end{rem} 
%rem

Theorem~\ref{thm:LISS-equivalence} can be used to establish LISS of a periodic orbit of \eqref{eq:sys-imp-eff} on the basis of LISS of a fixed point of the associated Poincar\'e map \eqref{eq:poinc-dyn}. However, in many  applications---see Section~\ref{sec:sim} for an example---the lack of analytical expressions for the forced Poincar\'e map makes it challenging to establish LISS for a fixed point of it. To alleviate this issue, the following theorem provides a tool for establishing that a 0-input fixed point $x^*$ of the forced Poincar\'e map \eqref{eq:poinc-dyn} is LISS by showing that $x^*$ is a LAS fixed point of the unforced Poincar\'e map. Hence, one can simply linearize the unforced Poincar\'e map and compute the eigenvalues of the associated linearization. If all the eigenvalues lie within the unit disc, the corresponding fixed point is a LAS fixed point of the unforced Poincar\'e map. Then, Theorem~\ref{thm:LAS-LISS} ensures that $x^*$ is a LISS fixed point of the forced Poincar\'e map and Theorem~\ref{thm:LISS-equivalence} establishes LISS of the associated periodic orbit. A result similar to Theorem~\ref{thm:LAS-LISS} can be found in \cite{sontag1996new} for continuous-time systems; however, to the best of our knowledge, we have not seen such a result for discrete systems and we provide it below.
%thm
\begin{theorem}\label{thm:LAS-LISS}
Consider the discrete dynamical system \eqref{eq:poinc-dyn}. Let $\delta>0$ such that $P$ is continuously differentiable in the Fr\'echet sense for $x\in B_\delta(x^*)\cap\mathcal{S}$, $u\in\mathcal{U}$ with $\|u\|_\infty<\delta$, and $v\in B_\delta(0)$. Then, the following are equivalent.
\begin{itemize}
\item[(i)] $x^*$ is an LISS fixed point of \eqref{eq:poinc-dyn};
\item[(ii)] $x^*$ is an LAS 0-input fixed point of \eqref{eq:poinc-dyn}.
\end{itemize}
\end{theorem}
\noindent A proof for Theorem~\ref{thm:LAS-LISS} is presented in Section~\ref{sec:thm-proof}.

%==============================================================
%==============================================================
\section{Proof of Proposition~\ref{prop:equiv-norm}}
\label{sec:prop-proof}
%==============================================================
%==============================================================

The proof of Proposition~\ref{prop:equiv-norm} is organized in a sequence of lemmas.
We begin with a lemma which establishes that for $\mathcal{O}$, the point-to-set distance is equal to the minimum Euclidean distance over the closure of the orbit. As the minimum will be attained by some point in $\overline{\mathcal{O}}$, Lemma~\ref{lem:inf-min} allows us to work with the Euclidean distance from that point instead of dealing with $\inf_{y\in\mathcal{O}} \|x-y\|$.

%lem
\begin{lemma}\label{lem:inf-min}
Let $\mathcal{O}$ be defined as in \eqref{eq:orbit-def} and satisfy assumptions A.\ref{ass:no-zeno}-A.\ref{ass:orbit-transversal}, then for all $x\in\mathbb{R}^n$, we have
%eq
\begin{equation}\nonumber %\label{eq:inf-min}
\mathrm{dist}(x,\mathcal{O}):=\inf_{y\in\mathcal{O}} \|x-y\| = \min_{y\in\overline{\mathcal{O}}} \|x-y\| \enspace.
\end{equation}
%eq
\end{lemma}
%proof
\begin{proof}
Let $x\in\mathbb{R}^n$. The fact that $\overline{\mathcal{O}} = \mathcal{O} \cup \{x^*\}$ implies
%eq
\begin{equation}\label{eq:inf-min-0}
\min_{y\in\overline{\mathcal{O}}} \|x-y\| = \min \{\|x-x^*\|,~\inf_{y\in\mathcal{O}} \| x-y \|\} \enspace. 
\end{equation}
%eq
%\vspace{-4mm}
On the other hand, 
\small
%eq
\begin{align}\label{eq:inf-min-1}
\inf_{y\in\mathcal{O}} \| x-y \| \leq \inf_{y\in\mathcal{O}} (\|x-x^*\| + \| x^*-y \|) = \|x-x^*\| \enspace.
\end{align}
%eq
\normalsize
because $\inf_{y\in\mathcal{O}} \| x^*-y \| = 0$ due to $x^*\in\overline{\mathcal{O}}$. The result follows from \eqref{eq:inf-min-0} in view of \eqref{eq:inf-min-1}. 
\end{proof}
%proof

To simplify notation in the proofs that follow, the closure of the orbit $\overline{\mathcal{O}}$ is parameterized as a function $y(\tau)$ in ``time" like coordinates $\tau$ taking values in a closed interval $[0,T]$. In more detail, let $\varphi^-(t,x^*,0)$ be the solution of the 0-input continuous system \eqref{eq:cont-dyn} backwards in time from the initial state $x^*$. The flow is chosen to be backwards so that $x^*$ is at $\tau=0$. This is primarily for convenience of notation; the flow can be chosen forwards in time starting from $\Delta(x^*,0)$ as well. Then, let $\tau:=t/s$, where $s>0$ is a scaling constant, and define the function $y: [0,T] \to \overline{\mathcal{O}}$ by
%\footnote{The flow is chosen to be backwards so that $x^*$ is at $\tau=0$. This is done merely for convenience, the flow can be chosen forwards in time starting from $\Delta(x^*)$ as well.}
%eq
\begin{align}
y(\tau) := \varphi^-(s\tau,x^*,0) \label{eq:y-tau-def}
\end{align}
%eq 
where 
%eq
\begin{equation} \label{eq:T-star-scaled}
T := T^*/s \enspace .
\end{equation}
%eq
The scaling is performed to ensure that in the Taylor expansion of $y(\tau)$ about $\tau=0$, the first derivative is a vector of unit magnitude. This is done only to simplify notation in the future. Note that $y(\tau)$ should be viewed as a parameterization of the set $\overline{\mathcal{O}}$ and \emph{not} as a solution of the system $\Sigma$ in \eqref{eq:sys-imp-eff}. In fact, this section only deals with geometric properties of $\mathcal{O}$ and $\mathcal{S}$ and does not study the dynamical system as such. The following lemma provides some useful properties of $y(\tau)$.

\begin{lemma}\label{lem:y-tau}
The map $\tau \mapsto y(\tau)$ is bijective and three-times continuously differentiable in $\tau$.
\end{lemma}
\noindent The proof of Lemma~\ref{lem:y-tau} can be found in Appendix~\ref{app:y-tau}.

Let\footnote{Recall from Section~\ref{subsec:notation} that $\mathcal{P}(\cdot)$ is the power set of its argument.} $\tau_{\rm m}: \mathbb{R}^n \to \mathcal{P}(\overline{\mathcal{O}})$ be a set-valued map defined as 
%eq
\begin{equation}\label{eq:tau-m-def}
\tau_{\rm m} (x) := \arg \min_{\tau\in[0,T]} \|x-y(\tau)\| \enspace,
\end{equation}
%eq
where $x\in\mathbb{R}^n$ and $y(\tau)$, $T$ as defined in \eqref{eq:y-tau-def} and \eqref{eq:T-star-scaled}, respectively. Intuitively, the map $x\mapsto\tau_{\rm m}(x)$ returns the \emph{set} $\tau_{\rm m}(x)$ of ``times''  $\tau$ that ``realize'' the points $y(\tau)$ on $\overline{\mathcal{O}}$ that are nearest to $x$.
Hence, for any $\tau_{\min}\in\tau_{\rm m}(x)$, we have 
%eq
\begin{equation}\label{eq:tau_min_prop}
\|x-y(\tau_{\min})\| \leq \|x-y(\tau)\| ~\mbox{for all}~ \tau\in[0,T].
\end{equation}
%eq
%
The next lemma shows that by selecting $x$ sufficiently close to $x^*$, the points on $\overline{\mathcal{O}}$ nearest to $x$ also remain close to $x^*$. 
%lem
\begin{lemma}\label{lem:limit}
Let $\tau_{\rm m}$ be defined as in \eqref{eq:tau-m-def}. Then, for every $\epsilon>0$ there exists $\delta>0$ such that $\|x-x^*\|<\delta$ implies $\tau_{\min}<\epsilon$ for all $\tau_{\min}\in \tau_{\rm m}(x)$.
%Then, 
%%eq
%\begin{align}%\label{eq:lim-tau}
%\forall \epsilon>0,~\exists \delta>0:\|x-x^*\|<\delta \implies & \tau_{\min}<\epsilon, \nonumber \\
%& \forall \tau_{\min}\in \tau_{\rm m}(x) \nonumber \enspace.
%\end{align}
%eq
\end{lemma}
%lem
%
%%lem
%\begin{lemma}\label{lem:limit}
%Let $\tau_{\min} (x) := \arg \min_{\tau\in[0,T]} \|x-y(\tau)\|$ for $x\in\mathbb{R}^n$ and $y(\tau)$ as defined in \eqref{eq:y-tau-def}. Then, 
%%eq
%\begin{equation}\label{eq:lim-tau}
%\lim_{x\to x^*} \tau_{\min}(x) = 0 \enspace.
%\end{equation}
%%eq
%\end{lemma}
%%lem
%proof
\begin{proof}
Let $0<\epsilon<T$. The map $y(\tau)$ is continuous, so $y([\epsilon,T])$ is compact (and thus closed) in $\mathbb{R}^n$. Hence, its complement $y([\epsilon,T])^{\rm c}$ in $\mathbb{R}^n$ is open and contains $y(0)=x^*$ using the injectivity of $y(\tau)$ from Lemma~\ref{lem:y-tau}. As a result, there exists a $\delta>0$ such that $B_{2\delta}(x^*) \subset y([\epsilon,T])^{\rm c}$. 

It can be seen that for any $x\in B_\delta(x^*)$, the points on $\overline{\mathcal{O}}$ closest to $x$ will be within $B_{2\delta}(x^*)$. This follows by a contradiction argument.
Indeed, take any $x\in B_\delta(x^*)$ and let $\tau_{\min}$ be any element of the set\footnote{It is straightforward to note that $\tau_{\rm m}(x)\neq\emptyset$ since $\overline{\mathcal{O}}\neq\emptyset$ and compact.} $\tau_{\rm m}(x)$ defined in \eqref{eq:tau-m-def}. Assume $y(\tau_{\min})$ is outside $B_{2\delta}(x^*)$ so that $\|x-y(\tau_{\min})\|\geq\delta$ by the reverse triangle inequality. But, since $y(0)=x^*\in\overline{\mathcal{O}}$, by \eqref{eq:tau_min_prop} we have $\|x-y(\tau_{\min})\| \leq\|x-x^*\|<\delta$, thus resulting in a contradiction.    
%But, since $x^*\in\overline{\mathcal{O}}$ we also have $\|x-y(\tau_{\min})\|\leq\|x-x^*\|<\delta$, resulting in a contradiction.  
%
%Let $\tau_{\min}$ be any element of\footnote{It is straightforward to note that $\tau_{\rm m}(x)\neq\emptyset$ if $\overline{\mathcal{O}}\neq\emptyset$ and compact.} $\tau_{\rm m}(x)$ defined in the statement of Lemma~\ref{lem:limit}. Let $x\in B_\delta(x^*)$ and assume $y(\tau_{\min})$ is outside $B_{2\delta}(x^*)$. Then $\|x-y(\tau_{\min})\|\geq\delta$. But, $x^*\in\overline{\mathcal{O}}$, therefore $\|x-y(\tau_{\min})\|\leq\|x-x^*\|<\delta$, resulting in a contradiction. 
%
It follows that for any $\tau_{\min}\in \tau_{\rm m}(x)$, $y(\tau_{\min})\in B_{2\delta}(x^*) \subset y([\epsilon,T])^{\rm c}$, and thus $\tau_{\min}<\epsilon$ as a consequence of the injectivity of $y(\tau)$ by Lemma~\ref{lem:y-tau}. As this holds for any $0<\epsilon<T$, the result trivially holds $\forall\epsilon>0$.
%So there is nbd B(x^*,2delta) in the complement. Then one sees that for any x in B(x^*,delta) the closest point on O will be inside B(x^*,2delta), so will correspond to some t in (T-ep, T].
%By Lemma~\ref{lem:inf-min} we have $\inf_{y\in\mathcal{O}} \|x-y\| = \min_{y\in\overline{\mathcal{O}}} \|x-y\|$ and by the continuity of $\inf_{y\in\mathcal{O}} \|x-y\|$ we get that $\min_{y\in\overline{\mathcal{O}}} \|x-y\|$ is continuous. Hence, $\lim_{x\to x^*} \min_{y\in\overline{\mathcal{O}}} \|x-y\| = \min_{y\in\overline{\mathcal{O}}} \|x^*-y\| = 0$. 
%
%As the above limit is convergent, for an arbitrary $\epsilon>0$, there exists a $\delta$ such that $0<\delta<\epsilon/2$, i.e., shrink $\delta$ if necessary. If $x\in B_\delta(x^*)$, then $\min_{y\in\overline{\mathcal{O}}} \|x-y\| = \| x-y(\tau_{\min}(x)) \| <\epsilon/2$. Using reverse triangle inequality $\| y(\tau_{\min}(x))-x^* \| - \| x-x^* \| \leq \| x-y(\tau_{\min}(x)) \| <\epsilon/2$. Noting that $\| x-x^*\|<\epsilon/2$, we get $\| y(\tau_{\min}(x))-x^* \| < \epsilon$. Hence we have 
%%eq
%\begin{equation}\label{eq:lim-1}
%\lim_{x \to x^*} y(\tau_{\min}(x)) = x^* = y(0) \enspace.
%\end{equation}
%%eq
%From Lemma~\ref{lem:inverse} the inverse of $y(\tau)$ is well-defined and continuous. Hence, applying $y^{-1}$ on \eqref{eq:lim-1} gives \eqref{eq:lim-tau}, completing the proof.
\end{proof}
%proof

Next, we present a lemma which shows that the lower bound of Proposition~\ref{prop:equiv-norm} holds locally around $x^*$.
\begin{lemma}\label{lem:equiv-norm-local}
Let $\mathcal{S}$ be defined as in \eqref{eq:S} and satisfy assumption A.\ref{ass:S}. Let $\mathcal{O}$ be as in \eqref{eq:orbit-def} and satisfy assumptions A.\ref{ass:no-zeno}-A.\ref{ass:orbit-transversal}. Then, there exists $\delta>0$ and $0<\bar{\lambda}< 1$ such that
%eq
\begin{equation}\nonumber
\mathrm{dist}(x,\mathcal{O}) \geq \bar{\lambda}\|x-x^*\| \enspace,
\end{equation}
%eq
for all $x\in B_\delta(x^*)\cap\mathcal{S}$.
\end{lemma}
\begin{proof}
This proof is structured as follows. We begin by establishing the desired inequality for states restricted to the vector spaces $T_{x^*} \mathcal{O}$ (tangent line to $\mathcal{O}$ at $x^*$) and $T_{x^*} \mathcal{S}$ (tangent plane to $\mathcal{S}$ at $x^*$) and subsequently introduce non-linearities one-by-one. First, we extend the result to $\mathcal{O}$ and $T_{x^*}\mathcal{S}$, and finally, we extend the result to $\mathcal{O}$ and $\mathcal{S}$. A geometric illustration of the setup can be seen in Fig.~\ref{fig:geometric-illus}. 

Performing Taylor's expansion \cite[p.~349]{lang1993real} of $y(\tau)$ about $\tau=0$, we get
%eq
\begin{equation}\label{eq:taylor-orbit}
y(\tau) = x^* + \tau \nu + \tau^2 r(\tau)
\end{equation}
%eq
where $\nu$ is a unit vector and $\tau^2 r(\tau)$ is the remainder. The scaling factor $s$ in the definition of $\tau$ above \eqref{eq:y-tau-def} is chosen to ensure that $\nu$ has unit length. Let $y'(\tau)$, $r'(\tau)$ be shorthand for $dy/d\tau$ and $dr/d\tau$, respectively. As $y(\tau)$ is three-times continuously differentiable by Lemma~\ref{lem:y-tau}, $r(\tau)$ and $r'(\tau)$ are continuous on $\tau\in[0,T]$, hence there exists $M_r\geq 0$ such that $\|r(\tau)\| \leq M_r$ and $\|r'(\tau)\| \leq M_r$ for all $\tau\in[0,T]$.

\noindent \emph{(i) $T_{x^*} \mathcal{O}$ and $T_{x^*} \mathcal{S}$}\\
The tangent line of $\mathcal{O}$ at $x^*$ is $T_{x^*} \mathcal{O}:=\{ x^*+\tau \nu ~|~ \tau\in\mathbb{R} \}$ as illustrated in Fig.~\ref{fig:geometric-illus}. Given any $z\in T_{x^*}\mathcal{S}$, we have ${\rm dist}(z, T_{x^*}\mathcal{O}) = \inf_{\tau \in \mathbb{R}} \| z - (x^*+ \tau \nu) \|$ and the point on $T_{x^*} \mathcal{O}$ closest to $z$ can be  obtained by projecting the vector $z-x^*$ along the unit vector $\nu$. Specifically, the point on $T_{x^*} \mathcal{O}$ closest to $z$ is given by $x^*+\hat{\tau}_{\min}(z)\nu$ with 
%\textcolor{red}{Given any $z \in T_{x^*}\mathcal{S}$, we have ${\rm dist}(z, T_{x^*}\mathcal{O}) = \inf_{\tau \in \mathbb{R}} \| z - (x^*+ \tau \nu) \|$. Let $\hat{\tau}_{\min}(z) \in \mathbb{R}$ be such that the point  $x^*+\hat{\tau}_{\min}(z)\nu \in T_{x^*} \mathcal{O}$ lies closest to $z$. Then,}
%eq
\begin{equation}\label{eq:tangent-min}
\hat{\tau}_{\min}(z) = \langle z-x^*,\nu \rangle \enspace,
%\langle z - (x^*+ \hat{\tau}_{\min}(z) \nu), \nu \rangle = 0 \Rightarrow
\end{equation}
%eq
where $\langle\cdot,\cdot\rangle$ represents the inner product in $\mathbb{R}^n$. 
Consider now the right triangle with vertices at $x^*$, $z$, and $x^*+\hat{\tau}_{\min}(z)\nu$. Then, $\|z-(x^*+\hat{\tau}_{\min}(z)\nu)\|=\|z-x^*\| \cdot |\sin(\theta(z))|$, where $\theta(z)$ is the angle between $z-x^*$ and $\nu$. By transversality of $\mathcal{O}$ and $\mathcal{S}$ at $x^*$ given by assumption A.\ref{ass:orbit-transversal}, $\theta(z)$ will never be $0$ or $\pi$, so $\min_{\|z-x^*\|= 1}|\sin(\theta(z))|=:\mu$ satisfies $0<\mu\leq 1$. Thus, for all $z\in T_{x^*} \mathcal{S}$, we have
%Let the angle between $z-x^*$ and $\nu$  be $\theta(z)$. By transversality of $\mathcal{O}$ and $\mathcal{S}$ at $x^*$ given by assumption A.\ref{ass:orbit-transversal}, $\theta(z)$ will never be $0$ or $\pi$, so $\min_{\textcolor{blue}{\|z-x^*\|= 1}}|\sin(\theta(z))|=:\mu$ satisfies $0<\mu\leq 1$. Constructing the right-angle triangle with vertices at $x^*$, $z$, and $x^*+\hat{\tau}_{\min}(z)\nu$, we see that $\|z-x^*-\hat{\tau}_{\min}(z)\nu\|=\|z-x^*\| \cdot |\sin(\theta(z))|$. Thus, for all $z\in T_{x^*} \mathcal{S}$, we have
%eq
\begin{align}
{\rm dist}(z, T_{x^*}\mathcal{O}) &= \inf_{\tau \in \mathbb{R}} \| z \!-\! (x^*\!+\! \tau \nu) \| =  \| z \!-\! (x^*\!+\! \hat{\tau}_{\min}(z) \nu) \| \nonumber \\
%&= | \langle z-x^*,~ \nu \rangle | \\
%&= \|z-x^*\| | \sin(\theta(z)) | \\
&= \|z-x^*\| \cdot |\sin(\theta(z))| %\nonumber \\ % \label{eq:right_angle_triangle}
 \geq \mu \|z-x^*\| \enspace.
\label{eq:plane-line}
\end{align}
%eq
%Constructing the right-angle triangle with vertices at $x^*$, $z$, and $x^*+\hat{\tau}_{\min}(z)\nu$, we see that $\|z-x^*-\hat{\tau}_{\min}(z)\nu\|=\|z-x^*\||\sin(\theta(z))|\geq \mu \|z-x^*\|$, thus 
%%eq
%\begin{equation}\label{eq:plane-line}
%\mathrm{dist}(z,T_{x^*}\mathcal{O}) = \|z-x^*-\hat{\tau}_{\min}(z)\nu\| \geq \mu \|z-x^*\| \enspace,
%\end{equation}
%%eq
%for all $z\in T_{x^*} \mathcal{S}$.

%fig
\begin{figure}[t]
\vspace{+0.05in}
\begin{centering}
\includegraphics[width=0.65\columnwidth]{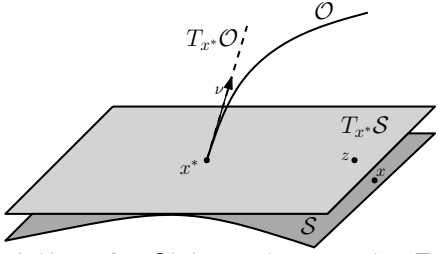} 
\par\end{centering}
\vspace{-0.15in}
\caption{The switching surface $\mathcal{S}$ is in grey, the tangent plane $T_{x^*}S$ is in light grey, the curved line is the orbit $\mathcal{O}$, and the dashed line is $T_{x^*} \mathcal{O}$.}
\vspace{-0.15in}
\label{fig:geometric-illus} 
\end{figure}
%fig

\noindent \emph{(ii) $\mathcal{O}$ and $T_{x^*} \mathcal{S}$}\\
Now we extend the result to $\mathcal{O}$ and $T_{x^*}\mathcal{S}$. Choose $\delta>0$ such that $z \in B_\delta(x^*) \cap T_{x^*}{\cal S}$ implies $\tau_{\min} < T$ with $T$ as in \eqref{eq:T-star-scaled} for all $\tau_{\min} \in \tau_{\rm m}(z)$; Lemma~\ref{lem:limit} guarantees that such a $\delta$ exists.
%Using Lemma~\ref{lem:limit}, there exists a $\delta>0$ such that for $z\in B_\delta(x^*)\cap T_{x^*}\mathcal{S}$ we have $\tau_{\min}<T,~\forall \tau_{\min}\in\tau_{\rm m}(z)$. 
Next, split the set $B_\delta(x^*)\cap T_{x^*}\mathcal{S}$ into two subsets:
\begin{enumerate}[(a)]
\item $E_1:=\{z\in B_\delta(x^*)\cap T_{x^*}\mathcal{S}~|~\exists\tau_{\min}\in\tau_{\rm m}(z):\tau_{\min}=0\}$
\item $E_2:=\{z\in B_\delta(x^*)\cap T_{x^*}\mathcal{S}~|~\forall\tau_{\min}\in\tau_{\rm m}(z),\tau_{\min}\neq 0\}$. 
\end{enumerate}
Clearly, $E_1$ and $E_2$ are disjoint and $E_1 \cup E_2 = B_\delta(x^*)\cap T_{x^*}\mathcal{S}$. 

If $z\in E_1$, then $0 \in \tau_{\rm m}(z)$ and thus
%eq
\begin{equation}\label{eq:E-1}
\mathrm{dist}(z,\mathcal{O})=\inf_{y\in\mathcal{O}} \|z-y\|=\| z-y(\tau_{\min})\| = \|z-x^*\| \enspace.
\end{equation}
%eq

For convenience, from here on when we use $\tau_{\min}$, it is understood that $\tau_{\min}$ can be any element of $\tau_{\rm m}(z)$. Also, we will drop the functional dependence of $\hat{\tau}_{\min}$ on $z$; see \eqref{eq:tangent-min}.

If $z\in E_2$ then $\tau_{\min}>0$ and the vector from $z$ to the nearest point on $\mathcal{O}$ must be orthogonal to the orbit. Hence, we have $\langle z-y(\tau_{\min}),y'(\tau_{\min}) \rangle=0$, which, on using \eqref{eq:taylor-orbit}, gives
%eq
\begin{align}
\langle z-x^*-\tau_{\min} \nu - \tau_{\min}^2 r(\tau_{\min}), \nonumber\\ 
\nu+2\tau_{\min} r(\tau_{\min})+\tau_{\min}^2 r'(\tau_{\min}) \rangle = 0 \enspace. \label{eq:orthogonal}
\end{align}
%eq
Next, we use \eqref{eq:orthogonal} to derive the following important estimate,
%eq
\begin{equation}\label{eq:t2}
\tau_{\min} = \langle z-x^*,\nu \rangle + 2\tau_{\min}\langle r(\tau_{\min}),z-x^* \rangle + \tau_{\min}^2 a(\tau_{\min},z) \enspace,
\end{equation}
%eq
where $a(\tau,z) := \langle r'(\tau),z-x^* \rangle-\langle 3\nu,r(\tau)\rangle - \tau \big(\langle \nu,r'(\tau)\rangle  + \langle 2r(\tau) , r(\tau)\rangle\big) - \tau^2 \langle r(\tau),r'(\tau) \rangle$. Since $r(\tau)$ and $r'(\tau)$ are bounded as discussed below \eqref{eq:taylor-orbit}, $z\in B_\delta(x^*)\cap T_{x^*}\mathcal{S}$, and $\tau_{\min}<T$, we have, for some constant $c>0$, that
\begin{align}
\tau_{\min} \leq \|z-x^*\| + c \tau_{\min}(  \|z-x^*\| + \tau_{\min} ).
\label{eq:t1}
\end{align}
Using $\epsilon=1/(4c)$ in Lemma~\ref{lem:limit} there exists a $\delta<1/(4c)$ (shrink $\delta$ if necessary) such that for $\|z-x^*\| < \delta$, we have $\tau_{\min} < 1/(4c)$. Then from \eqref{eq:t1} we have 
\begin{equation}\label{eq:t3}
\tau_{\min}\leq 2 \|z-x^*\|.
\end{equation}
Next, noting $\hat{\tau}_{\min} = \langle z-x^*, \nu \rangle$ by \eqref{eq:tangent-min}, from \eqref{eq:t2} we obtain
%eq
\begin{align*}
|\hat{\tau}_{\min} - \tau_{\min}| & \leq |2\tau_{\min}\langle r(\tau_{\min}),z-x^* \rangle + \tau_{\min}^2 a(\tau_{\min},z)| \\
& \leq  c ( \tau_{\min} \|z-x^*\| + \tau_{\min}^2 )
\end{align*}
%eq
where the last inequality follows by using bounds similar to those that led to \eqref{eq:t1}.
Using \eqref{eq:t3} in the above inequality and updating the constant\footnote{Intermediate constants of no particular importance are used as $c$ while updating the meaning of $c$ as we proceed with the proof.} $c>0$ accordingly, we have
\begin{equation}\label{eq:tau-diff}
|\hat{\tau}_{\min} - \tau_{\min}| \leq c \|z-x^*\|^2 \enspace,
\end{equation}
provided $\|z-x^*\| < \delta$.

Turning our attention to $\mathrm{dist}(z,\mathcal{O})$ and using Lemma~\ref{lem:inf-min},
%eq
\begin{align}
& \inf_{y\in\mathcal{O}} \|z-y\| = \|z-y(\tau_{\min})\|  \nonumber\\
 & = \|z-x^*-\tau_{\min} \nu- \tau_{\min}^2 r(\tau_{\min})\|\label{eq:linear-inf-norm-1}\\
 & = \|z-x^*-\hat{\tau}_{\min}\nu + (\hat{\tau}_{\min}-\tau_{\min}) \nu- \tau_{\min}^2 r(\tau_{\min})\| \label{eq:linear-inf-norm-2}\\
 & \geq \|z-x^*-\hat{\tau}_{\min}\nu\| - |\hat{\tau}_{\min} - \tau_{\min}| - \tau_{\min}^2 \|r(\tau_{\min})\|  \label{eq:linear-inf-norm-3}\\
 & \geq \mu \|z-x^*\| - c \|z-x^*\|^2 \label{eq:linear-inf-norm-4}\enspace,
\end{align}
%eq
where \eqref{eq:linear-inf-norm-1} is obtained by using \eqref{eq:taylor-orbit}; \eqref{eq:linear-inf-norm-2} is obtained by adding and subtracting $\hat{\tau}_{\min}\nu$; \eqref{eq:linear-inf-norm-3} is obtained by using the reverse triangle inequality; and \eqref{eq:linear-inf-norm-4} is obtained by the boundedness of $r(\tau)$, \eqref{eq:plane-line}, \eqref{eq:tau-diff}, and \eqref{eq:t3}. Again we update the constant $c>0$ accordingly.
Further we can write \eqref{eq:linear-inf-norm-4} as
%eq
\begin{align}
& \inf_{y\in\mathcal{O}} \|z-y\| \geq (\mu/2) \|z-x^*\| + \|z-x^*\|(\mu/2-c \|z-x^*\| ) \nonumber.
\end{align}
%eq
Choosing $\|z-x^*\| \leq \mu/(2c)$ (shrink $\delta>0$ if necessary) gives 
%eq
\begin{equation}\label{eq:E-2}
\mathrm{dist}(z,\mathcal{O}):=\inf_{y\in\mathcal{O}} \|z-y\| \geq (\mu/2) \|z-x^*\| \enspace,
\end{equation}
%eq
for all $z\in E_2$.

Putting together \eqref{eq:E-1} and \eqref{eq:E-2} for the sets $E_1$ and $E_2$, respectively, and noting that $\min \{1,\mu/2 \}=\mu/2$ as $\mu/2\leq 1/2<1$ (see below \eqref{eq:tangent-min} to recall the meaning of $\mu$) gives 
%eq
\begin{equation}\label{eq:orbit-plane-bound}
\mathrm{dist}(z,\mathcal{O}):=\inf_{y\in\mathcal{O}} \|z-y\| \geq (\mu/2) \|z-x^*\| \enspace,
\end{equation}
%eq
for all $z\in B_\delta(x^*)\cap T_{x^*}\mathcal{S}$. 
%This shows the lower bound for states on $\mathcal{O}$ and $T_{x^*} \mathcal{S}$.

\noindent \emph{(iii) $\mathcal{O}$ and $\mathcal{S}$}\\
Here, we extend the result to $x\in B_\delta(x^*)\cap\mathcal{S}$. First, note that if $x\in\mathcal{S}$ is a point in the neighborhood of $x^*$ and $z\in T_{x^*}\mathcal{S}$ is the projection of $x$ on $T_{x^*}\mathcal{S}$, then Appendix~\ref{app:S-projection} shows that there exists a constant $c>0$ such that  
%\footnote{A proof of this fact can be found in Appendix~\ref{app:S-projection}.}
%eq
\begin{equation}\label{eq:x-z-diff}
\|x-z\| \leq c \|x-x^*\|^2 \enspace.
\end{equation}
%eq
%Given $x\in\mathcal{S}$, let $z\in T_{x^*}\mathcal{S}$ be the projection of $x$ on $T_{x^*}\mathcal{S}$ in the local coordinates of $\mathcal{S}$ around $x^*$. As shown in Appendix~\ref{app:S-projection}, the Euclidean distance $\|x-z\|$ is upper bounded by $O(\|x-x^*\|^2)$, i.e., there exist a constant $c>0$ such that  
Since $\|z-x^*\| = \|z-x + x- x^*\| \leq \|x-x^*\| + c\|x-x^*\|^2$ by the triangle inequality and \eqref{eq:x-z-diff}, choosing $x \in \mathcal{S}$ so that $\|x-x^*\|<\delta$ for a sufficiently small $\delta>0$ ensures that the corresponding $z$ satisfies \eqref{eq:orbit-plane-bound}. Then,
%Using the reverse triangle inequality on $\|z-x\|=\|z-x^*+x^*-x\|$ followed by its use in \eqref{eq:x-z-diff} gives $\|z-x^*\|\leq \|x-x^*\|+c\|x-x^*\|^2$. Thus, choosing $\|x-x^*\|<\delta$ (shrink $\delta$ if necessary) ensures that the corresponding $z$ satisfies \eqref{eq:orbit-plane-bound}. Then,
%eq
\begin{align}
\inf_{y\in\mathcal{O}} \|x-y\| & = \inf_{y\in\mathcal{O}} \|x-z+z-y\|  \nonumber\\
& \geq \inf_{y\in\mathcal{O}} \|z-y\| - \|x-z\|  \nonumber\\
& \geq (\mu/2) \|z-x^*\| - \|x-z\|  \label{eq:inf-bound-1} \\
& \geq (\mu/2)\|x-x^*\| - (1+ \mu/2) \|x-z\| \label{eq:inf-bound-2}\\
& \geq (\mu/2)\|x-x^*\| - c \|x-x^*\|^2 \label{eq:inf-bound-3}\enspace.
\end{align}
%eq
where \eqref{eq:inf-bound-1} follows from \eqref{eq:orbit-plane-bound}; \eqref{eq:inf-bound-2} from the reverse triangle inequality on $\|z-x^*\| = \|z-x+x-x^*\|$; and \eqref{eq:inf-bound-3} follows from \eqref{eq:x-z-diff} with $c>0$ updated accordingly.
Write \eqref{eq:inf-bound-3} as
%eq
\begin{equation}\nonumber
 \inf_{y\in\mathcal{O}} \|x-y\| \geq (\mu/4)\|x-x^*\|  + \|x-x^*\| ( \mu/4-c \|x-x^*\| ), 
\end{equation}
%eq
and choose $\|x-x^*\| \leq \mu / (4c)$, shrinking $\delta$ if necessary. The result follows by letting $\bar{\lambda} = \mu/4<1$.
%Then, letting $\bar{\lambda} = \mu/4<1$, we get the bound 
%%eq
%\begin{equation}\label{eq:ub-around-x*}
%\mathrm{dist}(x,\mathcal{O})=\inf_{y\in\mathcal{O}} \|x-y\| \geq \bar{\lambda} \|x-x^*\| \enspace,
%\end{equation}
%%eq
%for all $x\in B_\delta(x^*)\cap\mathcal{S}$.
\end{proof}
 
Now we present the proof of Proposition~\ref{prop:equiv-norm}, which essentially extends Lemma~\ref{lem:equiv-norm-local} to the entire $\mathcal{S}$.

\begin{proof}[Proof of Proposition~\ref{prop:equiv-norm}]
The upper bound on $\mathrm{dist}(x,\mathcal{O})$ in Proposition~\ref{prop:equiv-norm} follows, for all $x\in\mathcal{S}$, directly from \eqref{eq:inf-min-1}. 
%=\inf_{y \in \mathcal{O}} \| x-y \|

For the lower bound, we begin by applying Lemma~\ref{lem:equiv-norm-local} to establish the existence of $\delta>0$ and $0<\bar{\lambda}< 1$ such that
%eq
\begin{equation}\label{eq:ub-around-x*}
\mathrm{dist}(x,\mathcal{O}) \geq \bar{\lambda}\|x-x^*\| \enspace,
\end{equation}
%eq
for all $x \in \mathcal{S}$ with $\|x-x^*\| < \delta$. To obtain a lower bound that holds for all $x \in \mathcal{S}$, we first consider the case where $\mathcal{S}$ is unbounded; then, the case where $\mathcal{S}$ is bounded follows easily. 

Let $\mathcal{S}$ be unbounded and distinguish the following regions.
 
\noindent \emph{(i)}  $\mathcal{R}_{\rm I} := \{x\in\mathcal{S}~|~\|x-x^*\|>\delta'\}$ for $\delta'>\delta$\\
We will show that a $\delta'>\delta$ exists so that for all $x\in\mathcal{R}_{\rm I}$ a lower bound for $\mathrm{dist}(x,\mathcal{O})$ similar to \eqref{eq:ub-around-x*} can be found. 
First note that, by the definition \eqref{eq:S}, the surface $\mathcal{S}$ is closed. Furthermore, by assumption A.\ref{ass:x*} we have $\overline{\mathcal{O}} \cap \mathcal{S} = \{ x^* \}$, and thus the only limit point that $\mathcal{O}$ and $\mathcal{S}$ share is $x^*$. Hence, $\mathrm{dist}(x,\mathcal{O}) > 0$ for all $x \in \mathcal{S}\backslash \{x^*\}$, as these points are in the complement of the closure of $\mathcal{O}$ in $\mathbb{R}^n$, and $\mathrm{dist}(x,\mathcal{O}) / \|x-x^*\| > 0$ is well defined for all $x \in \mathcal{S}\backslash \{x^*\}$. We claim that 
%eq 
\begin{equation}\label{eq:lim-dist}
\lim_{\|x\|\to\infty,~x\in\mathcal{S}\setminus \{x^*\} } \mathrm{dist}(x,\mathcal{O}) / \|x-x^*\| = 1 \enspace, 
\end{equation}
%eq
from which it follows easily that there exists $\delta'>0$ (expand $\delta'$ if necessary to ensure $\delta' > \delta$) such that 
%eq
\begin{equation}\label{eq:ub-to-infty}
\mathrm{dist}(x,\mathcal{O}) \geq (1/2) \|x-x^*\| \enspace,
\end{equation}
%eq
for all $x\in\mathcal{R}_{\rm I}$. %\mathcal{S}$ with $\|x-x^*\|>\delta'
To show the claim \eqref{eq:lim-dist}, take any $x\in\mathcal{S}\backslash \{x^*\}$, let $\tau_{\min}\in \tau_{\rm m}(x)$ and define $M_\mathcal{O}:=\max_{y_1,y_2\in\overline{\mathcal{O}}} \|y_1-y_2\|$ so that $\|x^*-y(\tau_{\min})\| \leq M_\mathcal{O}$. Then, $\mathrm{dist}(x,\mathcal{O}) := \|x-y(\tau_{\min})\| \geq \|x-x^*\|-\|x^*-y(\tau_{\min})\|$ implies 

\vspace{-3mm}
\small
%eq
\begin{equation}\label{eq:dist-ratio-up}
\frac{\mathrm{dist}(x,\mathcal{O})}{\|x-x^*\|} \geq 1 - \frac{\|x^*-y(\tau_{\min})\|}{\|x-x^*\|} \geq 1- \frac{M_\mathcal{O}}{\|x-x^*\|} \enspace,
\end{equation}
%eq
\normalsize
and $\mathrm{dist}(x,\mathcal{O}) = \|x-y(\tau_{\min})\| \leq \|x-x^*\|+\|x^*-y(\tau_{\min})\|$ implies

\vspace{-3mm}
\small
%eq
\begin{equation}\label{eq:dist-ratio-low}
\frac{\mathrm{dist}(x,\mathcal{O})}{\|x-x^*\|} \leq 1 + \frac{\|x^*-y(\tau_{\min})\|}{\|x-x^*\|} \leq 1+ \frac{M_\mathcal{O}}{\|x-x^*\|} \enspace.
\end{equation}
%eq
\normalsize
As a result, for any sequence of points $x_n \in \mathcal{S}\backslash \{x^*\}$ such that $\|x_n-x^*\|\to\infty$, it follows from \eqref{eq:dist-ratio-up} and \eqref{eq:dist-ratio-low} that

\vspace{-3mm}
\small
%eq
\begin{align}\nonumber
1 \leq \liminf_{n\to\infty} \frac{\mathrm{dist}(x_n,\mathcal{O})}{\|x_n-x^*\|} \leq \limsup_{n\to\infty} \frac{\mathrm{dist}(x_n,\mathcal{O})}{\|x_n-x^*\|} \leq 1 \enspace,
\end{align}
%eq
\normalsize
implying $\lim_{n\to\infty} \mathrm{dist}(x_n,\mathcal{O}) / \|x_n-x^*\| =1$, which by \cite[Theorem~4.2]{rudin1964principles} proves the claim \eqref{eq:lim-dist}.

\noindent \emph{(ii)}   $\mathcal{R}_{\rm II} := \{x\in\mathcal{S}~|~\delta\leq\|x-x^*\|\leq\delta'\}$\\
With $\delta>0$ provided by Lemma~\ref{lem:equiv-norm-local} and $\delta' > \delta$ selected as in case \emph{(i)}, let $\hat{\lambda}:=\min_{x\in\mathcal{R}_{\rm II}} \big( \mathrm{dist}(x,\mathcal{O}) / \|x-x^*\| \big) >0$, which is well defined since $\mathrm{dist}(x,\mathcal{O}) / \|x-x^*\| >0$ is continuous over the compact set $\mathcal{R}_{\rm II}$. Hence,   
%eq
\begin{equation}\label{eq:ub-annulus}
\mathrm{dist}(x,\mathcal{O}) \geq \hat{\lambda} \|x-x^*\| \enspace,
\end{equation}
%eq
for all $x\in\mathcal{R}_{\rm II}$. 
%\textcolor{red}{Although not necessary, we note here that $\hat{\lambda} \leq 1$ due to the fact that $\mathrm{dist}(x,\mathcal{O}) / \|x-x^*\| \leq 1$ for all $x \in \mathcal{S}\setminus \{x^*\}$ by the upper bound on $\mathrm{dist}(x,\mathcal{O})$ from \eqref{eq:inf-min-1}.} 

Finally, combining \eqref{eq:ub-around-x*}, \eqref{eq:ub-to-infty}, and \eqref{eq:ub-annulus} by choosing $\lambda=\min\{\bar{\lambda},1/2,\hat{\lambda}\}$ gives 
%eq
\begin{equation}\nonumber
\mathrm{dist}(x,\mathcal{O}):=\inf_{y\in\mathcal{O}} \|x-y\| \geq \lambda \|x-x^*\| ,
\end{equation}
%eq
for all $x\in\mathcal{S}$, completing the proof when $\mathcal{S}$ is unbounded.

For the case where $\mathcal{S}$ is bounded, we can choose $\delta'>0$ sufficiently large to ensure that $\mathcal{S}\subset B_{\delta'}(x^*)$. Then, an argument analogous to that used in \emph{(ii)} above establishes the desired lower  bound for this case, thereby completing the proof of Proposition~\ref{prop:equiv-norm}.
\end{proof}
%proof

%==============================================================
%==============================================================
\section{Proof of Theorem~\ref{thm:LISS-equivalence} and Theorem~\ref{thm:LAS-LISS}}
\label{sec:thm-proof}
%==============================================================
%==============================================================

In the proof of Theorem~\ref{thm:LISS-equivalence}, we will compare different solutions based on their initial states and inputs. With this in mind, we present the following lemma, which is a straightforward adaptation of~\cite[Theorem~3.5]{khalil2002nonlinear} and its proof will be omitted.

\begin{lemma}\label{lem:cont-solutions}
Suppose $f$ in \eqref{eq:cont-dyn} satisfies assumption A.\ref{ass:f-c2}. Let $u_1 \in \mathcal{U}$ and $x_1(t)$ be the solution of
\begin{align*}
\dot{x}(t) = f(x(t),u_1(t)), ~~ x_1(0) = a_1 \enspace,
\end{align*}
which exists for all $t\in[0,\hat{T}]$.
Then, there exists $L>0$ and $\delta>0$ such that if $\|a_2-a_1\|<\delta$ and $\|u_2-u_1\|_\infty<\delta$, 
\[
\dot{x}(t) = f(x(t),u_2(t)), ~~ x_2(0) = a_2  \enspace
\]
has a unique solution $x_2(t)$ for all $t\in[0,\hat{T}]$. Further,

\vspace{-3mm}
\small
\begin{align}\label{eq:lipschitz-solution}
\|x_1(t) - x_2(t)\| & \leq \mathrm{e}^{L\hat{T}} \|a_1 - a_2\| +( \mathrm{e}^{L\hat{T}}-1) \|u_1 - u_2\|_\infty ,
\end{align}
\normalsize
for all $t\in[0,\hat{T}]$.
\end{lemma}
%\noindent The proof of Lemma~\ref{lem:cont-solutions} is analogous to the proof of \cite[Theorem~3.5]{khalil2002nonlinear} and is omitted for brevity.

The following remark extends the unperturbed (zero-input) solution $\varphi(t,\Delta(x^*,0),0)$ of \eqref{eq:cont-dyn} that starts from $\Delta(x^*,0)$.

%rem
\begin{rem}\label{rem:unperturbed}
By assumption~A.\ref{ass:orbit-exist}, the solution $\varphi(t,\Delta(x^*,0),0)$ of \eqref{eq:cont-dyn} exists for all $t\in[0,T^*]$, where $T^* = T_{\rm I}(x^*,0,0)$ is the period of $\mathcal{O}$. Hence, by \cite[Theorem~3.1]{khalil2002nonlinear}, there exists $\overline{T} \in (T^*, \infty)$ so that $\varphi(t,\Delta(x^*,0),0)$ can be extended over the interval $[T^*, \overline{T}]$. As $L_f H(\varphi(T^*,\Delta(x^*,0),0),0)<0$ from assumption~A.\ref{ass:orbit-transversal} while $L_f H(\varphi(t,\Delta(x^*,0),0),0)$ is continuous in time, it follows that for a sufficiently small $\overline{T}>T^*$, we can ensure $L_f H(\varphi(t,\Delta(x^*,0),0),0)<0$ for all $T^* \leq t\leq \overline{T}$, so $\varphi(t,\Delta(x^*,0),0)\in\mathcal{S}^-$ over the interval $(T^*,\overline{T}]$.
\end{rem}
%rem

Lemma~\ref{lem:solution-compare-S} below shows that the perturbed solution $\varphi(t,\Delta(x,v),u)$ of \eqref{eq:cont-dyn} is locally well defined over the interval $[0, \overline{T}]$ where the unperturbed solution $\varphi(t,\Delta(x^*,0),0)$ can be extended. Moreover, the lemma provides a \emph{linear} upper bound on the distance of $\varphi(t,\Delta(x,v),u)$ from $\mathcal{O}$ that will be used in proving Theorem~\ref{thm:LISS-equivalence}. The proof of Lemma~\ref{lem:solution-compare-S} is in Appendix~\ref{app:T-I}.
%, provided that the initial conditions and inputs are in a sufficiently small neighborhood of those corresponding to $\varphi(t,\Delta(x^*,0),0)$
%lem
\begin{lemma}\label{lem:solution-compare-S}
%Assume that the conditions of Theorem~\ref{thm:LISS-equivalence} hold. Let $\overline{T}>T^*$ with $T^* = T_{\rm I}(x^*,0,0)$ being the period of $\mathcal{O}$ be as in Remark~\ref{rem:unperturbed}.  
Under the assumptions of Theorem~\ref{thm:LISS-equivalence}, there exist $\delta>0$, and $\underline{T}>0$, $\overline{T}>T^*$ with $T^* = T_{\rm I}(x^*,0,0)$ being the period of $\mathcal{O}$, such that for all $x\in B_\delta(x^*)\cap\mathcal{S}$, $u\in\mathcal{U}$ with $\|u\|_\infty<\delta$, and $\bar{v}\in\mathcal{V}$ with $\|\bar{v}\|_\infty <\delta$ the following hold 
\begin{enumerate}[(i)]
\item The perturbed solution $\varphi(t,\Delta(x,v),u)$ of \eqref{eq:cont-dyn}, where $v$ is an element of $\bar{v}$, exists and is unique for $t \in [0,\overline{T}]$. 
\item The solution $\varphi(t,\Delta(x,v),u)$ crosses $\mathcal{S}$ in finite time $T_{\rm I}(x,u,v)$ given by \eqref{eq:time-to-imp}, and $0<\underline{T} < T_{\rm I}(x,u,v) < \overline{T}$, where $\underline{T} < T^* < \overline{T}$. Moreover, it does so transversally to $\mathcal{S}$ with  $L_f H(\varphi(T_\mathrm{I}(x,u,v),\Delta(x,v),u),u)<0$. 
\item There exists a $c>0$ such that 
%eq
\begin{align}
 \sup_{0 \leq t < T_{\rm I}} \mathrm{dist} (\varphi(t,\Delta(x,v),u),\mathcal{O}) \leq &~ c \|x-x^*\| + c \|u\|_\infty \nonumber \\
& + c\|\bar{v}\|_\infty . \label{eq:compare-S}
\end{align}
\end{enumerate}
\end{lemma}
%lem

In proving Theorem~\ref{thm:LISS-equivalence}, we will also need Lemma~\ref{lem:solution-compare-S+} which is an ``orbital'' analogue of Lemma~\ref{lem:solution-compare-S}. Lemma~\ref{lem:solution-compare-S+} does not require $x$ to be confined on $\mathcal{S}$ so that any $x \in \mathcal{S}^+$ can be used, as long as it is sufficiently close to $\mathcal{O}$. 
%Lemma~\ref{lem:solution-compare-S+} shows that solutions that start close to $\mathcal{O}$ and evolve under sufficiently small inputs exist, are unique and they cross $\mathcal{S}$ in finite time. Furthermore, such solutions satisfy linear bounds similar to \eqref{eq:compare-S}.  
%bounds on the distance of such solutions from $\mathcal{O}$ that are similar to \eqref{eq:compare-S} are also provided. %when the solution starts in $\mathcal{S}^+$
%lem
\begin{lemma}\label{lem:solution-compare-S+}
Under the assumptions of Theorem~\ref{thm:LISS-equivalence}, there exist $\delta>0$ and $\overline{T}>T^*$ with $T^* = T_{\rm I}(x^*,0,0)$ being the period of $\mathcal{O}$, such that for all $x\in\mathcal{S}^+$ with $\mathrm{dist}(x,\mathcal{O})<\delta$ and $u\in\mathcal{U}$ with $\|u\|_\infty<\delta$,  the following hold 
\begin{enumerate}[(i)]
\item The perturbed solution $\varphi(t,x,u)$ of \eqref{eq:cont-dyn} exists and is unique for $t \in [0, \overline{T}-\xi]$ where $\xi \in [0, T^*]$.
%a subinterval of $[0,\overline{T}]$. %[0,\textcolor{red}{\overline{T}_x}] \subset
%
\item The solution $\varphi(t,x,u)$ crosses $\mathcal{S}$ in finite time $\hat{T}_{\rm I}(x,u)$ given by \eqref{eq:time-to-imp-S+}, and $0 < \hat{T}_{\rm I}(x,u) < \overline{T}$. Moreover, it does so transversally to $\mathcal{S}$ with  $L_f H(\varphi(\hat{T}_\mathrm{I}(x,u),x,u),u)<0$. 
\item There exists a $c>0$ such that 
%eq
\begin{align}
\hspace{-3.5mm} \sup_{0 \leq t < \hat{T}_{\rm I}} \mathrm{dist} (\varphi(t,x,u),\mathcal{O}) \leq & ~c~\mathrm{dist}(x,\mathcal{O}) + c \|u\|_\infty . \label{eq:compare-not-S}
\end{align}
\end{enumerate}
\end{lemma}
%lem
The proof of Lemma~\ref{lem:solution-compare-S+}  is similar to the proof of Lemma~\ref{lem:solution-compare-S}, albeit more technical as it requires the construction of a suitable open cover for $\overline{\mathcal{O}}$; thus, it is relegated to Appendix~\ref{app:T-I}. 
%app:T-hat-I

Now we are ready to provide the proof of Theorem~\ref{thm:LISS-equivalence}. 

%proof
\begin{proof}[Proof of Theorem~\ref{thm:LISS-equivalence}]
We first show (i)$\implies$(ii) and then (ii)$\implies$(i) by explicitly constructing class-$\mathcal{KL}$ and class-$\mathcal{K}$ functions to satisfy Definitions~\ref{def:orbit-LISS} and \ref{def:poinc-LISS}. To avoid ambiguity in notation, we use $x(0)\in\mathbb{R}^n$ as the initial condition for the system with impulse effects \eqref{eq:sys-imp-eff} and $x_0\in\mathcal{S}$ as the initial condition for the discrete system \eqref{eq:poinc-dyn}.

\noindent (i)$\implies$(ii)\\
Assume that $\mathcal{O}$ is a LISS orbit of \eqref{eq:sys-imp-eff}, and let $x(t)=\psi(t,x(0),u,\bar{v})$ be a solution of \eqref{eq:sys-imp-eff} that satisfies Definition~\ref{def:orbit-LISS} for some $\delta_\Sigma>0$ and for suitable functions $\alpha_1, \alpha_2 \in \mathcal{K}$, $\beta \in \mathcal{KL}$. 
%Lemma 9
Choose $\delta_\Sigma$ sufficiently small to further ensure that Lemma~\ref{lem:solution-compare-S+} is satisfied. Then, $t_0:=\hat{T}_{\rm I}(x(0),u)$ is finite, and $x_0:=\lim_{t \nearrow t_0} x(t)$ is well-defined. 
%Lipschitz of \Delta
On arriving at $\mathcal{S}$, the solution jumps to $x(t_0)=\Delta(x_0, v_0)$ where $v_0$ is the first element of the sequence $\bar{v}$. However, to establish the estimate \eqref{eq:poinc-LISS} in Definition~\ref{def:poinc-LISS} we need $x_0$ to appear in the RHS instead of $x(t_0)$. To do this we will use the fact that, by assumption~A.\ref{ass:Delta}, the map $\Delta$ is continuously differentiable and thus locally Lipschitz. Hence, there exists a $\delta_\Delta>0$ for which the Lipschitz condition holds uniformly for all $x \in B_{\delta_\Delta}(x^*) \cap \mathcal{S}$ and $\|\bar{v}\|_\infty < \delta_\Delta$ for some constant $L_\Delta >0$, so that
%eq
\begin{align}
\mathrm{dist}(\Delta(x_0,v_0),\mathcal{O})  & = \mathrm{dist}(\Delta(x_0,v_0),\mathcal{O})-\mathrm{dist}(\Delta(x^*,0),\mathcal{O}) \nonumber \\
& \leq \|\Delta(x_0,v_0)-\Delta(x^*,0)\| \nonumber \\
& \leq L_\Delta \big( \|x_0-x^*\| + \|\bar{v}\|_\infty \big) \enspace, \label{eq:dist-delta-bound}
\end{align}
%eq
where we used $\mathrm{dist}(\Delta(x^*,0),\mathcal{O})=0$ and the fact that $\mathrm{dist}(\cdot,\mathcal{O})$ is Lipschitz continuous with constant equal to $1$.
 %Lemma 8
Finally, to guarantee that the time to impact is well defined for subsequent intersections of $x(t)$ with $\mathcal{S}$, choose $\delta_T>0$ so that Lemma~\ref{lem:solution-compare-S} is satisfied for $\underline{T}>0$ and $\overline{T}>T^*$.

%Translate to the entire orbit, uniformly in k
Now we refine the selection of $\delta$ so that the aforementioned properties hold uniformly along the entire orbit. Pick $0<\delta<\min\{\delta_\Sigma,\delta_T, \delta_\Delta \}$ sufficiently small to also ensure that $\beta(\delta,0)+\alpha_1(\delta) + \alpha_2(\delta) < \min\{ \delta_\Sigma, \lambda \delta_T, \lambda \delta_\Delta \}$; here, by Proposition~\ref{prop:equiv-norm}, $\lambda \in (0,1)$ is a constant such that \eqref{eq:prop-bound} holds for all $x \in \mathcal{S}$. Then, for $\mathrm{dist}(x(0),\mathcal{O})<\delta$, $\|u\|_\infty<\delta$, and $\|\bar{v}\|_\infty<\delta$, 
%we have that
%To do this, we will use the orbital LISS of $\mathcal{O}$
%eq
\begin{align}
\mathrm{dist}(x(t),\mathcal{O}) & \leq \beta(\mathrm{dist}(x(0),\mathcal{O}),t)+\alpha_1(\|u\|_\infty) + \alpha_2(\|\bar{v}\|_\infty) \nonumber \\
& \leq \beta(\delta,0)+\alpha_1(\delta) + \alpha_2(\delta) \nonumber \\
&< \min\{ \delta_\Sigma, \lambda \delta_T, \lambda \delta_\Delta \} \nonumber
\end{align}
%eq
for all $t \geq 0$. Now, by continuity of the distance function, $\lim_{t \nearrow t_k}\mathrm{dist}(x(t),\mathcal{O})=\mathrm{dist}(x_k,\mathcal{O})$, and this choice of $\delta>0$ ensures that $\mathrm{dist}(x_k,\mathcal{O}) < \lambda \delta_\Delta$. Hence, by Proposition~\ref{prop:equiv-norm}, we get $\|x_k-x^*\| < \delta_\Delta$ for all $k\in\mathbb{Z}_+$ and, since this $\delta$ also guarantees $\| \bar{v} \|_\infty < \delta_\Delta$, the bound~\eqref{eq:dist-delta-bound} holds for all $k\in\mathbb{Z}_+$.Similarly, this choice of $\delta$ ensures $\|x_k-x^*\| < \delta_T$ for all $k\in\mathbb{Z}_+$, ensuring that $\underline{T} < T_{\rm I}(x_k,u_k,v_k) < \overline{T}$  uniformly for all $k\in\mathbb{Z}_+$. Finally, this $\delta$ guarantees that $\mathrm{dist}(x_0,\mathcal{O})<\delta_\Sigma$, allowing the use of \eqref{eq:orbit-LISS-ineq} with the same $\beta$, $\alpha_1$, $\alpha_2$ as above.
Putting all these together, for all $k\in\mathbb{Z}_+$ we can write
%eq
%\small
\begin{align}
\| &  x_k -x^*\| \nonumber \\
& \leq \lambda^{-1} \mathrm{dist}(x_k,\mathcal{O}) \nonumber \\ 
& \leq \lambda^{-1}(\beta(\mathrm{dist}(\Delta(x_0,v_0),\mathcal{O}),t_k-t_0)+\alpha_1(\|u\|_\infty) \nonumber\\
& ~~~ + \alpha_2(\|\bar{v}\|_\infty)) \nonumber \\
& \leq \lambda^{-1}(\beta(L_\Delta \big(\|x_0-x^*\| + \|\bar{v}\|_\infty \big),t_k-t_0)+\alpha_1(\|u\|_\infty) \nonumber \\
& ~~~ + \alpha_2(\|\bar{v}\|_\infty))  \nonumber\\
& \leq \lambda^{-1}(\beta(2L_\Delta \|x_0-x^*\|,t_k-t_0) +\alpha_1(\|u\|_\infty) \nonumber \\
& ~~~ + \alpha_2(\|\bar{v}\|_\infty)) \nonumber \\
& \leq \lambda^{-1}(\beta(2L_\Delta\|x_0-x^*\|,k \underline{T})+\alpha_1(\|u\|_\infty) + \alpha_2(\|\bar{v}\|_\infty)  ) ~, \nonumber
\end{align}
%\normalsize%
%eq
where the first inequality follows from Proposition~\ref{prop:equiv-norm}; the second from \eqref{eq:orbit-LISS-ineq} with the solution starting from $\Delta(x_0,v_0)$; the third from \eqref{eq:dist-delta-bound}; the fourth follows by \cite[Lemma~14]{sarkans2016input} on $\beta(L_\Delta \big(\|x_0-x^*\| + \|\bar{v}\|_\infty \big),t_k-t_0)$ followed by absorbing the $\|\bar{v}\|_\infty$ term in $\alpha_2$; and the last inequality follows from the fact that $\beta$ is monotonically decreasing in time and $t_{k+1}-t_k=T_{\rm I}(x_k,u_k,v_k) \geq \underline{T}$ for all $k\in\mathbb{Z}_+$.
Noting that $(1/\lambda) \beta  \in \mathcal{KL}$, $(1/\lambda) \alpha_1 \in \mathcal{K}$ and $(1/\lambda)\alpha_2 \in \mathcal{K}$, and comparing the last inequality with \eqref{eq:poinc-LISS}, we get that the solution of \eqref{eq:poinc-dyn} satisfies Definition~\ref{def:poinc-LISS}, thus completing the proof.
 %(similar to the proof of (ii)$\implies$(i))

%%eq
%\textcolor{blue}{\begin{align}
%\| & x_k -x^*\| \nonumber \\
%& \leq \lambda^{-1} \mathrm{dist}(x_k,\mathcal{O}) \nonumber \\ 
%& \leq \lambda^{-1}(\beta(\mathrm{dist}(x_0,\mathcal{O}),t_k-t_0)+\alpha_1(\|u\|_\infty) + \alpha_2(\|\bar{v}\|_\infty)) \nonumber \\
%& \leq \lambda^{-1}(\beta(\mathrm{dist}(x_0,\mathcal{O}),k \underline{T})+\alpha_1(\|u\|_\infty) + \alpha_2(\|\bar{v}\|_\infty) ) \nonumber \\
%& \leq \lambda^{-1}(\beta(\|x_0-x^*\|,k \underline{T})/\lambda+\alpha_1(\|u\|_\infty) + \alpha_2(\|\bar{v}\|_\infty)  ) \enspace, \nonumber
%\end{align}}
%%eq
%\noindent where the first inequality follows from Proposition~\ref{prop:equiv-norm}; the second  from \eqref{eq:orbit-LISS-ineq}; the third from the fact that $\beta$ is monotonically decreasing in time and $t_{k+1}-t_k=T_{\rm I}(x_k,u_k,v_k) \geq \underline{T}$ for all $k\in\mathbb{Z}_+$; and the last inequality follows from Proposition~\ref{prop:equiv-norm} again. Noting that $(1/\lambda) \beta  \in \mathcal{KL}$, $(1/\lambda) \alpha_1 \in \mathcal{K}$ and $(1/\lambda)\alpha_2 \in \mathcal{K}$, and comparing the last inequality with \eqref{eq:poinc-LISS}, we get that the solution of \eqref{eq:poinc-dyn} satisfies Definition~\ref{def:poinc-LISS}, thus completing the first part of the proof.

\noindent (ii)$\implies$(i)\\
%LISS definition
Assume that $x^*$ is a LISS fixed point of \eqref{eq:poinc-dyn}, and let $\{x_k\}_{k \in \mathbb{Z}_+}$ with $x_k \in \mathcal{S}$ for all $k \in \mathbb{Z}_+$ be a solution of \eqref{eq:poinc-dyn} that satisfies Definition~\ref{def:poinc-LISS} for some $\delta_1>0$. Then, for $x_0 \in B_{\delta_1}(x^*)\cap\mathcal{S}$, $\|u\|_\infty<\delta_1$, and $\|\bar{v}\|_\infty<\delta_1$, there exist suitable functions $\alpha_1, \alpha_2 \in \mathcal{K}$, $\beta \in \mathcal{KL}$ such that \eqref{eq:poinc-LISS} is satisfied. 
%Lemma 7
This implies $\|x_k-x^*\|\leq \beta(\delta_1,0)+\alpha_1(\delta_1) +\alpha_2(\delta_1)$ for all $k\in\mathbb{Z}_+$, and thus $\delta_1$ can be chosen (shrinking it if necessary) so that Lemma~\ref{lem:solution-compare-S} is satisfied for $\varphi(t,\Delta(x_k,v_k),u_k)$ \emph{for all} integers $k \geq 0$. Then,  there exist $\overline{T}_1 > T^*$ (obtained as in Remark~\ref{rem:unperturbed}) so that $T_\mathrm{I}(x_k,u_k,v_k) < \overline{T}_1$ for all $k\in\mathbb{Z}_+$. 

%Lemma 8
The setting above provides a uniform (over $k$) upper bound $\overline{T}_1$ to the impact times $T_\mathrm{I}(x_k,u_k,v_k)$ defining the intervals $[t_k, t_{k+1})$, where $t_{k+1} = t_k+T_\mathrm{I}(x_k,u_k,v_k)$ for $k \in \mathbb{Z}_+$. However, to establish a relation between the discrete-time solution of \eqref{eq:poinc-dyn} and the continuous-time solution \eqref{eq:sys-imp-eff}, we also need to address the interval $[0,t_0)$. By Lemma~\ref{lem:solution-compare-S+}, there exists $\delta>0$ such that, if $x(0)\in\mathcal{S}^+$ satisfies $\mathrm{dist}(x(0),\mathcal{O})<\delta$ and $u\in\mathcal{U}$ satisfies $\|u\|_\infty<\delta$, the solution crosses $\mathcal{S}$ in finite time $t_0 := \hat{T}_\mathrm{I}(x(0), u)$, and there exists a bound $\overline{T}_2>T^*$ so that $t_0 < \overline{T}_2$. We need to make sure that $\delta$ can be selected in a way that $x_0$ satisfies $\|x_0 - x^*\| < \delta_1$ so that Lemma~\ref{lem:solution-compare-S} continues to hold. As before, let $x(t)=\psi(t,x(0),u,\bar{v})$; see Remark~\ref{rem:hyb-cont-sol-rel} for the form of $x(t)$ over the intervals $[t_k, t_{k+1})$. From Lemma~\ref{lem:solution-compare-S+}\emph{(iii)} we have
%; that is, the interval until the solution intersects $\mathcal{S}$ for the first time
%eq
\begin{equation}\label{eq:dist-x-O-nbhd}
\sup_{0\leq t < t_0}\mathrm{dist}(x(t),\mathcal{O})\leq c_1\mathrm{dist}(x(0),\mathcal{O})+c_1\|u\|_\infty \enspace,
\end{equation}
%eq
for some $c_1>0$, which since $x_0 := \lim_{t \nearrow t_0}x(t)$ and the function $\mathrm{dist}(x, \mathcal{O})$ is continuous in $x$, implies that 
%eq
\begin{equation}\label{eq:dist-x0}
\mathrm{dist}(x_0,\mathcal{O})\leq c_1\mathrm{dist}(x(0),\mathcal{O})+c_1\|u\|_\infty \enspace.
\end{equation}
%eq
Since $x_0 \in \mathcal{S}$, Proposition~\ref{prop:equiv-norm} implies $\lambda \|x_0-x^*\| \leq \mathrm{dist}(x_0,\mathcal{O})$ for $\lambda \in (0,1)$, and by \eqref{eq:dist-x0} we have $\lambda \|x_0-x^*\| \leq  c_1\mathrm{dist}(x(0),\mathcal{O})+c_1\|u\|_\infty < 2 c_1 \delta$. Hence, by the inequality above, choosing $\delta < \min \{ \delta_1, \lambda \delta_1/(2c_1) \}$ in Lemma~\ref{lem:solution-compare-S+} ensures that Lemma~\ref{lem:solution-compare-S} continues to hold. Furthermore, in what follows, we define $\overline{T} := \max \{\overline{T}_1,~\overline{T}_2\}$. 
 %%KEEP: the following argument does not use Lemma 8 for the proof.
%First note that, by Proposition~\ref{prop:equiv-norm}, there exist $\lambda<1$ such that $\lambda \|x(0) - x^*\| \leq \mathrm{dist}(x(0),\mathcal{O})$ so that $\mathrm{dist}(x(0),\mathcal{O}) < \delta$ implies $\|x(0) - x^*\| < \delta / \lambda$. This allows us to shrink $\delta$ until $x(0)$ is within a small enough neighborhood of $x^*$ so that Lemma~\ref{lem:cont-solutions} holds; that is, if $\delta_2>0$ is such that Lemma~\ref{lem:cont-solutions} holds, we choose $\delta < \lambda \delta_2$. Then, for $\|x(0) - x^*\| < \delta / \lambda$ and $\| u \|_\infty < \delta / \lambda$, Lemma~\ref{lem:cont-solutions} implies that there exist $L>0$ such that  
%% since $\lambda < 1$
%%eq
%\begin{align}
%\| \varphi(t,x(0),u) &- \varphi(t,x^*,0)\|  \nonumber \\  
%&\leq \mathrm{e}^{L\overline{T}_2} \|x(0) - x^*\| +( \mathrm{e}^{L\overline{T}_2}-1) \|u\|_\infty \nonumber \\
%&\leq (2\mathrm{e}^{L\overline{T}_2} -1) (\delta / \lambda) \nonumber \enspace,
%\end{align}
%%eq
%which holds for all $t \in [0,\overline{T}_2]$ and thus for $t_0 < \overline{T}_2$, resulting in $\|x_0 - x^* \| \leq (2\mathrm{e}^{L\overline{T}_2} -1) (\delta / \lambda)$. Hence, in what follows, we select $\delta < \min \{ \lambda \delta_1 / (2\mathrm{e}^{L\overline{T}_2} -1),~ \lambda \delta_2 \}$. Furthermore, we define $\overline{T} = \min \{\overline{T}_1,~\overline{T}_2\}$.   

The analysis above shows that, under the assumption of $x^*$ being  a LISS fixed point of \eqref{eq:poinc-dyn}, there exist a $\delta > 0$ such that, for $x(0) \in \mathcal{S}^+$ with $\mathrm{dist}(x(0), \mathcal{O})<\delta$ and for $u \in \mathcal{U}$ with $\| u \|_\infty < \delta$, Lemma~\ref{lem:solution-compare-S+} implies $t_0 < \overline{T}$ where $t_0 := \hat{T}_\mathrm{I}(x(0), u)$, and that the solution satisfies the bound \eqref{eq:dist-x-O-nbhd}. If, in addition, $\bar{v} \in \mathcal{V}$ with $\|\bar{v}\|_\infty<\delta$, Lemma~\ref{lem:solution-compare-S} implies $t_{k+1}-t_k < \overline{T}$ with $t_{k+1} = t_k+T_\mathrm{I}(x_k,u_k,v_k)$ for all $k\in\mathbb{Z}_+$. In this case, the solution satisfies the bound
%eq
\begin{align}
 \sup_{t_k \leq t < t_{k+1}} \mathrm{dist} (x(t),\mathcal{O}) \leq &~ c_2 \|x_k-x^*\| + c_2 \|u\|_\infty \nonumber \\
 & + c_2\|\bar{v}\|_\infty , \label{eq:compare-1}
\end{align}
%eq
for some $c_2>0$ and for all $k\in\mathbb{Z}_+$. Substituting \eqref{eq:poinc-LISS} in \eqref{eq:compare-1}, and (for compactness of notation) defining $\alpha(u,\bar{v}):=\hat{\alpha}_1(\|u\|_\infty)+\hat{\alpha}_2(\|\bar{v}\|_\infty)$ with $\hat{\alpha}_1(\|u\|_\infty):=c_2\alpha_1(\|u\|_\infty)+c_2\|u\|_\infty$ and $\hat{\alpha}_2(\|\bar{v}\|_\infty):=c_2\alpha_2(\|\bar{v}\|_\infty)+c_2\|\bar{v}\|_\infty$, results in    
%eq
\begin{align}
\mathrm{dist}(x(t),\mathcal{O}) &  \leq c_2\beta(\|x_0-x^*\|,k) + \alpha(u,\bar{v}) \nonumber \\
& \leq c_2\beta(\mathrm{dist}(x_0,\mathcal{O})/\lambda,k) + \alpha(u,\bar{v}) \label{eq:compare-2}
\end{align}
%eq
for all $t_k \leq t <t_{k+1}$ and $k\in\mathbb{Z}_+$. Note that \eqref{eq:compare-2} was obtained by using Proposition~\ref{prop:equiv-norm} for $\lambda \in (0,1)$. 

With this information, we now construct suitable class-$\mathcal{K}$ and class-$\mathcal{KL}$ functions to prove that the orbit $\mathcal{O}$ is LISS in the system \eqref{eq:sys-imp-eff}. We distinguish the following cases.\\  

\noindent \emph{Case (a):} $t \in [t_k, t_{k+1})$ for any $k\geq1$ .\\
By Remark~\ref{rem:hyb-cont-sol-rel} we have that $t_{k+1} \leq (k+2) \overline{T}$. 
Pick $c_3>0$ such that for all $k \geq 1$, $c_3(k+2)\overline{T} \leq k$. This can be  ensured by selecting $c_3\leq 1/(3\overline{T})$. Using this in \eqref{eq:compare-2} implies
%eq
\begin{align}
\mathrm{dist}(x(t),\mathcal{O}) &  \leq c_2\beta(\mathrm{dist}(x_0,\mathcal{O})/\lambda,c_3(k + 2)\overline{T} ) + \alpha(u,\bar{v}) \nonumber \\
& \leq c_2\beta(\mathrm{dist}(x_0,\mathcal{O})/\lambda,c_3t) + \alpha(u,\bar{v}) \enspace, \label{eq:compare-3}
\end{align}
%eq
which is obtained by using the fact $t < t_{k+1} \leq (k+2)\overline{T}$. The estimate \eqref{eq:compare-3} holds for all $k \in \mathbb{Z}_+$ with $k \geq 1$, and thus it holds for all $t\geq t_1$.\\

\noindent \emph{Case (b):} $t \in [t_k, t_{k+1})$ for $k=0$ . \\
Choose $\gamma:=\ln(2)/(2\overline{T})$ so that $c_2 \beta(\mathrm{dist}(x_0,\mathcal{O})/\lambda,0) \leq 2 c_2 \beta(\mathrm{dist}(x_0,\mathcal{O})/\lambda,0) \mathrm{e}^{-\gamma t} $ over the interval $[t_0,t_0+\overline{T}) \subset [0, 2\overline{T}]$. Then, since $t_1 < t_0 + \overline{T}$, the estimate \eqref{eq:compare-2} implies that, for all $t \in [t_0,t_1)$,
%eq
\begin{align}
\hspace{-2mm} \mathrm{dist}(x(t),\mathcal{O}) & \leq c_2 \beta(\mathrm{dist}(x_0,\mathcal{O})/\lambda,0) + \alpha(u,\bar{v}) \nonumber \\
& \leq 2 c_2 \beta(\mathrm{dist}(x_0,\mathcal{O})/\lambda,0) \mathrm{e}^{-\gamma t} + \alpha(u,\bar{v})\enspace. \label{eq:compare-4}
\end{align}
%eq 

Now we can combine the estimates \eqref{eq:compare-3} and \eqref{eq:compare-4} to obtain a bound for all $t \geq t_0$ of the distance from $\mathcal{O}$ of a solution starting from $x_0$ at time $t_0$. Note that the function
%eq
\begin{align} 
\hat{\beta}(\mathrm{dist}(x_0,\mathcal{O}),t):= c_2 \max \{ & 2 \beta(\mathrm{dist}(x_0,\mathcal{O})/\lambda,0) \mathrm{e}^{-\gamma t},\nonumber \\
& \beta(\mathrm{dist}(x_0,\mathcal{O})/\lambda,c_3 t) \} \nonumber \enspace,%\label{eq:intermediate} 
\end{align}
%eq
is continuous, monotonically increasing in $\mathrm{dist}(x_0,\mathcal{O})$, and monotonically decreasing in $t$ because the individual functions in the $\max$ have the same properties; hence, $\hat{\beta}\in\mathcal{KL}$. Upper bounding \eqref{eq:compare-3} and \eqref{eq:compare-4} with $\hat{\beta}$ and remembering  that $\alpha(u,\bar{v}):=\hat{\alpha}_1(\|u\|_\infty) + \hat{\alpha}_2(\|\bar{v}\|_\infty)$ results in
%eq
\begin{align}\label{eq:LISS-x0}
\mathrm{dist}(x(t),\mathcal{O}) \leq & \hat{\beta}(\mathrm{dist}(x_0,\mathcal{O}),t) + \hat{\alpha}_1(\|u\|_\infty) + \hat{\alpha}_2(\|\bar{v}\|_\infty) \enspace,
\end{align}
%eq
which holds for all $t\geq t_0$. 

To complete the proof of this case, we need an estimate in which the class-$\mathcal{KL}$ function in the RHS of \eqref{eq:LISS-x0} depends on $\mathrm{dist}(x(0),\mathcal{O})$ and not $\mathrm{dist}(x_0,\mathcal{O})$; recall that $x(0)$ is the initial state of the solution of \eqref{eq:sys-imp-eff}, i.e., $\psi(t,x(0),u,\bar{v})$, while $x_0$ is the first intersection of $\psi$ with $\mathcal{S}$. To remedy this, use \eqref{eq:dist-x0} noting that $\hat{\beta}(\mathrm{dist}(x_0,\mathcal{O}),t)$ in \eqref{eq:LISS-x0} is a class-$\mathcal{K}$ function for any fixed $t$, followed by \cite[Lemma~14]{sarkans2016input} to get
%
%\footnote{\textcolor{blue}{Note that \cite[Lemma~14]{sarkans2016input} holds for $\alpha\in\mathcal{K}_\infty$ but it is straightforward to extend this result for $\alpha\in\mathcal{K}$.} }
%
% $\epsilon=1$ to get,
%eq
\begin{align}
\hat{\beta}(\mathrm{dist}(x_0,\mathcal{O}),t) &\leq \hat{\beta}(c_1\mathrm{dist}(x(0),\mathcal{O})+c_1\|u\|_\infty,t) \nonumber \\
& \leq \hat{\beta}(2c_1\mathrm{dist}(x(0),\mathcal{O}),t)+\hat{\beta}(2c_1\|u\|_\infty,t)\nonumber \\
& \leq \hat{\beta}(2c_1\mathrm{dist}(x(0),\mathcal{O}),t)+\hat{\beta}(2c_1\|u\|_\infty,0)\nonumber .
\end{align}
%eq
Use this inequality in \eqref{eq:LISS-x0} and with an abuse of notation absorb the second term of the above inequality in $\hat{\alpha}_1(\|u\|_\infty)$ to obtain
%eq
\begin{align}
\mathrm{dist}(x(t),\mathcal{O})  \leq & \hat{\beta}(2c_1\mathrm{dist}(x(0),\mathcal{O}),t)  \nonumber \\
& + \hat{\alpha}_1(\|u\|_\infty)+ \hat{\alpha}_2(\|\bar{v}\|_\infty) \enspace . \label{eq:LISS-t0-onwards}
\end{align}
%eq
However, \eqref{eq:LISS-t0-onwards} merely holds for $t\geq t_0$ and not for all $t\geq 0$. To address this issue we consider the following case.\\

\noindent \emph{Case (c):} $t \in [0, t_0)$ . \\
We use the bound \eqref{eq:dist-x-O-nbhd}, which is a consequence of Lemma~\ref{lem:solution-compare-S+}. Employing a trick similar to the one used for constructing the class $\mathcal{KL}$ function in \emph{Case (b)}, let\footnote{In fact, for this case, $\gamma=\ln(2)/\overline{T}$ would suffice, but we use the same $\gamma$ as in \emph{Case (b)} to avoid introducing additional constants.}
$\gamma:=\ln(2)/(2 \overline{T})$. Then, $c_1 \mathrm{dist}(x(0),\mathcal{O}) \leq 2 c_1 \mathrm{dist}(x(0),\mathcal{O}) \mathrm{e}^{-\gamma t}$ over the interval $[0, \overline{T}]$. Then, since $t_0 \leq \overline{T}$, \eqref{eq:dist-x-O-nbhd} gives for $t \in [0, t_0)$,
%eq
\begin{equation}\label{eq:LISS-0-t0}
\mathrm{dist}(x(t),\mathcal{O})\leq 2c_1 \mathrm{dist}(x(0),\mathcal{O}) \mathrm{e}^{-\gamma t}+c_1\|u\|_\infty \enspace.
\end{equation}
%eq

We now combine the bound \eqref{eq:LISS-t0-onwards} for $t \geq t_0$ with the bound \eqref{eq:LISS-0-t0} for $t \in [0, t_0)$ to construct class-$\mathcal{KL}$ and class-$\mathcal{K}$ functions that satisfy Definition~\ref{def:orbit-LISS}. Indeed, defining
%eq
\begin{align*}
\tilde{\beta}(\mathrm{dist}(x(0),\mathcal{O}),t) :=\max\{ & 2c_1\mathrm{dist}(x(0),\mathcal{O}) \mathrm{e}^{-\gamma t}, \\
& \hat{\beta}(2c_1\mathrm{dist}(x(0),\mathcal{O}),t)\} \enspace,
\end{align*}
%eq
and $\tilde{\alpha}_1(\|u\|_\infty) :=\hat{\alpha}_1(\|u\|_\infty)+c_1\|u\|_\infty$ and $\tilde{\alpha}_2(\|\bar{v}\|_\infty) :=\hat{\alpha}_2(\|\bar{v}\|_\infty)$ implies that the solution $x(t)$ satisfies Definition~\ref{def:orbit-LISS} for all $t\geq 0$, thereby completing the proof of Theorem~\ref{thm:LISS-equivalence}.  
%\vspace{-7mm}
%\begin{align*}
%\textcolor{blue}{\tilde{\alpha}}_1(\|u\|_\infty) & :=\hat{\alpha}_1(\|u\|_\infty)+c_1\|u\|_\infty, ~ \textcolor{blue}{\tilde{\alpha}}_2(\|\bar{v}\|_\infty)& :=\hat{\alpha}_2(\|\bar{v}\|_\infty).
%\end{align*}
%Then, for all $t\geq 0$, we have
%%eq
%\begin{align}\nonumber
%\mathrm{dist}(x(t),\mathcal{O}) & \leq \tilde{\beta}(\mathrm{dist}(x(0),\mathcal{O}),t) + \tilde{\alpha}_1(\|u\|_\infty) + \tilde{\alpha}_2(\|\bar{v}\|_\infty),
%\end{align}
%%eq
%which satisfies Definition~\ref{def:orbit-LISS}.
\end{proof}

Next, we present a proof of Corollary~\ref{cor:ES-equivalence}.
%proof
\begin{proof}[Proof of Corollary~\ref{cor:ES-equivalence}]
The proof is identical to that of Theorem~\ref{thm:LISS-equivalence} with $u\equiv 0$, $\bar{v}\equiv 0$. Ony note that in proving (ii)$\implies$(i) we choose $\omega>0$ such that $\rho=\mathrm{e}^{-\omega \overline{T}}$ as in Definition~\ref{def:poinc-ES}.
%(see Definition~\ref{def:orbit-ES}) such that $\rho$ in Definition~\ref{def:poinc-ES} satisfies $\rho=\mathrm{e}^{-\omega \overline{T}}$.
\end{proof}
%proof

Finally, we present the proof of Theorem~\ref{thm:LAS-LISS}.
%proof
\begin{proof}[Proof of Theorem~\ref{thm:LAS-LISS}]
The proof of \emph{(i)}$\implies$\emph{(ii)} trivially follows by substituting $u= 0$ and $v=0$ in Definition~\ref{def:poinc-LISS} to recover Definition~\ref{def:poinc-ES} that establishes LAS. For \emph{(ii)}$\implies$\emph{(i)},  we assume that $x^*$ is a LAS fixed point of the 0-input system. By \cite[Theorem~1(1)]{jiang2002converse}, there exists a smooth (discrete) Lyapunov function $V : \mathbb{R}^n \to \mathbb{R}_+$ so that for all $x\in B_\delta(x^*)\cap\mathcal{S}$,        
%eq
\begin{equation}\label{eq:V-Lyap}
V(P(x,0,0))-V(x) \leq -\alpha(\|x-x^*\|) \enspace,
\end{equation}
%eq
with $\alpha\in\mathcal{K}$. Since $V$ is smooth and $P$ is continuously differentiable, $V \circ P : \mathcal{S} \times \mathcal{U} \times \mathbb{R}^q \to \mathbb{R}_+$ is locally Lipschitz, and $\delta>0$ can be chosen (shrink if necessary) so that the Lipschitz condition holds uniformly for some $L_{PV}>0$ for all $x\in B_\delta(x^*)\cap\mathcal{S}$, $u\in\mathcal{U}$, $\|u\|_\infty<\delta$, and $v\in B_\delta(0)$. Then,
%eq
\begin{align}
 V(P( & x,u,v))-V(x) = V(P(x,0,0))-V(x) \nonumber \\
 & + V(P(x,u,v))-V(P(x,0,0)) \nonumber \\
& \leq -\alpha(\|x-x^*\|) + L_{PV} \|u\|_\infty + L_{PV} \|\bar{v}\|_\infty \enspace,
\end{align}
%eq
where the inequality follows from \eqref{eq:V-Lyap} and the Lipschitz continuity of $V\circ P$. Hence, $V$ is a LISS Lyapunov function as required by \cite[Definition~3.2]{jiang2001input} and from \cite[Lemma 3.5]{jiang2002converse} it follows that the system is LISS.
\end{proof}
%proof

%===============================================================
%===============================================================
\section{Example: LISS of a Bipedal Walker}
\label{sec:sim}
%===============================================================
%===============================================================

This section presents an example of applying Theorems~\ref{thm:LISS-equivalence} and~\ref{thm:LAS-LISS} to establish ISS for a periodic walking gait for the biped of Fig.~\ref{fig:model}. The model is underactuated, having five degrees of freedom (DOF) and four actuators; two actuators are placed at the knees  and two at the hip. Ground contact is modeled as a passive pivot. More details about the model along with the mechanical properties used here can be found in \cite[Table~6.3]{westervelt2007feedback}.  
%we will study the ISS of a planar bipedal robot model which has two legs with knees and a torso; see Fig.~\ref{fig:model}. The model has two actuators at the knee joint and two at the hip joint while the contact between the stance foot and the ground is modeled as a passive pivot. The mechanical properties of the model can be found in \cite[Table~6.3]{westervelt2007feedback}.

As shown in Fig.~\ref{fig:model}, we choose $q:=(q_1,q_2,q_3,q_4,q_5)$ as coordinates for the configuration space $\mathcal{Q}$. Let $\Gamma$ be the actuator inputs and $F_{\rm e}\in\mathcal{F}:=\{F_{\rm e}:\mathbb{R}_+\to\mathbb{R}^2~|~\|F_{\rm e}\|_\infty<\infty,~F_{\rm e}~{\rm is~continuous} \}$ be an external force acting at the torso as shown in Fig.~\ref{fig:model}. Further, let $x\in\mathcal{TQ}:=\{(q,\dot{q})~|~q\in\mathcal{Q},~\dot{q}\in\mathbb{R}^5\}$ be the state. The swing phase dynamics is 
%eq
\begin{equation}\label{eq:robot-state-space-dyn}
\dot{x} = f(x) + g(x)\Gamma + g_{\rm e}(x)F_{\rm e} \enspace.
\end{equation}
%eq
This phase terminates when the swing foot makes contact with the ground. This occurs when $x \in \mathcal{S}:=\{(q,\dot{q})\in\mathcal{TQ}~|~p_{\rm v}(q)=0 \}$, where $p_{\rm v}(q)$ is the height of the swing foot from the ground. The ensuing impact map $\Delta$ takes the states $x^-$ just before to the states $x^+$ just after impact, under the influence of an impulsive disturbance $F_\mathrm{I} \in \mathbb{R}^2$ applied at the same point as $F_{\rm e}$. The resulting system takes the form  
%i.e., $x^+ = \Delta(x^-,F_\mathrm{I})$, where $\Delta$ is smooth.
%The state $x^+$ reinitializes \eqref{eq:robot-state-space-dyn}, giving rise to
%eq
\begin{eqnarray}
\Sigma:
\begin{cases}
\begin{aligned}
\dot{x} &= f(x) + g(x)\Gamma + g_{\rm e}F_{\rm e}(t) & \hspace{-1mm}\mathrm{if}~x\notin \mathcal{S}	\\
x^+ &= \Delta(x^-,F_\mathrm{I}) &\hspace{-1mm}\mathrm{if}~x^-\in \mathcal{S}
		\end{aligned}
	\end{cases} \enspace, \label{eq:SIE-robot}
\end{eqnarray}
%eq
in which $\Delta$ is smooth. Note that $F_{\rm e}$ and $F_\mathrm{I}$ are viewed as continuous and discrete disturbances, respectively. There is a variety of methods available for designing control laws $\Gamma(x)$ that result in asymptotically stable limit-cycle gaits in the absence of the disturbances; here, we use the method in \cite{veer2015adaptation}.
%$F_{{\rm e},k}$ and $J$; see \cite{westervelt2007feedback,Gregg2010Geometric,freidovich2009passive} for examples

Let $\bar{F}_\mathrm{I}:=\{ F_{\mathrm{I},k}\}_{k=0}^\infty$ be the sequence of impulsive disturbances and $\{F_{{\rm e},k}\}_{k=0}^\infty$ be the sequence of continuous inputs as in Section~\ref{subsec:poinc}. Then, \eqref{eq:SIE-robot} in closed loop with $\Gamma(x)$, gives rise to the forced Poincar\'e map $P:\mathcal{S}\times\mathcal{F}\times\mathbb{R}^2\to\mathcal{S}$
%eq
\begin{equation}\label{eq:poinc-discrete-robot}
x_{k+1} = P(x_k,F_{{\rm e},k}, F_{\mathrm{I},k}) \enspace,
\end{equation}
%eq
which captures the dynamics of \eqref{eq:SIE-robot} as it goes through $\mathcal{S}$. A simple calculation shows that all the eigenvalues of the linearization of the Poincar\'e map about a 0-input fixed point $x^*$ are within the unit disk so that $x^*$ is LAS. Then, Theorem~\ref{thm:LAS-LISS} implies that $x^*$ is a LISS fixed point of the forced Poincar\'e map \eqref{eq:poinc-discrete-robot} and Theorem~\ref{thm:LISS-equivalence} ensures that the corresponding periodic orbit $\mathcal{O}$ is LISS in the presence of disturbances.
%in the presence of disturbances

% using the method of Lagrange, the swing-phase dynamics are
%%eq
%\begin{equation}\label{eq:lagrange-dyn}
%D(q)\ddot{q} + C(q,\dot{q})\dot{q} + G(q) = B\Gamma + J^{\rm T}_{\rm e}(q) F_{\rm e} \enspace,
%\end{equation}
%%eq
%where $D$ is the inertia matrix, $C$ is the Coriolis matrix, $G$ is the gravitational vector, $B$ maps actuator torques to the corresponding joints, and $J_{\rm e}$ is the Jacobian of the point on the torso where $F_{\rm e}$ is applied.
%
%Let $x\in\mathcal{TQ}:=\{(q,\dot{q})~|~q\in\mathcal{Q},~\dot{q}\in\mathbb{R}^5\}$ be the state of the model, then \eqref{eq:lagrange-dyn} can be written in the state-space form as
%%eq
%\begin{equation}\label{eq:robot-state-space-dyn}
%\dot{x} = f(x) + g(x)\Gamma + g_{\rm e}(x)F_{\rm e} \enspace.
%\end{equation}
%%eq

%fig
\begin{figure}[t]
\vspace{+0.05in}
\begin{centering}
\includegraphics[width=0.47\columnwidth]{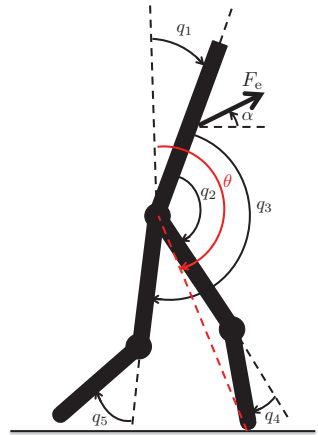} 
\par\end{centering}
\vspace{-0.1in}
\caption{Robot model with a choice of generalized coordinates.}
\vspace{-0.1in}
\label{fig:model} 
\end{figure}
%fig 

%fig
\begin{figure*}[t]
\centering
\subfigure[]
{
\includegraphics[width=0.39\textwidth]{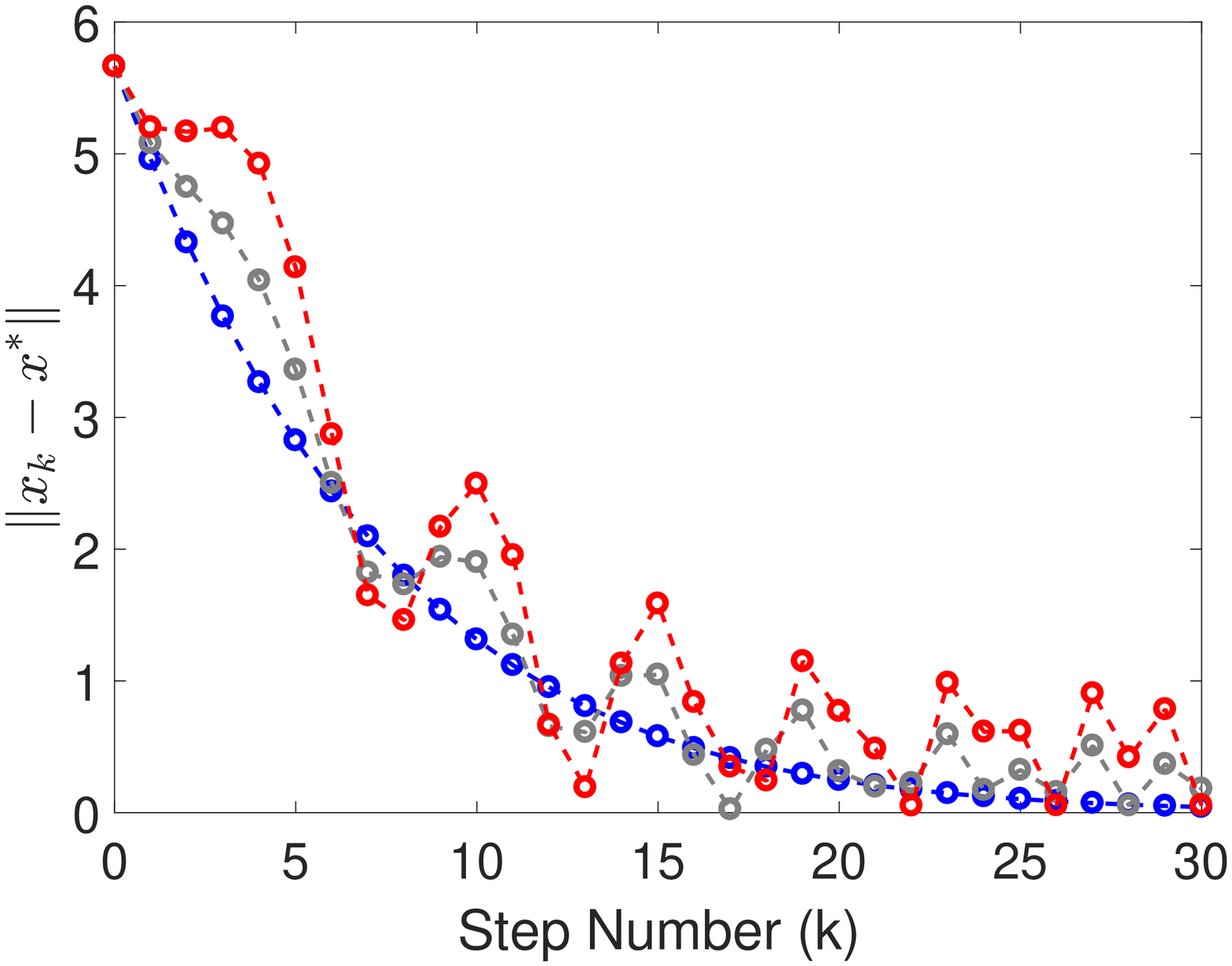}
\label{fig:poincare-evolution}
}
\centering
\hspace{10mm}
\subfigure[]
{
\includegraphics[width=0.39\textwidth]{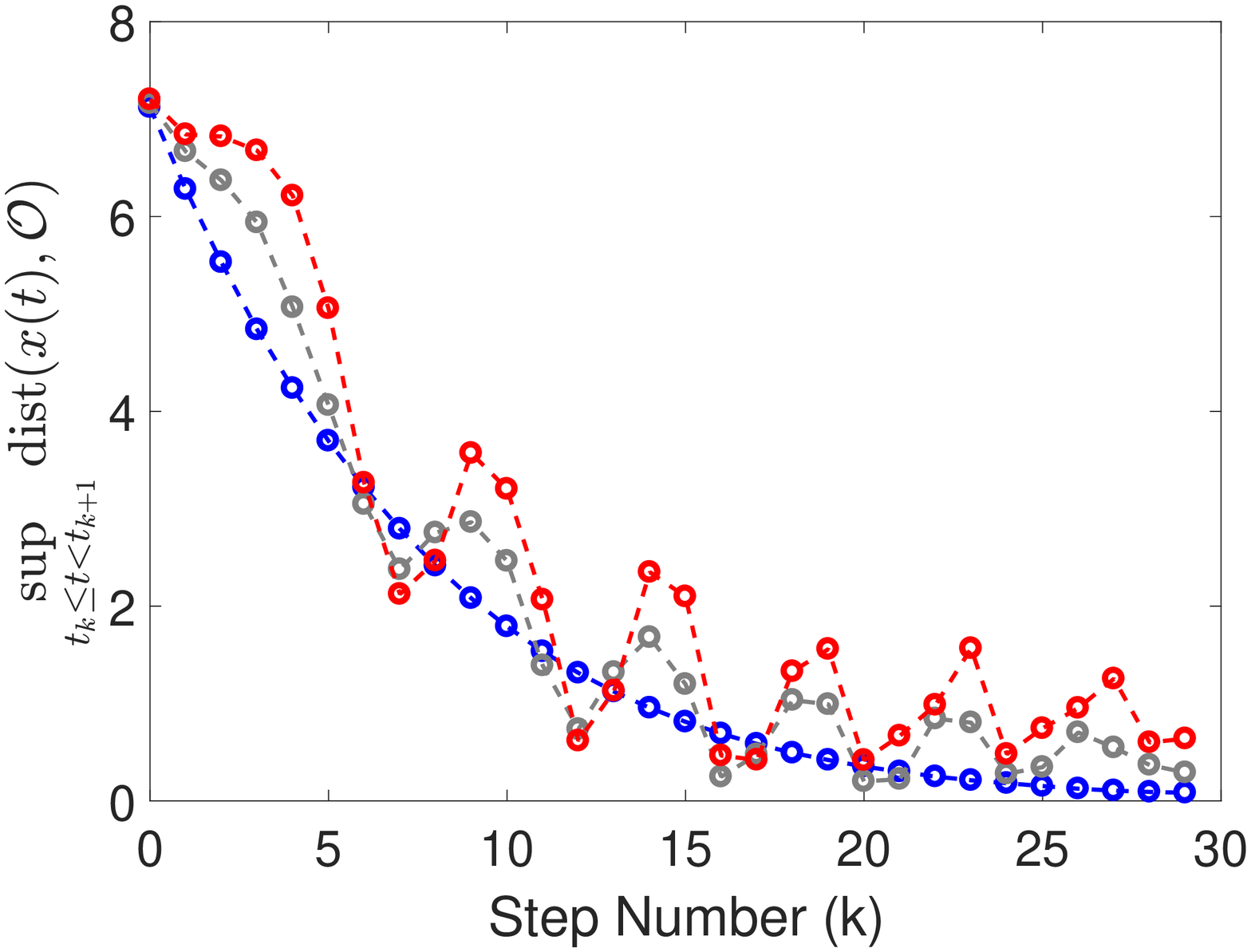}
\label{fig:SIE-evolution}
}
\vskip -10pt
\caption{Response of the biped from an initial condition away from the orbit with $F_{\rm e}=0$ N, $\|\bar{F}_\mathrm{I}\|_\infty=0$ N.s in blue; $F_{\rm e} = [5 \sin(4t)~0]^{\rm T}$ N, $\|\bar{F}_\mathrm{I}\|_\infty=0.1$ N.s in gray; and $F_{\rm e}=[10\sin(4t)~0]^{\rm T}$ N, $\|\bar{F}_\mathrm{I}\|_\infty=0.2$ N.s in red.  \textbf{(a)} Evolution of $\|x_k-x^*\|$ over step number $k$ where $\{x_k\}_{k=0}^\infty$ is the solution of \eqref{eq:poinc-discrete-robot}. \textbf{(b)} Supremum deviation of  $x(t)=\psi(t,x(0),F_{\rm e},\bar{F}_\mathrm{I})$, the solution of \eqref{eq:SIE-robot}, from the orbit $\mathcal{O}$ over each step. \label{fig:simulation}}
\vskip -10pt
\end{figure*}
%fig

%\subsection{Simulations}
Figure~\ref{fig:simulation} shows the behavior of the model when $F_\mathrm{e}$ is a horizontal sinusoidal force and $\bar{F}_\mathrm{I}$ consists of horizontal impulses $\bar{F}_{\mathrm{I},k}$ uniformly sampled from the interval $[-\|\bar{F}_\mathrm{I}\|_\infty,\|\bar{F}_\mathrm{I}\|_\infty]$ for a predetermined $\|\bar{F}_\mathrm{I}\|_\infty<\infty$. It can be seen that the solution $\{x_k\}_{k=0}^\infty$ of \eqref{eq:poinc-discrete-robot} in Fig.~\ref{fig:poincare-evolution} as well as the solution $x(t)=\psi(t,x(0),F_{\rm e},\bar{F}_\mathrm{I})$ of \eqref{eq:SIE-robot} in Fig.~\ref{fig:SIE-evolution}, asymptotically converge to 0 in the absence of disturbances (blue); are ultimately bounded when $\|F_{\rm e}\|_\infty = 5$~N, $\|\bar{F}_\mathrm{I}\|_\infty=0.1$~N.s (gray) and $\|F_{\rm e}\|_\infty = 10$~N, $\|\bar{F}_\mathrm{I}\|_\infty=0.2$~N.s (red), with the ultimate bound for the former being smaller than that of the latter, indicating LISS behavior.
%For the purpose of simulation, a horizontal force in the form of a sine function of time is applied on the biped. The sequence of discrete disturbances $\bar{F}_\mathrm{I}$ consists of horizontal impulses $J_k$ which are uniformly sampled from the interval $[-\|\bar{J}\|_\infty,\|\bar{J}\|_\infty]$ for a predetermined $\|\bar{J}\|_\infty<\infty$. We compare the evolution of the state of the biped with varying intensities of external force and impulsive disturbance; see Fig.~\ref{fig:simulation}. The solution $\{x_k\}_{k=0}^\infty$ of \eqref{eq:poinc-discrete-robot} in Fig.~\ref{fig:poincare-evolution} as well as the solution $x(t)=\psi(t,x(0),F_{\rm e},\bar{J})$ of \eqref{eq:SIE-robot} in Fig.~\ref{fig:SIE-evolution}, asymptotically converge to 0 in the absence of disturbances (blue); are ultimately bounded when $\|F_{\rm e}\|_\infty = 5$~N, $\|\bar{J}\|_\infty=0.1$~N.s (gray) and $\|F_{\rm e}\|_\infty = 10$~N, $\|\bar{J}\|_\infty=0.2$~N.s (red), however the ultimate bound for the former is smaller than the latter. Hence, both the solutions behave in an LISS manner.

%===============================================================
%===============================================================
\section{Conclusion}
\label{sec:conclusions}
%===============================================================
%===============================================================

This paper presents a method for analyzing robustness of limit cycles exhibited by systems with impulse effects. It is shown that ISS of the limit cycle is equivalent to that of the forced Poincar\'e map. This result allows us to analyze the robustness of \emph{hybrid} limit cycles by merely analyzing a discrete dynamical system. The proof of this result, provides ISS estimates that could be used to quantify the robustness. Furthermore, exploiting the availability of these estimates, we establish an equivalence between ES of the limit cycle and the 0-input Poincar\'e map. The overarching goal of this work is to develop a framework within which the robustness of periodic orbits can be rigorously analyzed.

%===============================================================
%===============================================================
\appendices
%===============================================================
%===============================================================

%===============================================================
\section{}\label{app:existence}
%===============================================================

This appendix provides proofs to Lemmas~\ref{lem:frechet} and~\ref{lem:TI-cont} and to Proposition~\ref{prop:global-existence}, clarifying properties of the forced solution $\varphi(t,x(0),u)$ of \eqref{eq:cont-dyn} and $\psi(t,x(0),u,\bar{v})$ of \eqref{eq:sys-imp-eff}. 
 
%proof
\begin{proof}[Proof of Lemma~\ref{lem:frechet}]
The proof relies on \cite[Theorem~5.2,~p.~377]{lang1993real}. To apply this result, define $F:\mathbb{R}_+\times\mathbb{R}^n\times\mathcal{U} \to \mathbb{R}^n$ as $F(t,x,u):=f(x,G(t,u))$, where $G:\mathbb{R}_+\times\mathcal{U} \to \mathbb{R}^p$ is $G(t,u)=u(t)$. Then,~\cite[Theorem~5.2,~p.~377]{lang1993real}  states that $\varphi$ is continuously differentiable in $x$ and $u$ if $F$ is so. Note that, in view of the procedure in~\cite[p. 369]{lang1993real} for treating time-dependent vector fields, the statement of~\cite[Theorem~5.2,~p.~377]{lang1993real} would require $F$ to be continuously differentiable in $t$. However, the proof of~\cite[Theorem~5.2,~p.~377]{lang1993real} \emph{only} uses continuity of $F$ in $t$. Hence, below we show that $F$ is continuously differentiable in $x$ and $u$, but only continuous in $t$.    

We begin by showing that $G$ is continuous in $t$ and continuously differentiable in $u$. Continuity of $G$ in $t$ is clear. Continuity of $G$ in $u \in \mathcal{U}$ follows from the fact that $\hat{G}: \mathcal{U} \to \mathbb{R}^p$ defined by $\hat{G}(u):=G(t,u)=u(t)$ is bounded and linear for each \emph{fixed} time $t \in \mathbb{R}_+$; by \cite[p.~257,~Theorem~1]{royden2010real} this implies that $\hat{G}$ is continuous in $u$, thus $G$ is also continuous in $u$. Indeed, linearity is immediate by the definition of $\hat{G}$, while boundedness follows from $\|\hat{G}(u)\|=\|u(t)\|\leq \|u\|_\infty$, implying that the operator norm is upper bounded by 1 for any $t\in\mathbb{R}_+$. As a result, $G(t,u)$ is continuous in both arguments. 
Furthermore, $G$ is linear with respect to $u$, and using \cite[p.~339,~Theorem~3.1]{lang1993real}, we have that the Fr\'echet (partial) derivative of $G$ with respect to $u$ is continuous, for it is $G$ itself. Thus, $G(t,u)$ is continuous in $t$ and continuously differentiable in $u$. Using this fact with the assumption that $f$ is continuously differentiable, it follows that $F$ is continuous in $t$ and continuously differentiable in $(x,u)$. 
The result then follows from~\cite[Theorem~5.2,~p.~377]{lang1993real}. 
%, noting  that continuous differentiability  \textcolor{blue}{of $F$ in $t$} can be relaxed to continuity in $t$
\end{proof}
%proof

%proof
\begin{proof}[Proof of Lemma~\ref{lem:TI-cont}]
Before proceeding with the proof, note that even though the domain of $T_{\rm I}$ is restricted to $\mathcal{S} \times \mathcal{U} \times \mathbb{R}^q$, $T_{\rm I}$ is well-defined on $\mathbb{R}^n\times \mathcal{U} \times \mathbb{R}^q$ since $\Delta$ and $\varphi$ are well-defined maps for any $x\in\mathbb{R}^n$. Hence, we will consider this extended domain of $T_{\rm I}$ in the proof, which follows from the implicit mapping theorem \cite[Chapter~XIV,~Theorem~2.1]{lang1993real}. 

Let $\tilde{H}(t,x,u,v):=H\circ\varphi(t,\Delta(x,v),u)$ where $x\in\mathbb{R}^n$, $u\in\mathcal{U}$, and $v\in\mathbb{R}^q$. From Lemma~\ref{lem:frechet} in Appendix~\ref{app:existence}, we have that the solution $\varphi$ is continuously differentiable in all its arguments in the Fr\'echet sense. Using this with assumption~A.\ref{ass:S} we have that $\tilde{H}$ is continuously differentiable. From assumption~A.\ref{ass:orbit-exist} it follows that $\tilde{H}(T^*,x^*,0,0)=0$. Further, from assumption~A.\ref{ass:orbit-transversal} we have $\partial \tilde{H}/\partial t|_{(T^*,x^*,0,0)}\neq 0$. Next, noting that $(\mathbb{R}^n\times\mathbb{R}^q,\|\cdot\|)$ and $(\mathcal{U},\|\cdot\|_\infty)$ are Banach spaces, we can use \cite[Chapter~XIV,~Theorem~2.1]{lang1993real} to establish the existence of a unique map $T_{\rm I}(x,u,v)$ which satisfies $\tilde{H}(T_{\rm I}(x,u,v),x,u,v)=0$ for a sufficiently small $\delta>0$ such that $x\in B_\delta(x^*)$, $u\in\mathcal{U}$ with $\|u\|_\infty<\delta$, and $v\in B_\delta(0)$. Additionally, since $\tilde{H}(t,x,u,v)$ is continuously differentiable with respect to its arguments, so is $T_{\rm I}(x,u,v)$. As this holds for any $x\in B_\delta(x^*)$, it also holds for any $x\in B_\delta(x^*)\cap\mathcal{S}$.
\end{proof}
%proof

%Next, we turn our attention in the solution $\psi(t,x(0),u,\bar{v})$ of \eqref{eq:sys-imp-eff}, and we provide a proof of Proposition~\ref{prop:global-existence} in Section~\ref{subsec:defs}.
%which establishes certain properties of interest  regarding the hybrid solutions $\psi$ of \eqref{eq:sys-imp-eff} that are in a sufficiently small neighborhood of a LISS periodic orbit $\mathcal{O}$.
%proof
\begin{proof}[Proof of Proposition~\ref{prop:global-existence}]
We restrict attention to initial conditions $x(0)$ that result in solutions that hit $\mathcal{S}$;  otherwise the solution is continuous and it can be extended indefinitely, trivially excluding the occurrence of Zeno and beating phenomena.

We first show parts \emph{(i)} and \emph{(ii)} simultaneously. 
Let $x(t) = \psi(t, x(0), u, \bar{v})$ be a solution of \eqref{eq:sys-imp-eff} that is defined over some interval $[0, t_\mathrm{f})$ and satisfies Definition~\ref{def:orbit-LISS}. As in Remark~\ref{rem:hyb-cont-sol-rel}, let $t_k$ and $t_{k+1}$ denote two subsequent impact times so  that $t_{k+1} - t_k = T_{\rm I}(x_k, u_k, v_k)$, where $x_k := \lim_{t \nearrow t_k} x(t) \in \mathcal{S}$, $u_k(t) =u(t)$ for $t \in [t_k, t_{k+1})$, and $v_k$ is the $k$-th element of the sequence $\bar{v}$. We will show that there exists $\underline{T} > 0$ such that $t_{k+1} - t_k > \underline{T}$ for all $k \in \mathbb{Z}_+$ with $[t_k, t_{k+1}) \subset [0, t_\mathrm{f})$. 
By Definition~\ref{def:orbit-LISS}, there exists $\delta>0$ so that $\mathrm{dist}(x(0),\mathcal{O})<\delta$, $u\in\mathcal{U}$ with $\|u\|_\infty<\delta$, and $\bar{v}\in\mathcal{V}$ with $\|\bar{v}\|_\infty<\delta$ imply that $x(t)$ satisfies \eqref{eq:orbit-LISS-ineq}. By properties of class-$\mathcal{KL}$ and class-$\mathcal{K}$ functions (see Section~\ref{subsec:notation}) we have
%eq
\begin{align}\label{eq:trapping-compact-set}
\mathrm{dist}(x(t),\mathcal{O}) \leq \beta(\delta,0)+\alpha_1(\delta) + \alpha_2(\delta) \enspace,
\end{align}
%eq
for all $t\in [0,t_\mathrm{f})$. Furthermore, by continuity of the distance function, we have $\lim_{t \nearrow t_k} \mathrm{dist}(x(t),\mathcal{O}) = \mathrm{dist}(x_k,\mathcal{O})$, and \eqref{eq:trapping-compact-set} in view of Proposition~\ref{prop:equiv-norm} implies that, for some $\lambda \in (0,1)$,
%; clearly, by properties of class-$\mathcal{KL}$ and class-$\mathcal{K}$ functions (see Section~\ref{subsec:notation}), the upper bound in \eqref{eq:trapping-compact-set} is monotonically increasing in $\delta$. Hence, the solution $x(t)$ can be ``trapped'' arbitrarily close to $\mathcal{O}$ by choosing a sufficiently small $\delta$
%eq
\begin{equation}\label{eq:trapping-compact-set-k}
\|x_k-x^*\| \leq \frac{1}{\lambda} \left( \beta(\delta,0)+\alpha_1(\delta) + \alpha_2(\delta) \right) \enspace,
\end{equation}
%eq 
for all $k \in \mathbb{Z}_+$ with $[t_k, t_{k+1}) \subset [0, t_\mathrm{f})$.
The result now follows from continuity of the time-to-impact function by Lemma~\ref{lem:TI-cont}. Indeed, as in the proof of part \emph{(ii)} of Lemma~\ref{rem:unperturbed}, continuity of $T_{\rm I}$ implies that, for some $\underline{T}>0$, there exists a $\delta_T>0$ such that $x\in B_{\delta_T}(x^*)\cap\mathcal{S}$, $u\in\mathcal{U}$ with $\|u\|_\infty<\delta_T$, and $v\in B_{\delta_T}(0)$ imply $\underline{T} < T_{\rm I}(x,u,v)$. Choosing $\delta$ in \eqref{eq:trapping-compact-set-k} so that $\frac{1}{\lambda} \left( \beta(\delta,0)+\alpha_1(\delta) + \alpha_2(\delta) \right) < \delta_T$ ensures that $x_k\in B_{\delta_T}(x^*)\cap\mathcal{S}$ for all $k \in \mathbb{Z}_+$ with $[t_k, t_{k+1}) \subset [0, t_\mathrm{f})$. Such choice of $\delta$ is always possible since the upper bound in \eqref{eq:trapping-compact-set-k} is a class-$\mathcal{K}$ function of $\delta$. Shrinking $\delta$ further (if necessary) to ensure that $\delta < \delta_T$ guarantees that for $x_k\in B_{\delta}(x^*)\cap\mathcal{S}$, $u\in\mathcal{U}$ with $\|u\|_\infty<\delta$, and $v\in B_{\delta}(0)$ we have that $T_{\rm I}(x_k,u_k,v_k)> \underline{T}$ for all discrete events $k$ of the solution. As a result, $t_{k+1} - t_k > \underline{T}$ for all $k \in \mathbb{Z}_+$ with $[t_k, t_{k+1}) \subset [0, t_\mathrm{f})$. This ensures that any two discrete events are punctuated by a time gap of $\underline{T}$, thereby precluding solutions that are purely discrete or eventually discrete, or exhibit Zeno behavior, completing the proof of \emph{(i)} and \emph{(ii)}.

To prove part \emph{(iii)}, from \eqref{eq:trapping-compact-set} it is clear that $x(t)$ is trapped in a compact set. Neglecting -- as was mentioned at the beginning of the proof -- the trivial case where the continuous solution never approaches $\mathcal{S}$, using arguments similar to the proof of \cite[Theorem~3.3]{khalil2002nonlinear} the solution can be extended until it reaches $\mathcal{S}$. At this point, a well-defined discrete jump occurs that ensures the post-discrete-event state is still trapped within the same compact set and lies outside $\mathcal{S}$ because of part \emph{(i)}; hence, the solution must flow again according to the continuous dynamics until it reaches $\mathcal{S}$. Additionally \emph{(ii)} ensures the absence of Zeno behavior. Hence, we can propagate this argument forward for all time to obtain \emph{(iii)}.
%, or it can be extended indefinitely if it never reaches $\mathcal{S}$. In the \textcolor{blue}{latter} case the proof of part \emph{(ii)} is complete. Now, if the solution does reach $\mathcal{S}$
\end{proof}
%proof

%===============================================================
\section{}\label{app:y-tau}
\begin{proof}[Proof of Lemma~\ref{lem:y-tau}]
%===============================================================

By assumption A.\ref{ass:f-c2}, $f(x,0)$ is twice continuously differentiable; thus, its backward flow $\varphi^-(s\tau,x^*,0)=y(\tau)$ is three-times continuously differentiable. Surjectivity of $y(\tau)$ is obvious from the definition of $y(\tau)$ in \eqref{eq:y-tau-def}. Injectivity follows from a contradiction argument. Assume $y(\tau)$ is not injective, then there exist $\tau_1 < \tau_2$ in $[0,T]$ with $y(\tau_1) = y(\tau_2)$. Let $f^-(x,0):=-f(x,0)$ be the vector field for the backwards flow. If $f^-(y(\tau_1),0)\neq f^-(y(\tau_2),0)$ then $f$ would not be well defined. If $f^-(y(\tau_1),0)= f^-(y(\tau_2),0)$, we have:\\
\noindent \emph{Case (i):} $0<\tau_1<\tau_2$ \\
Since $y(\tau_1)=y(\tau_2)$, we return to the same state after an interval $\tau_2-\tau_1>0$. Hence, $E:=\{y(\tau)~|~\tau_1\leq\tau\leq\tau_2\}\subset \overline{\mathcal{O}}$ is a periodic orbit of the backwards-flow continuous system and we have the following sub-cases:
\begin{enumerate}
\item[\emph{(a)}] $y(\tau)\in E$ for all $\tau\in[\tau_1,T]$:
The orbit $E$ would also exist in the forward flow. As the periodic orbit is an invariant set under the 0-input continuous dynamics \eqref{eq:cont-dyn}, the forward flow starting from $y(T)=\Delta(x^*,0)$ will be trapped in $E$ and never reach $\mathcal{S}$, contradicting assumption~A.\ref{ass:orbit-exist}, according to which the solution must reach $\mathcal{S}$ in finite time $T^*$.
% as well with the exception that the flow will be in the opposite direction to the backward flow
\item[\emph{(b)}] There exists $\hat{\tau}\in[\tau_1,T]$ such that $y(\hat{\tau})\not\in E$:
This contradicts uniqueness of the backwards solution, as starting from $y(\tau_1)\in E$, one solution flows to $y(\hat{\tau})\not\in E$ while the other gets trapped in $E$.
\end{enumerate}
\noindent \emph{Case (ii):} $0=\tau_1<\tau_2$ \\
Note that $H(y(\tau))$ and $L_{f^-} H(y(\tau,0))$ are continuous in $\tau$. Additionally, from assumption~A.\ref{ass:x*}-A.\ref{ass:orbit-transversal}, $H(y(0))=H(x^*)=0$ and $L_{f^-} H(y(0),0)>0$ (flipped sign from assumption A.\ref{ass:orbit-transversal} due to the flow being backwards in time), thus there exists a $\delta>0$ such that $L_{f^-} H(y(\tau),0)>0$ for all $\tau\in[0,\delta)$. Hence for the interval $(0,\delta)$ we have $H(y(\tau))>0$, i.e., $\{y(\tau) ~|~ \tau\in(0,\delta)\}\subset \mathcal{S}^+$. Again using continuity of $H(y(\tau))$ and $L_{f^-} H(y(\tau),0)$ at $\tau_2$ we have that $H(y(\tau))$ is strictly increasing in the interval $(\tau_2-\delta,\tau_2+\delta)$ (shrink $\delta>0$ if necessary to ensure $\delta<\tau_2-\delta$) but $H(y(\tau_2))=0$, hence, $H(y(\tau))<0$ for all $\tau\in (\tau_2-\delta,\tau_2)$, i.e., $\{y(\tau)~|~\tau\in(\tau_2-\delta,\tau_2) \}\subset\mathcal{S}^-$. Thus, for some $\bar{\tau}$ such that $\delta<\bar{\tau}<\tau_2-\delta$, the solution must cross over from $\mathcal{S}^+$ to $\mathcal{S}^-$ at a point other than $x^*$, resulting in a contradiction to assumption A.\ref{ass:x*}.
\end{proof}

%===============================================================
\section{}\label{app:S-projection}
%===============================================================

First note that $\mathcal{S}$ is a twice continuously differentiable embedded submanifold in $\mathbb{R}^n$ defined by $H(x)=0$. Clearly, $H(x^*)=0$, and, without loss of generality, assume that for the $n$-th coordinate $\frac{\partial H}{\partial x_n}\big|_{x^*} \neq 0$, which follows from assumption~A.\ref{ass:S}. Hence, using the implicit function theorem we can write $x_n = h(\bar{x})$ where $\bar{x}:=(x_1,...,x_{n-1})$ for $\bar{x} \in B_\delta(\bar{x}^*)$ where $x^*=(\bar{x}^*,h(\bar{x}^*))$. As a result, the local coordinates of states $x\in\mathcal{S}$ in a neighborhood of $x^*$ are $x=(\bar{x},h(\bar{x}))$. The Taylor expansion of $h(\bar{x})$ at $\bar{x}^*$ gives
%eq
\begin{equation}\label{eq:S-taylor}
h(\bar{x}) = h(\bar{x}^*) +  A (\bar{x}-\bar{x}^*) + O(\|\bar{x}-\bar{x}^*\|^2) \enspace,
\end{equation}
%eq
where $A= \frac{\partial h}{\partial \bar{x}}\big|_{\bar{x}^*}$.
Let $z\in T_{x^*} \mathcal{S}$ be the projection of $x$ on $T_{x^*}\mathcal{S}$ along $(x_1,x_2,...,x_{n-1})$, then its coordinates are $z=(\bar{x},h(\bar{x}^*)+A(\bar{x}-\bar{x}^*))$. Hence, for $x\in B_\delta(x^*)\cap\mathcal{S}$ and $z\in T_{x^*}\mathcal{S}$, there exists a $c>0$ such that,
%eq
\begin{align}
\|x-z\| & = \|(\bar{x},h(\bar{x})) - (\bar{x},h(\bar{x}^*)+A(\bar{x}-\bar{x}^*))\| \nonumber \\
 & = |O(\|\bar{x}-\bar{x}^*\|^2)| \leq c \|\bar{x}-\bar{x}^*\|^2 \leq c \|x-x^*\|^2 \enspace. \nonumber
\end{align}
%eq

%===============================================================
\section{}\label{app:T-I}
\begin{proof}[Proof of Lemma~\ref{lem:solution-compare-S}]
%===============================================================

We begin with part \emph{(i)}. By Remark~\ref{rem:unperturbed}, $\varphi(t,\Delta(x^*,0),0)$ exists and is unique over $[0, \overline{T}]$ with $\overline{T}>T^*$. Lemma~\ref{lem:cont-solutions} establishes the existence of $\delta_1>0$ for which, when $x \in \mathcal{S}$ and $v \in \mathbb{R}^q$ are such that $\|\Delta(x,v)-\Delta(x^*,0)\| < \delta_1$, and when $\|u\|_\infty<\delta_1$, the perturbed solution $\varphi(t,\Delta(x,v),u)$ exists and is unique over the same interval $[0,\overline{T}]$.      
By the continuity of $\Delta$ following from assumption A.\ref{ass:Delta}, there  exists a $\delta_2>0$ for which $\|x-x^*\| < \delta_2$ and $\| v \| < \delta_2$ guarantee $\|\Delta(x,v)-\Delta(x^*,0)\| < \delta_1$. Hence, choosing $\delta=\min\{\delta_1,\delta_2\}$ we have  that, for $x \in B_\delta(x^*) \cap \mathcal{S}$, $u \in \mathcal{U}$ with $\|u\|_\infty<\delta$ and $\bar{v} \in \mathcal{V}$ with $\| \bar{v} \|_\infty < \delta$, the perturbed solution $\varphi(t,\Delta(x,v),u)$ of \eqref{eq:cont-dyn} exists and is unique over $[0, \overline{T}]$, thus proving part \emph{(i)}.       

%%%%NEW:
To prove part \emph{(ii)}, let $\epsilon_T:=\overline{T}-T^*>0$. By continuity of $T_{\rm I}$ there exists a $\delta_T>0$ such that for $x\in B_{\delta_T}(x^*) \cap \mathcal{S}$, $u\in\mathcal{U}$ with $\|u\|_\infty<\delta_T$, and $v\in B_{\delta_T}(0)$, we have that $| T_{\rm I}(x,u,v) - T_{\rm I}(x^*,0,0) | < \epsilon_T$, which implies that $\underline{T} < T_{\rm I}(x,u,v) < \overline{T}$ where $\underline{T}:=T^*-\epsilon_T=2T^*-\overline{T}>0$ (shrink $\overline{T}$ if necessary to ensure that $\overline{T} < 2T^*$). To show transversality, define $\ell(x,u,v):= \frac{\partial H}{\partial x}\big|_{\varphi(T_\mathrm{I}(x,u,v),\Delta(x,v),u)} f(\varphi(T_\mathrm{I}(x,u,v),\Delta(x,v),u), u)$ so that $\ell^*:=\ell(x^*,0,0)<0$ by assumption A.\ref{ass:orbit-transversal}, and let $\epsilon_\ell \in(0, -\ell^*)$. By continuity of $\ell$, there is a $\delta_\ell>0$ such that for $x\in B_{\delta_\ell}(x^*) \cap \mathcal{S}$, $u\in\mathcal{U}$ with $\|u\|_\infty<\delta_\ell$, and $v\in B_{\delta_\ell}(0)$, we have that $| \ell(x,u,v) - \ell(x^*,0,0) | < \epsilon_\ell$, implying $\ell(x,u,v) < \epsilon_\ell + \ell^* <0$. Choosing $\delta = \min\{\delta_1,\delta_2, \delta_T, \delta_\ell\}$ with $\delta_1,\delta_2>0$ as in the proof of part \emph{(i)} completes the proof of part \emph{(ii)}. 
%%L_f H(\varphi(T_\mathrm{I}(x,u,v),\Delta(x,v),u),u)
%Hence, choosing $\delta = \min\{\delta_1,\delta_2, \delta_T, \delta_\ell\}$} with $\delta_1,\delta_2>0$ as in the proof of part \emph{(i)} implies both that $\varphi(t,\Delta(x,v),u)$ is well defined over  $[0,\overline{T}]$ and that it crosses $\mathcal{S}$ transversally in finite time $T_{\rm I}(x,u,v)$ satisfying $\underline{T} < T_{\rm I}(x,u,v) < \overline{T}$. Hence, part \emph{(ii)} is proved. 
%%%%%%%%%%%%%%%%%%%OLDER VERSION!!!! (NO TRANSVERSALITY)
%To prove part \emph{(ii)}, let $\epsilon_T:=\overline{T}-T^*>0$. By continuity of $T_{\rm I}$ there exists a $\delta_T>0$ such that for $x\in B_{\delta_T}(x^*) \cap \mathcal{S}$, $u\in\mathcal{U}$ with $\|u\|_\infty<\delta_T$, and $v\in B_{\delta_T}(0)$, we have that $| T_{\rm I}(x,u,v) - T_{\rm I}(x^*,0,0) | < \epsilon_T$, which implies that $\underline{T} < T_{\rm I}(x,u,v) < \overline{T}$ where $\underline{T}=T^*-\epsilon_T=2T^*-\overline{T}>0$ (shrink $\overline{T}$ if necessary to ensure that $\overline{T} < 2T^*$). Choosing $\delta = \min\{\delta_1,\delta_2, \delta_T\}$ with $\delta_1,\delta_2>0$ chosen as in the proof of part \emph{(i)} implies both that $\varphi(t,\Delta(x,v),u)$ is well defined over  $[0,\overline{T}]$ and that it crosses $\mathcal{S}$ in finite time $T_{\rm I}(x,u,v)$ satisfying $\underline{T} < T_{\rm I}(x,u,v) < \overline{T}$. Hence, part \emph{(ii)} is proved. 

%%%%NEW: 
Finally, for part \emph{(iii)}, we begin by setting up the region within which we work. Let $w(t)=\varphi(t,\Delta(x^*,0),0)$ be the unperturbed (zero-input) solution of \eqref{eq:cont-dyn}, which by Remark~\ref{rem:unperturbed} is well-defined over the interval $[0, \overline{T}]$ with $\overline{T}>T^*$. Let $\delta>0$ be as in the proof of part \emph{(ii)} so that, for $x \in B_\delta(x^*) \cap \mathcal{S}$, $u \in \mathcal{U}$ with $\|u\|_\infty<\delta$ and $v \in B_{\delta}(0)$, the perturbed solution $\varphi(t,\Delta(x,v),u)$ exists and is unique over $[0, \overline{T}]$ and $T_{\rm I}(x,u,v) < \overline{T}$. Furthermore, for this choice of $\delta$, Lemma~\ref{lem:cont-solutions} guarantees the existence of $L>0$ such that
%eq
%\vspace{-3mm}
%\small
\begin{align}\nonumber%\label{eq:lipschitz-solution}
\|\varphi(t,\Delta(x,v),u) - w(t)\| & \leq \mathrm{e}^{L\overline{T}} \|\Delta(x,v) - \Delta(x^*,0)\| \\ &+( \mathrm{e}^{L\overline{T}}-1) \|u\|_\infty , \nonumber
\end{align}
%eq
%\normalsize%\small
for all $t\in[0,\overline{T}]$. Since $\Delta$ is continuously differentiable from assumption A.\ref{ass:delta-C-1}, it is locally Lipschitz. Hence, shrinking $\delta$ (if necessary) ensures that the Lipschitz condition can be satisfied uniformly over $B_\delta(x^*) \cap \mathcal{S}$ and $\| \bar{v} \|_\infty < \delta$ for some constant $L_\Delta > 0$, thereby resulting in the following estimate
%eq
\begin{align}
 \sup_{0\leq t \leq \overline{T}}& \|\varphi(t,\Delta(x,v),u) - w(t)\| \nonumber \\
 & \!\!\! \leq \! \mathrm{e}^{L\overline{T}} L_\Delta (\| x- x^*\| + \| \bar{v} \|_\infty)+( \mathrm{e}^{L\overline{T}}-1) \|u\|_\infty . \label{eq:lipschitz-solution-bound}
\end{align}
%eq
%which holds for all $t\in[0,\overline{T}]$.      
In addition, since by Lemma~\ref{lem:TI-cont} the map $T_{\rm I}$ is continuously differentiable and thus locally Lipschitz. Shrinking $\delta$ further (if necessary) ensures that there exists a $L_T >0$ such that    
%eq
\begin{equation}\label{eq:lipschitz-time-bound}
|T_{\rm I}(x,u,v)- T_{\rm I}(x^*,0,0)| \leq L_T (\| x- x^*\| +  \|u\|_\infty + \| \bar{v} \|_\infty)
\end{equation}
%eq
for all $x \in B_\delta(x^*) \cap \mathcal{S}$, $u \in \mathcal{U}$ with $\| u \|_\infty < \delta$ and $\bar{v} \in \mathcal{V}$ with $\| \bar{v} \|_\infty < \delta$. In what follows, we work in such region so that the perturbed solution $\varphi(t,\Delta(x,v),u)$ and the corresponding time to impact $T_{\rm I}(x,u,v)$ satisfy \eqref{eq:lipschitz-solution-bound} and \eqref{eq:lipschitz-time-bound}, respectively.   

Now, let us consider the distance of the perturbed solution $\varphi(t,\Delta(x,v),u)$ from $\mathcal{O}$. For any $t \in [0, \overline{T}]$ we have 
%eq
\begin{align}
 \mathrm{dist}(\varphi &(t,\Delta(x,v),u),\mathcal{O}) \nonumber \\
 & :=\inf_{y\in\mathcal{O}}\|\varphi(t,\Delta(x,v),u)-w(t)+w(t)-y\| \nonumber \\
 & \leq \|\varphi(t,\Delta(x,v),u)-w(t)\| + \mathrm{dist}(w(t),\mathcal{O}) \label{eq:dist-orb-sol} \enspace,
\end{align}
%eq
where \eqref{eq:dist-orb-sol} is obtained by the triangle inequality. Regarding the term $\mathrm{dist}(w(t),\mathcal{O})$ in \eqref{eq:dist-orb-sol}, note that when $t\in[0,T^*)$ we have $\mathrm{dist}(w(t),\mathcal{O})=0$, while for $t\in[T^*,\overline{T}]$ we have
%eq
\begin{align}
{\rm dist}(w(t), \mathcal{O}) &:= \inf_{y \in \mathcal{O}} ||w(t) - x^* + x^* - y || \nonumber \\
&\leq ||w(t) - x^*|| + \inf_{y \in \mathcal{O}} ||x^* - y || \nonumber \\
&= ||w(t) - x^*|| \label{eq:dist-w-O}\enspace,  
\end{align} 
%eq
since $\inf_{y \in \mathcal{O}} ||x^* - y || = 0$. We distinguish the following cases.

\noindent \emph{Case (a):} Assume that $T_\mathrm{I}(x,v,u) < T^*$. Then, for all $t\in[0,T_\mathrm{I}(x,v,u)]$, we have $\mathrm{dist}(w(t),\mathcal{O})=0$ and application of $\sup_{0\leq t < T_{\rm I}}$ on \eqref{eq:dist-orb-sol} results in the bound
%eq
\begin{align}
\sup_{0\leq t < T_{\rm I}} & \mathrm{dist} (\varphi (t,\Delta(x,v),u),\mathcal{O}) \nonumber \\ &\leq \sup_{0 \leq t < T_{\rm I}} \|\varphi(t,\Delta(x,v),u) -w(t)\| \enspace. \label{eq:compare-00}
\end{align}
%eq
Since $T^* < \overline{T}$, for this case we have $[0,T_{\rm I}(x,u,v))\subset [0,\overline{T}]$, and the bound \eqref{eq:lipschitz-solution-bound} implies
%eq
\begin{align}
\sup_{0\leq t < T_{\rm I}} & \mathrm{dist}  (\varphi (t,\Delta(x,v),u),\mathcal{O}) \nonumber \\ 
& \leq  \mathrm{e}^{L\overline{T}} L_\Delta (\| x \!-\! x^*\| \!+\! \| \bar{v} \|_\infty) +( \mathrm{e}^{L\overline{T}} \!\! -\!\! 1) \|u\|_\infty ~\!. \nonumber%\label{eq:case-a-bound}
\end{align}
%eq
 
\noindent \emph{Case (b):} Assume that $T_\mathrm{I}(x,v,u) \geq T^*$. Then,  \eqref{eq:dist-w-O} implies
%eq
\begin{align}
\sup_{0\leq t < T_{\rm I}}  \mathrm{dist}(w(t),\mathcal{O}) & \leq \sup_{T^*\leq t < T_{\rm I}} \|w(t)-x^*\|  \enspace. \label{eq:orbit-extra-flow}
\end{align}
%eq
Applying $\sup_{0\leq t < T_{\rm I}}$ on \eqref{eq:dist-orb-sol} followed by  \eqref{eq:orbit-extra-flow} results in
%eq
\begin{align}
\sup_{0\leq t < T_{\rm I}} & \mathrm{dist} (\varphi (t,\Delta(x,v),u),\mathcal{O}) \nonumber \\
&  \leq \sup_{0 \leq t < T_{\rm I}} \|\varphi(t,\Delta(x,v),u) -w(t)\| \nonumber \\
& + \sup_{T^* \leq t < T_{\rm I}} \|w(t)-x^*\| \enspace. \label{eq:compare-0}
\end{align}
%eq
Regarding the first term in the RHS of \eqref{eq:compare-0}, the upper bound \eqref{eq:lipschitz-solution-bound} can be used since $[0, T_\mathrm{I}) \subset [0, \overline{T}]$. Next, we look at the second term. Let $K:=\max_{T^*\leq t \leq \overline{T}} \|f(w(t),0)\|$. As mentioned in Remark~\ref{rem:unperturbed}, we have $L_f H(w(t),0)<0$ for all $T^* \leq t\leq \overline{T}$ which implies that $f(w(t),0)\neq 0$ for all $t\in[T^*,\overline{T}]$, hence $K>0$. Use $K$ in the following,
%eq
\begin{align}
\sup_{T^* \leq t <T_{\rm I}}  & \|w(t)-x^*\| = \sup_{T^* \leq t < T_{\rm I}} \Big\|\int_{T^*}^{t} f(w(s),0)~ds \Big\| \nonumber \\
& \leq K  |T_{\rm I}(x,u,v)- T_{\rm I}(x^*,0,0)|  \nonumber \\
& \leq K L_T (\|x-x^*\| + \|u\|_\infty + \|\bar{v}\|_\infty) \enspace, \label{eq:TI-compare}
\end{align}
%eq 
where \eqref{eq:lipschitz-time-bound} was used. The desired estimate in Case \emph{(b)} is then found by combining \eqref{eq:lipschitz-solution-bound} and \eqref{eq:TI-compare} according to \eqref{eq:compare-0}.

Finally, Case \emph{(a)} and \emph{(b)} can be combined for an appropriate $c>0$ to obtain \eqref{eq:compare-S}, thereby completing the proof.
\end{proof}

\begin{proof}[Proof of Lemma~\ref{lem:solution-compare-S+}] We first prove Lemma~\ref{lem:solution-compare-S+}\emph{(i), (ii)} simultaneously. We begin by using Lemma~\ref{lem:cont-solutions} to show that, for any point of $\mathcal{O}$, there exists an open ball of initial conditions and inputs for which a unique (forced) solution exists over a time interval that is sufficiently long to cross $\mathcal{S}$ transversally in finite time. Then, compactness of $\overline{\mathcal{O}}$ is used to extract a finite open subcover of $\overline{\mathcal{O}}$, based on which an upper bound on the time to cross $\mathcal{S}$ is established uniformly in a tube around $\mathcal{O}$. 
%The proof follows from Lemma~\ref{lem:cont-solutions} and the fact that $\overline{\mathcal{O}}$ is compact.

Let $w(t)=\varphi(t,\Delta(x^*,0),0)$ be the unperturbed (zero-input) solution of \eqref{eq:cont-dyn} starting from $\Delta(x^*,0)$. By Remark~\ref{rem:unperturbed}, $w(t)$ is well-defined over an  interval $[0, \overline{T}]$ with $\overline{T}>T^*$. Then define $w_{\rm e}:=w(\overline{T})$ which is in the open set $\mathcal{S}^-$. Hence, there exists a $\delta_{\rm e}>0$ such that $B_{\delta_{\rm e}}(w_{\rm e})\subset\mathcal{S}^-$. Given an arbitrary time $\xi\in[0,T^*]$ and the corresponding point  $w(\xi) \in \overline{\mathcal{O}}$, Lemma~\ref{lem:cont-solutions} ensures that there exists a $\delta_\xi>0$ such that, for any $x\in B_{\delta_\xi}(w(\xi))$ and $u\in\mathcal{U}$ with $\|u\|_\infty<\delta_\xi$, the (perturbed) solution $\varphi(t,x,u)$ exists and is unique over the interval $t\in[0,\overline{T}-\xi]$. Furthermore, for some $L>0$, \eqref{eq:lipschitz-solution} of Lemma~\ref{lem:cont-solutions} gives $\|\varphi(t,x,u)-w(\xi+t)\| \leq (2 e^{L(\overline{T}-\xi)} - 1)\delta_\xi$ for all $t\in[0,\overline{T}-\xi]$, which implies that $\delta_\xi$ can be chosen sufficiently small to ensure that $\sup_{0\leq t\leq \overline{T}-\xi}\|\varphi(t,x,u)-w(t+\xi)\|<\delta_{\rm e}$. 
As a result, $\varphi(\overline{T}-\xi,x,u) \in B_{\delta_{\rm e}}(w_{\rm e})$, from which it follows that the perturbed solution has crossed $\mathcal{S}$ during the interval $[0,\overline{T}_x]$ with $\overline{T}_x := \overline{T}-\xi$; clearly, $\overline{T}_x$ depends on the choice of $w(\xi)$. Finally, the fact that the solution crosses $\mathcal{S}$ transversally can be shown similarly to the proof of part \emph{(ii)} of Lemma~\ref{lem:solution-compare-S}.
 %Finally, let $\hat{T}_\mathrm{I}(x,u) \in [0,\overline{T}_x]$ denote the first time instant at which $\varphi(t,x,u)$ crosses $\mathcal{S}$, as defined in \eqref{eq:time-to-imp-S+}. Let $\hat{\ell}(x,u):= \frac{\partial H}{\partial x}\big|_{\varphi(\hat{T}_\mathrm{I}(x,u),x,u)} f(\varphi(\hat{T}_\mathrm{I}(x,u),x,u),u)$ and note that by time-invariance of the unforced system \eqref{eq:cont-dyn} and assumption A.\ref{ass:orbit-transversal} we have $\ell^* = \hat{\ell}(w(\xi),0) = \frac{\partial H}{\partial x}\big|_{x^*} f(x^*,0) < 0$. Similarly to the proof of Lemma~\ref{lem:solution-compare-S}, invoking continuity of $\ell$, shrink $\delta_\xi$ if necessary to ensure that $\hat{\ell}(x,u) < 0$, implying that the solution crosses $\mathcal{S}$ transversally.}  

Now, to extend these results over the entire orbit, construct an open cover of $\overline{\mathcal{O}}$ by using the open balls $B_{\delta_\xi}(w(\xi))$ for each $\xi\in[0,T^*]$. As $\overline{\mathcal{O}}$ is compact, there exists $\xi_1,\xi_2,...,\xi_N\in[0,T^*]$, and $\delta_i := \delta_{\xi_i}$, $w_i:=w(\xi_i)$ such that $\cup_{i=1}^N B_{\delta_i}(w_i) \supset \overline{\mathcal{O}}$ is a finite sub-cover. Choose $0<\delta<\min_{1 \leq i\leq N} \delta_i$ such that $\{x\in\mathbb{R}^n~|~\mathrm{dist}(x,\mathcal{O})<\delta\}\subset \cup_{i=1}^N B_{\delta_i}(w_i)$. Choosing such $\delta>0$ is always possible. Indeed, since $\cup_{i=1}^N B_{\delta_i}(w_i)$ is bounded, its boundary $\mathcal{B}:=\partial (\cup_{i=1}^N B_{\delta_i}(w_i))$ is compact. Let $0<\delta<\min_{x\in \mathcal{B}} \mathrm{dist}(x,\mathcal{O})/2$ (shrink if necessary to ensure $\delta<\min_{1 \leq i\leq N} \delta_i$) where the $\min$ is well-defined because $\mathrm{dist}(x,\mathcal{O})$ is continuous. Note that $\{x\in\mathbb{R}^n~|~\mathrm{dist}(x,\mathcal{O})<\delta\}$ is connected because $\mathcal{O}$ is connected; hence, this choice of $\delta$ ensures that $\{x\in\mathbb{R}^n~|~\mathrm{dist}(x,\mathcal{O})<\delta\}\subset \cup_{i=1}^N B_{\delta_i}(w_i)$. With this choice of $\delta$, any solution $\varphi(t,x,u)$ such that $x\in\mathcal{S}^+$ with $\mathrm{dist}(x,\mathcal{O})<\delta$ and $u\in\mathcal{U}$ with $\|u\|_\infty<\delta$ exists and is unique for $t\in[0,\overline{T}-\xi_i]$, where $\xi_i$ corresponds to the ball $B_{\delta_i}(w_i)$ that includes $x$ (note that $x$ may belong to multiple such balls; picking any one would suffice).  Moreover, the solution will cross $\mathcal{S}$ transversally and in time within $\max_{1\leq i \leq N}(\overline{T}-\xi_i) \leq \overline{T}$. This completes the proof of Lemma~\ref{lem:solution-compare-S+}\emph{(i), (ii)}.

The proof of Lemma~\ref{lem:solution-compare-S+}\emph{(iii)} closely follows the proof of Lemma~\ref{lem:solution-compare-S}\emph{(iii)}, with the difference that, instead of comparing the solution $\varphi(t,x,u)$ with $\varphi(t,\Delta(x^*,0),0)$ as was the case in Lemma~\ref{lem:solution-compare-S}\emph{(iii)}, now we compare $\varphi(t,x,u)$ with $\varphi(t,w(\xi_{\min}),0)=w(t+\xi_{\min})$, where $\xi_{\min}\in[0,T^*]$ is such that\footnote{Such $\xi_{\min}$ exists by Lemma~\ref{lem:inf-min} and depends on $x$ since $\xi_{\min} \in \arg \min_{\xi\in[0,T^]}\|x-w(\xi)\|$; we do not display this dependence here.} $\mathrm{dist}(x,\mathcal{O})=\|x-w(\xi_{\min})\|$. To do this comparison, we use $T^*-\xi_{\min}$ and $\overline{T}-\xi_{\min}$ instead of $T^*$ and $\overline{T}$, and we use Claim 1 and Claim 2 below to generalize estimates related to $\hat{T}_{\rm I}$ and Lemma~\ref{lem:cont-solutions}, respectively, from an open-ball around a point in $\mathbb{R}^n$ to an entire open neighborhood of $\mathcal{O}$ in $\mathcal{S}^+$. In the end, we replace $\|x-w(\xi_{\min})\|$ with $\mathrm{dist}(x,\mathcal{O})$ 

\noindent \emph{Claim 1:} Let $\hat{T}_{\rm I}$ be as in \eqref{eq:time-to-imp-S+}. There exist $L_{\hat{T}}>0$ and $\delta>0$ such that for $x\in\mathcal{S}^+$ with $\mathrm{dist}(x,\mathcal{O})<\delta$ and $u\in\mathcal{U}$ with $\|u\|_\infty<\delta$, the following is satisfied
%eq
\begin{equation}\nonumber
|\hat{T}_{\rm I}(x,u)-\hat{T}_{\rm I}(w(\xi_{\min}),0)| \leq L_{\hat{T}} (\|x-w(\xi_{\min})\| + \|u\|_\infty).
\end{equation}
%eq

\noindent \emph{Claim 2:} There exist $L>0$ and $\delta>0$ such that for any $x\in\mathcal{S}^+$ with $\mathrm{dist}(x,\mathcal{O})<\delta$ and $u\in\mathcal{U}$ with $\|u\|_\infty<\delta$, the following is satisfied for $0\leq t \leq \overline{T}-\xi_{\min}$,
%eq
\begin{equation}\nonumber
\|\varphi(t,x,u)-w(t+\xi_{\min})\| \!\leq\! \mathrm{e}^{L\overline{T}} \|x-w(\xi_{\min})\| +( \mathrm{e}^{L\overline{T}}-1) \|u\|_\infty.
\end{equation}
%eq
%Hence, in proving Lemma~\ref{lem:solution-compare-S+}\emph{(iii)} we follow the steps of proving Lemma~\ref{lem:solution-compare-S}\emph{(iii)}, where, instead of local Lipschitz continuity of $T_{\rm I}$ and Lemma~\ref{lem:cont-solutions}, we use Claim 1 and Claim 2, respectively, to compare $\varphi(t,x,u)$ with $\varphi(t,w(\xi_{\min}),0)=w(t+\xi_{\min})$. We use $T^*-\xi_{\min}$ and $\overline{T}-\xi_{\min}$ instead of $T^*$ and $\overline{T}$. In the end, we replace $\|x-w(\xi_{\min})\|$ with $\mathrm{dist}(x,\mathcal{O})$.
\end{proof}
Finally, we present proofs of Claim 1 and 2.
\begin{proof}[Proof of Claim 1]
For any $\xi\in[0,T^*]$, there exists a $\delta_\xi>0$ such that $\hat{T}_{\rm I}$ is continuously differentiable, and thus locally Lipschitz for $x \in B_{\delta_\xi}(w(\xi))$ and $u \in \mathcal{U}$ with $\|u\|_\infty<\delta_{\xi}$. Shrink (if necessary) $\delta_\xi$ to ensure that the Lipschitz condition is satisfied uniformly for some  $L_{\hat{T},\xi}>0$. Construct an open-cover for $\overline{\mathcal{O}}$ and extract a finite subcover $\cup_{i=1}^N B_{\delta_i/2}(w_i)$ as above, but now use $\delta_i/2$. Let $x\in B_{\delta_i/2}(w_i)$ for some $i$, then $w(\xi_{\min})\in B_{\delta_i}(w_i)$ for the same $i$. This holds because $\|x-w(\xi_{\min})\|\leq \|x-w_i\|<\delta_i/2$, hence using the triangle inequality $\|w(\xi_{\min})-w_i\|\leq \|x-w(\xi_{\min})\| + \|w_i-x\|<\delta_i$. Therefore, for any $x\in \cup_{i=1}^N B_{\delta_i/2}(w_i)$, we can write the Lipschitz condition with the  constant $L_{\hat{T},\xi_i}$ replaced by $L_{\hat{T}}=\max_{1\leq i \leq N} L_{\hat{T},\xi_i}$. Choosing $0<\delta< \min_{1 \leq i\leq N} \delta_i/2$ such that $\{x\in\mathcal{S}^+~|~\mathrm{dist}(x,\mathcal{O})<\delta\} \subset \cup_{i=1}^N B_{\delta_i/2}(w_i)$ and $u\in\mathcal{U}$ with $\|u\|_\infty < \delta$ completes the proof. 
%For any $\xi\in[0,T^*]$, there exists a $\delta_\xi>0$ such that $\hat{T}_{\rm I}$ is continuously differentiable, hence locally Lipschitz, with a Lipschitz constant $L_{\hat{T},\xi}$ for any $x\in B_{\delta_\xi}(w(\xi))$ and $u\in\mathcal{U}$ with $\|u\|_\infty<\delta_{\xi}$. Construct an open-cover for $\overline{\mathcal{O}}$ as above, but now use  $B_{\delta_\xi/2}(w(\xi))$ for all $\xi\in[0,T^*]$. As $\overline{\mathcal{O}}$ is compact, there are $\xi_1,\xi_2,...,\xi_N\in[0,T^*]$ and $\delta_i:=\delta_{\xi_i}$, $w_i:=w(\xi_i)$ such that $\overline{\mathcal{O}}\subset \cup_{i=1}^N B_{\delta_i/2}(w_i)$. Let $x\in B_{\delta_i/2}(w_i)$ for some $i$, then $w(\xi_{\min})\in B_{\delta_i}(w_i)$ for the same $i$. This holds because $\|x-w(\xi_{\min})\|\leq \|x-w_i\|<\delta_i/2$, hence using the triangle inequality $\|w(\xi_{\min})-w_i\|\leq \|x-w(\xi_{\min})\| + \|w_i-x\|<\delta_i$. Therefore, for any $x\in \cup_{i=1}^N B_{\delta_i/2}(w_i)$, we can write the Lipschitz bound with the Lipschitz constant $L_{\hat{T},\xi_i}$ replaced by $L_{\hat{T}}=\max_{1\leq i \leq N} L_{\hat{T},\xi_i}$. Choosing $0<\delta< \min_{1 \leq i\leq N} \delta_i/2$ such that $\{x\in\mathcal{S}^+~|~\mathrm{dist}(x,\mathcal{O})<\delta\} \subset \cup_{i=1}^N B_{\delta_i/2}(w_i)$ and $u\in\mathcal{U}$ with $\|u\|_\infty < \delta$ completes the proof. 
\end{proof}

\begin{proof}[Proof of Claim 2]
The proof is similar to \cite[Theorem~3.5]{khalil2002nonlinear}; hence, we only provide a sketch. Note that $L$ is the Lipschitz constant of $f$ on the compact set $\{(x,\hat{u})\in\mathbb{R}^n\times\mathbb{R}^p ~|~ \|x-w(t)\|\leq \epsilon,\forall t\in[0,T^*],~\mathrm{and}~\|\hat{u}\|\leq\epsilon\}$ where $\hat{u}$ denotes the \emph{value} of $u\in\mathcal{U}$ at a certain time instant and should not be confused with the function itself. The proof follows by the Gronwall-Bellman inequality \cite[Lemma~A.1]{khalil2002nonlinear}, which gives
%between $\varphi(t,x,u)$ and $w(t+\xi_{\min})$
%eq
\begin{align}
& \|\varphi (t,x,u)-w(t+\xi_{\min})\|  \\
& \leq\mathrm{e}^{L(\overline{T}-\xi_{\min})} \|x-w(\xi_{\min})\|  +( \mathrm{e}^{L(\overline{T}-\xi_{\min})}-1) \|u\|_\infty, \nonumber \\
& \leq \mathrm{e}^{L\overline{T}} \|x-w(\xi_{\min})\|  +( \mathrm{e}^{L\overline{T}}-1) \|u\|_\infty. \nonumber
\end{align}
%eq
for $t\in[0,\overline{T}-\xi_{\min}]$. Choosing $\delta=\epsilon/(2\mathrm{e}^{L\overline{T}})$ such that $\mathrm{dist}(x,\mathcal{O})=\|x-w(\xi_{\min})\|<\delta$ and $\|u\|_\infty<\delta$ ensures the solution remains trapped in the compact $\epsilon$-neighborhood of $\overline{\mathcal{O}}$ for $t\in[0,T^*-\xi_{\min}]$ and hence the inequality holds.
\end{proof}

%\section*{Acknowledgment}

\ifCLASSOPTIONcaptionsoff
  \newpage
\fi



%\bibliographystyle{IEEEtran}
%\begin{thebibliography}{1}

%\end{thebibliography}
\bibliographystyle{IEEEtran}
\bibliography{poincare_LISS}

\vspace{-1.6cm}

\begin{IEEEbiography}[{\includegraphics[width=1in,height=1.25in,clip,keepaspectratio]{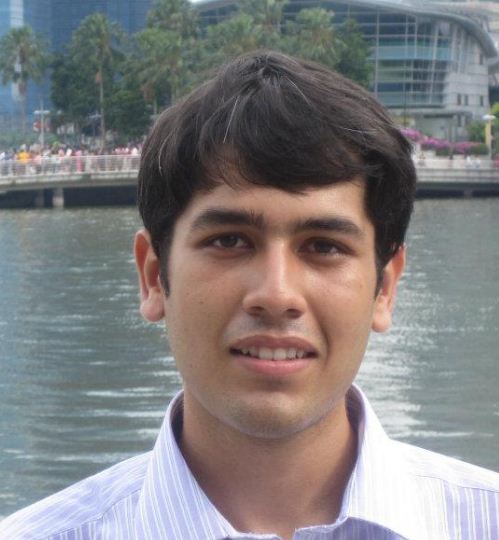}}]{Sushant Veer}
is a PhD Candidate in the Department of Mechanical Engineering at the University of Delaware. He received his B.Tech in Mechanical Engineering from the Indian Institute of Technology Madras (IIT-M) in 2013. His research interests lie in the control of complex dynamical systems with application to dynamically stable robots. 
\end{IEEEbiography}

\vspace{-1.5cm}

\begin{IEEEbiography}[{\includegraphics[width=1in,height=1.25in,clip,keepaspectratio]{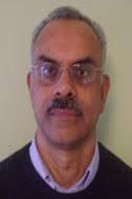}}]{Rakesh}
received an MA in Mathematics from Delhi University in 1978 and a PhD in Mathematics from Cornell University in 1986. He is currently a Professor in the Department of Mathematical Sciences at University of Delaware. His main area of interest is Inverse Problems for Hyperbolic PDEs.
\end{IEEEbiography}

\vspace{-1.5cm}

\begin{IEEEbiography}[{\includegraphics[width=1in,height=1.25in,clip,keepaspectratio]{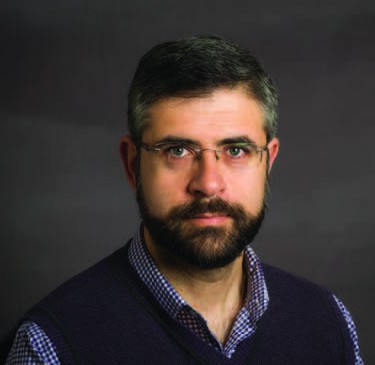}}]{Ioannis Poulakakis}
received his Ph.D. in Electrical Engineering (Systems) from the University of Michigan, MI, in 2009. From 2009 to 2010 he was a post-doctoral researcher with the Department of Mechanical and Aerospace Engineering at Princeton University, NJ. Since September 2010 he has been with the Department of Mechanical Engineering at the University of Delaware, where he is currently an Associate Professor. His research interests lie in the area of dynamics and control with applications to robotic systems, particularly dynamically dexterous legged robots. Dr. Poulakakis received the National Science Foundation CAREER Award in 2014. 
\end{IEEEbiography}



%\vfill


\end{document}